\documentclass[12pt]{article}

\usepackage[colorlinks, linkcolor=blue, citecolor=blue]{hyperref}           
\usepackage{color}
\usepackage{graphicx,subfigure,amsmath,amssymb,amsfonts,bm,epsfig,epsf,url,dsfont}
\usepackage{amsthm,mathrsfs}

\newtheorem{thm}{Theorem}[section]
\newtheorem{lem}[thm]{Lemma}
\newtheorem{assum}[thm]{Assumption}
\newtheorem{proposition}[thm]{Proposition}
\newtheorem{remark}[thm]{Remark}
\newtheorem{model}[thm]{Model}
\newtheorem{corollary}[thm]{Corollary}

\usepackage{tikz}
\usepackage{bbm}
\usepackage[small,bf]{caption}
\usepackage[top=1in,bottom=1in,left=1in,right=1in]{geometry}
\usepackage{fancybox}
\usepackage{hyperref}
\usepackage{algorithm}
\usepackage{algpseudocode}

\usepackage{verbatim}
\usepackage[titletoc,toc]{appendix}
\usepackage{relsize}
\usepackage{svg}

\usepackage{mathtools}

\usepackage{scalerel,stackengine}
\stackMath
\newcommand\reallywidehat[1]{%
\savestack{\tmpbox}{\stretchto{%
  \scaleto{%
    \scalerel*[\widthof{\ensuremath{#1}}]{\kern-.6pt\bigwedge\kern-.6pt}%
    {\rule[-\textheight/2]{1ex}{\textheight}}
  }{\textheight}%
}{0.5ex}}%
\stackon[1pt]{#1}{\tmpbox}%
}

\newcommand*{\rom}[1]{\expandafter\@slowromancap\romannumeral #1@}

\newcommand{\abs}[1]{\left|#1\right|}

\DeclareMathOperator*{\argmax}{arg\,max}
\DeclareMathOperator*{\argmin}{arg\,min}

\newcommand{\p}[1]{\left(#1\right)}
\newcommand{\pp}[1]{\left[#1\right]}
\newcommand{\ppp}[1]{\left\{#1\right\}}
\newcommand{\norm}[1]{\left\|#1\right\|}

\newcommand{\s}[1]{\mathsf{#1}}
\usepackage{xspace}

\numberwithin{equation}{section}
\usepackage{tabularx}
\usepackage{booktabs}
\usepackage{authblk}    
\usepackage{url}

\newcommand{\tamir}[1]{\textcolor{red}{[Tamir: #1]}}

\allowdisplaybreaks

\begin{document}

\title{Structure from Noise: Confirmation Bias in Particle Picking in Structural Biology}

\author[1]{Amnon Balanov\thanks{Corresponding author: \url{amnonba15@gmail.com}}}
\author[1]{Alon Zabatani}
\author[1]{Tamir Bendory}

\affil[1]{\normalsize School of Electrical and Computer Engineering, Tel Aviv University, Tel Aviv 69978, Israel}

\maketitle

\begin{abstract}
The computational pipelines of single-particle cryo-electron microscopy (cryo-EM) and cryo-electron tomography (cryo-ET) include an early particle-picking stage, in which a micrograph or tomogram is scanned to extract candidate particles, typically via template matching or deep-learning-based techniques. The extracted particles are then passed to downstream tasks such as classification and 3D reconstruction.
Although it is well understood empirically that particle picking can be sensitive to the choice of templates or learned priors, a quantitative theory of the bias introduced by this stage has been lacking.

Here, we develop a mathematical framework for analyzing bias in template matching-based detection with concrete applications to cryo-EM and cryo-ET. We study this bias through two downstream tasks: (i) maximum-likelihood estimation of class means in a Gaussian mixture model and (ii) 3D volume reconstruction from the extracted particle stack. We show that when template matching is applied to pure noise, then under broad noise models, the resulting maximum-likelihood estimates converge asymptotically to deterministic, noise-dependent transforms of the user-specified templates, yielding a structure from noise effect. We further characterize how the resulting bias depends on the noise statistics, sample size, dimension, and detection threshold. Finally, controlled experiments using standard cryo-EM software corroborate the theory, demonstrating reproducible structure from noise artifacts in low-SNR data.
\end{abstract}

\newpage
\tableofcontents

\newpage

\section{Introduction}\label{sec:intro}
Cryo-electron microscopy (cryo-EM) and cryo-electron tomography (cryo-ET) have revolutionized structural biology by enabling visualization of macromolecular complexes in near-native states and at increasingly high resolutions. These complementary techniques are now at the heart of the field. Cryo-EM (Figure \ref{fig:1}(a)) reconstructs the 3D structures of isolated macromolecules from 2D projection images acquired at unknown orientations, and has become essential for resolving atomic-level structures of proteins that are difficult to crystallize~\cite{milne2013cryo, bai2015cryo, nogales2016development, renaud2018cryo, yip2020atomic}. Cryo-ET (Figure \ref{fig:2}(a)) captures 3D volumes of macromolecular assemblies in situ by collecting a tilt series of 2D projections at known angles~\cite{chen2019complete, schaffer2019cryo, zhang2019advances, turk2020promise, watson2024advances}. While cryo-ET typically yields lower resolution, it preserves the native spatial context and is uniquely suited for probing the complexity of cellular environments.

As these techniques mature, a critical concern has emerged: the presence of biases, both experimental and algorithmic, that can affect the accuracy of the final 3D reconstructions. This concern is well recognized in the structural biology community, and practical strategies have been developed to mitigate known sources of bias, e.g.,~\cite{sorzano2022bias}. However, a deeper theoretical understanding of how such biases arise and propagate through the downstream reconstruction pipeline remains lacking.

Among the algorithmic sources of bias in single-particle reconstruction are incorrect alignment parameters, inaccurate contrast transfer function (CTF) correction, improper image normalization, erroneous masking in real or Fourier space, and the use of incorrect initial templates~\cite{sorzano2022bias}. This work focuses on the latter and aims to develop a theoretical framework for quantifying its impact. In particular, we seek to establish both theoretical and empirical foundations for analyzing \emph{confirmation bias during the particle-picking stage} in cryo-EM and cryo-ET.

Originally studied in the social sciences, confirmation bias refers to the tendency to interpret information in ways that support pre-existing beliefs or expectations, rather than allowing the data to speak for itself \cite{klayman1995varieties, kassin2013forensic, mynatt1977confirmation}. In the context of structural biology, confirmation bias, viewed as a specific instance of the broader category of model bias, can manifest when researchers, often unconsciously, assume that an unknown structure resembles a previously determined model or a prediction generated by computational tools. 
These assumptions, introduced early in the pipeline, can bias downstream processing and produce reconstructions that resemble prior models, even when unsupported by the data~\cite{scheres2012prevention,sorzano2022bias,balanov2025confirmation}. This raises a fundamental question that guides this work: 
\begin{quote}
    \textit{To what extent does confirmation bias in particle-picking algorithms shape and potentially distort the downstream pipeline, and in particular the final 3D structure in cryo-EM and cryo-ET?}
\end{quote}

\subsection{Examples of confirmation bias in structural biology}
\paragraph{The Einstein from Noise phenomenon.}
To illustrate how methodological assumptions can generate structure from noise, we revisit a striking example widely discussed in the cryo-EM community, the Einstein from Noise phenomenon, which serves as the inspiration for this work. In this thought experiment, researchers analyze a dataset believed to contain noisy, shifted instances of a known image, typically an image of Einstein, but in reality, the data consists entirely of pure noise. Despite the absence of any underlying signal, aligning and averaging the images relative to the template yields a reconstruction that closely resembles the original image~\cite{shatsky2009method, sigworth1998maximum}.

This phenomenon played a key role in a significant scientific controversy over the structure of an HIV molecule~\cite{mao2013molecular, henderson2013avoiding, van2013finding, subramaniam2013structure, mao2013reply} and has since served as a cautionary tale for methodological overfitting. For a detailed statistical treatment, see~\cite{balanov2024einstein, balanov2025expectation}. While often cited in discussions of particle picking using template-matching, it is crucial to note that the Einstein from Noise effect arises from a distinct mechanism: in template-matching particle picking methods, confirmation bias arises from selecting micrograph patches that most closely match the templates (essentially a form of \emph{selection bias}, which is the topic of this work). In contrast, the Einstein from Noise effect emerges from aligning already extracted particles to a reference template. We further elaborate on this distinction in Appendix~\ref{sec:different-statistical-models}.

\paragraph{The risk of confirmation bias.} 
The risk of confirmation bias is especially pronounced in scenarios with low signal-to-noise ratios (SNRs), a defining characteristic of cryo-EM and cryo-ET datasets. These datasets also exhibit significant uncertainty in particle orientations, rendering reconstructions highly sensitive to initialization and heuristic choices \cite{singer2020computational, bendory2020single}.
Figures~\ref{fig:1} and~\ref{fig:2} illustrate the cryo-EM and cryo-ET setups and their respective particle reconstruction workflows, highlighting critical stages where confirmation bias may emerge. The computational workflows for both modalities involve several common stages: particle detection and extraction (i.e., particle picking), alignment and classification, and 3D reconstruction. In cryo-EM, individual 2D particles are picked from micrographs and grouped into 2D classes (Figure~\ref{fig:1}(a)); in cryo-ET, smaller subvolumes, termed subtomograms, are extracted from a reconstructed tomogram and then aligned and averaged to enhance the SNR (Figure~\ref{fig:2}(a)) \cite{zhang2019advances, watson2024advances}. In both pipelines, downstream reconstruction accuracy depends critically on the fidelity of the initial \emph{particle-picking stage---the main focus of this work.}

\paragraph{Confirmation bias in particle-picking.}
This study presents a theoretical and empirical framework to systematically investigate confirmation bias during the crucial \emph{particle-picking} stage of cryo-EM and cryo-ET. Particle picking is the process of extracting 2D particle projections from noisy micrographs (in cryo-EM, Figure~\ref{fig:1}(b--c)) or 3D subtomograms from noisy tomograms (in cryo-ET, Figure~\ref{fig:2}(b)). While there are many different particle-picking algorithms, most fall into three broad categories: template-matching algorithms (e.g.,~\cite{cruz2024high, huang2004application,martinez2025template,bohm2000toward}), template-free methods based on generic shape cues such as blob- or Laplacian-of-Gaussian (LoG)-based picking, and deep learning algorithms (e.g.,~\cite{bepler2019positive,wagner2019sphire}); see Section~\ref{sec:discussion} for further discussion. 

Template-matching methods work by scanning the micrograph or tomogram with a predefined reference structure (or template) and selecting regions that correlate highly with it. Template-free methods, by contrast, avoid explicit structural references and instead rely on coarse geometric or intensity-based cues to identify candidate particles. Deep learning methods train neural networks, often convolutional neural networks (CNNs), to detect particles based on features learned directly from annotated training data. While these approaches often perform well under moderate-to-high SNR, we show that in low-SNR conditions, common to both cryo-EM and cryo-ET, they can introduce systematic structural biases.

\begin{figure}[t!]
    \centering
    \captionsetup{width=1.0\linewidth}
    \includegraphics[width=0.8\linewidth]{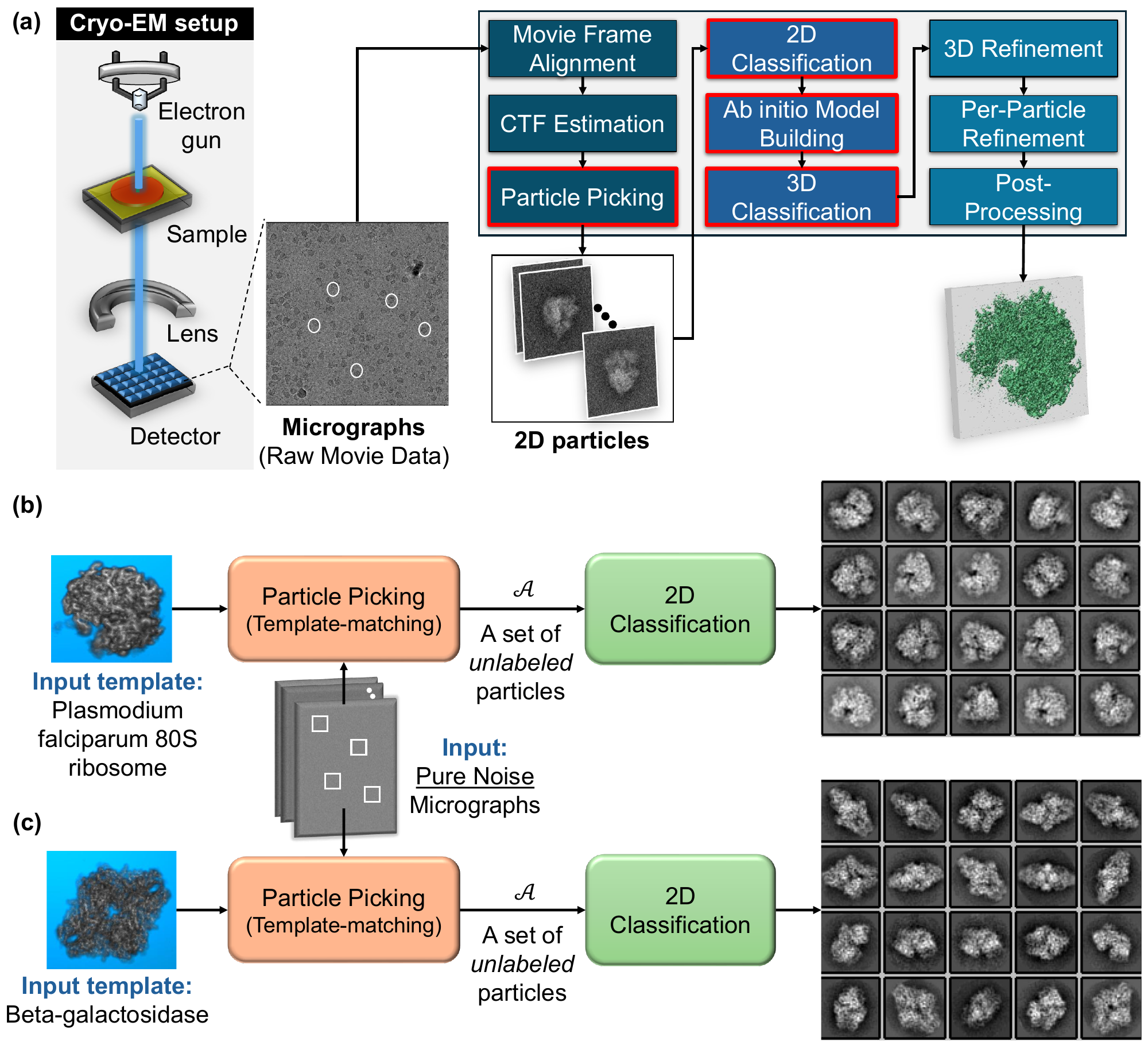}
    \caption{\textbf{The cryo-electron microscope (cryo-EM) computational pipeline and structure from noise in particle picking.}
    \textbf{(a)} Schematic cryo-EM workflow from 2D micrographs to a 3D reconstruction~\cite{bendory2020single}: preprocessing (motion correction and CTF estimation/correction) is followed by particle picking, 2D classification, and 3D ab initio reconstruction and refinement. Red boxes mark stages particularly vulnerable to confirmation bias.
    \textbf{(b--c)} \textbf{Structure from noise under template-matching picking.}
    The template-matching picker (Algorithm~\ref{alg:particlePickerTemplateMatching}) is applied to \emph{pure-noise} micrographs using two different 3D templates:
    (b) Plasmodium falciparum 80S ribosome~\cite{wong2014cryo} and (c) beta-galactosidase~\cite{bartesaghi2014structure}.
    The extracted particles are then processed by 2D classification (RELION VDAM~\cite{kimanius2021new}, without in-plane alignment, yielding 2D class centers that closely resemble the picking templates despite the noise-only input.
    Settings: $L=20$ template projections (and $L=20$ 2D classification classes), patch size $d=48\times48$, $M=2\times10^5$ extracted particles (from 1000 micrographs), and threshold $T=3.8$; ; see Section~\ref{sec:empirical} for an analogous experiment \emph{with in-plane alignment}, together with the full set of parameters.}
    \label{fig:1}
\end{figure}

Figure~\ref{fig:1}(b–c) demonstrates confirmation bias in cryo-EM particle picking using a template-matching algorithm applied to micrographs composed entirely of pure noise (as in the Einstein from Noise experiment). In these methods, predefined particle templates are correlated with input micrographs to identify potential particle locations. After template matching with different initial templates, the selected particles have been processed using RELION's VDAM 2D classification using a reference-free initialization~\cite{kimanius2021new}. Remarkably, despite the absence of any informative signal in the data, the resulting 2D class averages closely resemble projections of the original 3D templates. This behavior is consistent across multiple templates, algorithms, and experimental configurations, and also appears in the cryo-ET pipeline, as shown in Figure~\ref{fig:2}(b).
We refer to this phenomenon, where downstream reconstructions produce biologically plausible structures from pure noise, as \emph{structure from noise}.

\begin{figure}[t!]
    \centering
    \captionsetup{width=1.0\linewidth}
    \includegraphics[width=0.8\linewidth]{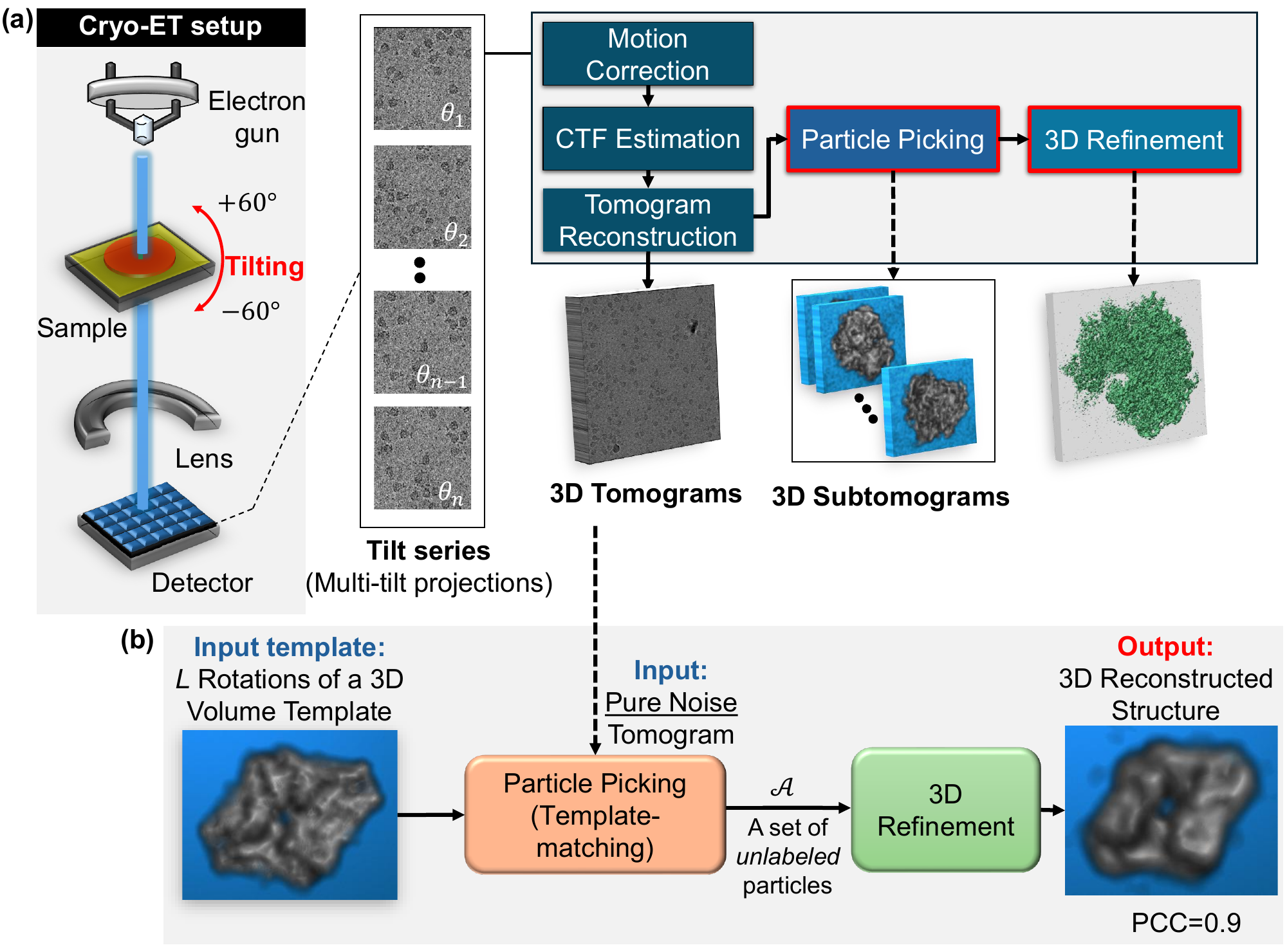}
    \caption{\textbf{Structure from noise in the particle-picking stage of the cryo-electron tomography (cryo-ET) pipeline.}
    \textbf{(a)} In cryo-ET, a tilt series of 2D projections (typically $-60^\circ$ to $+60^\circ$) is acquired and reconstructed into a 3D tomogram; after standard preprocessing (motion correction, tilt-series alignment, and CTF estimation/correction), particle picking extracts subtomograms that are refined by subtomogram averaging to improve SNR~\cite{watson2024advances}.
    \textbf{(b)} \textbf{Confirmation bias under template-matching picking.}
    Template matching (Algorithm~\ref{alg:particlePickerTemplateMatching}) is applied to a \emph{pure-noise} tomogram using $L=50$ rotated templates of a beta-galactosidase volume~\cite{bartesaghi2014structure}.
    From $M=2\times10^4$ extracted subtomograms (size $d=24\times24\times24$, threshold $T=3.5$), expectation-maximization 3D refinement algorithm yields a reconstruction highly correlated with the template volume (PCC $=0.9$).}
    \label{fig:2}
\end{figure}

\subsection{Main results and work structure}  
While prior work~\cite{balanov2025confirmation, balanov2024einstein, balanov2025expectation}  studied confirmation bias arising in downstream  classification, where the bias is induced at the level of the estimation and classification procedure itself, the present work focuses on a different statistical mechanism: confirmation bias introduced already at the particle-picking stage, where template matching selects noise patches that are subsequently passed to downstream reconstruction. Thus, the main novelty of this work is the analysis of particle-picking-induced selection bias and how it affects downstream tasks. Section~\ref{sec:different-statistical-models} and Figure~\ref{fig:7} further compare this mechanism with other forms of confirmation bias arising elsewhere in the cryo-EM/cryo-ET pipeline.

To systematically analyze confirmation bias at the particle-picking stage, we study both classical template matching and deep learning-based particle pickers. To quantify the resulting bias, we introduce a metric that measures the similarity between the extracted particles and the original templates, defined via the outputs of two key downstream tasks: maximum-likelihood estimation of Gaussian mixture model (GMM) class centers (i.e., the statistical model that governs the 2D classification task) and 3D volume reconstruction. We show that, when picking is applied to pure noise (under several models, including white, spherically symmetric, and stationary Gaussian models), the resulting maximum-likelihood estimates converge to deterministic transforms of the picking templates. This reveals a form of model bias: even in the absence of any true particles, the reconstruction is systematically shaped by the user-specified templates, rather than by genuine structure in the data.

We further characterize how the induced bias depends on the noise statistics, the number of extracted particles, the patch or volume dimension, and the detection threshold. On the empirical side, we demonstrate on synthetic cryo-EM and cryo-ET datasets that standard pipelines (RELION VDAM 2D classification and ab initio 3D refinement~\cite{kimanius2021new}) exhibit the predicted structure from noise behavior under template-based picking, and that a widely used neural-network picker, Topaz~\cite{bepler2019positive}, likewise introduces reproducible structure when applied to pure-noise and low-SNR micrographs.

The remainder of the paper is organized as follows. Section~\ref{sec:preliminaries} formalizes the template-based detection problem, introduces the probabilistic noise models, and describes GMM-based class averaging as a canonical reconstruction task. Section~\ref{sec:model} develops the theoretical framework for confirmation bias in template-matching algorithms and analyzes its asymptotic and finite-sample behavior, with specializations to cryo-EM and cryo-ET pipelines given in Section~\ref{sec:applications}. Section~\ref{sec:empirical} presents empirical results using RELION~\cite{kimanius2021new} and Topaz~\cite{bepler2019positive}, and Section~\ref{sec:discussion} discusses implications for routine workflows and strategies for mitigating bias in the particle-picking stage.

\section{Preliminaries and problem formulation} \label{sec:preliminaries}
In this section, we introduce a general framework for analyzing bias in template matching. We begin by formalizing the detection and reconstruction problem in which template-matching methods are used to extract signals from a large noisy observation. We then define the template-matching selection rule and specify the probabilistic noise models used as input to the template-matching algorithm. Finally, we describe the downstream reconstruction procedure that serves to quantify how the selected patches are biased toward the templates. Concrete instances arising in cryo-EM 2D classification and cryo-EM/cryo-ET 3D reconstruction are presented in Section~\ref{sec:applications}. Throughout the text, $\langle a, b \rangle$ represents the inner product between the vectors $a$ and $b$. 

\subsection{The detection and reconstruction process}
We begin with a setting in which a researcher observes a single long measurement $y\in\mathbb{R}^{N}$ (e.g., a time series, a micrograph, or a tomogram). This measurement is believed to contain multiple occurrences (i.e., instances) of a small collection of $L$ unknown signal types $\{s_\ell\}_{\ell=0}^{L-1}\subset\mathbb{R}^{d}$.
Across the observation $y$, these signal types appear repeatedly at unknown locations, where in total there are $Q \gg L$ copies embedded in $y$, and the goal is to estimate the underlying patterns $\{s_\ell\}$; see Figure~\ref{fig:10} for a conceptual overview.

Formally, the researcher postulates the generative model
\begin{align}
    (\text{Postulated statistics})\qquad
    y = \sum_{\ell=0}^{L-1} c_\ell \ast s_\ell + \xi ,
\label{eq:general-MTD-model}
\end{align}
where $\ast$ denotes discrete linear convolution, $\xi$ is additive background noise, and each $c_\ell\in\{0,1\}^{N}$ is a sparse location indicator: $c_\ell(t)=1$ indicates that an instance of type $\ell$ starts at position $t$, and $c_\ell(t)=0$ otherwise. The total number of embedded instances is $Q \triangleq \sum_{\ell=0}^{L-1}\sum_{t=0}^{N-1} c_\ell(t)$, and it is assumed that signal instances are non-overlapping, that is, no index of $y$ is covered by more than one shifted copy of any signal; equivalently, if $c_{\ell_1}(t_1)=c_{\ell_2}(t_2)=1$ with $(\ell_1,t_1)\neq(\ell_2,t_2)$, then $|t_1-t_2|\ge d$.

The goal is to recover the unknown signals $\{s_\ell\}_{\ell=0}^{L-1}$ from the observation $y$. Many imaging and signal-processing pipelines adopt a two-stage strategy (Figure~\ref{fig:10}(c)). First, a detection and extraction stage scans $y$ and produces a collection of patches that are believed to contain signal occurrences. Second, a reconstruction stage aggregates these noisy patches to estimate the underlying signals, for example via a downstream procedure, such as maximum-likelihood estimation of class centers in a GMM fitted to the selected patches.

In practice, the detection stage is often implemented by template matching. The researcher specifies a finite bank of reference templates $\{x_\ell\}_{\ell=0}^{L-1} \subset \mathbb{R}^d$ encoding prior expectations or approximate guesses for the signals $\{s_\ell\}$, and a patch is selected whenever its correlation with at least one template exceeds a prescribed threshold.

In our analysis of confirmation bias, we focus on the null regime in which no true signal
instances are present, i.e., $c_\ell \equiv 0$ for all $\ell$, so that the underlying statistics reduce to
\begin{align}
    (\text{Underlying statistics})\qquad
    y = \xi.
\label{eq:underlying-statistics}
\end{align}
This null hypothesis can also be viewed as the extreme low-SNR limit $\mathrm{SNR} \to 0$ (equivalently $\|s_\ell \| \to 0$), reflecting the highly noise-dominated regime common in many applications, including cryo-EM and cryo-ET.
We then ask: if the detection stage is applied to \eqref{eq:underlying-statistics} rather than to~\eqref{eq:general-MTD-model} and it selects pure-noise patches solely because they correlate with the templates $\{x_\ell\}_{\ell=0}^{L-1}$, what structure does a downstream reconstruction procedure recover from these selected patches? 

\begin{figure}[t!]
    \centering
    \includegraphics[width=0.9\linewidth]{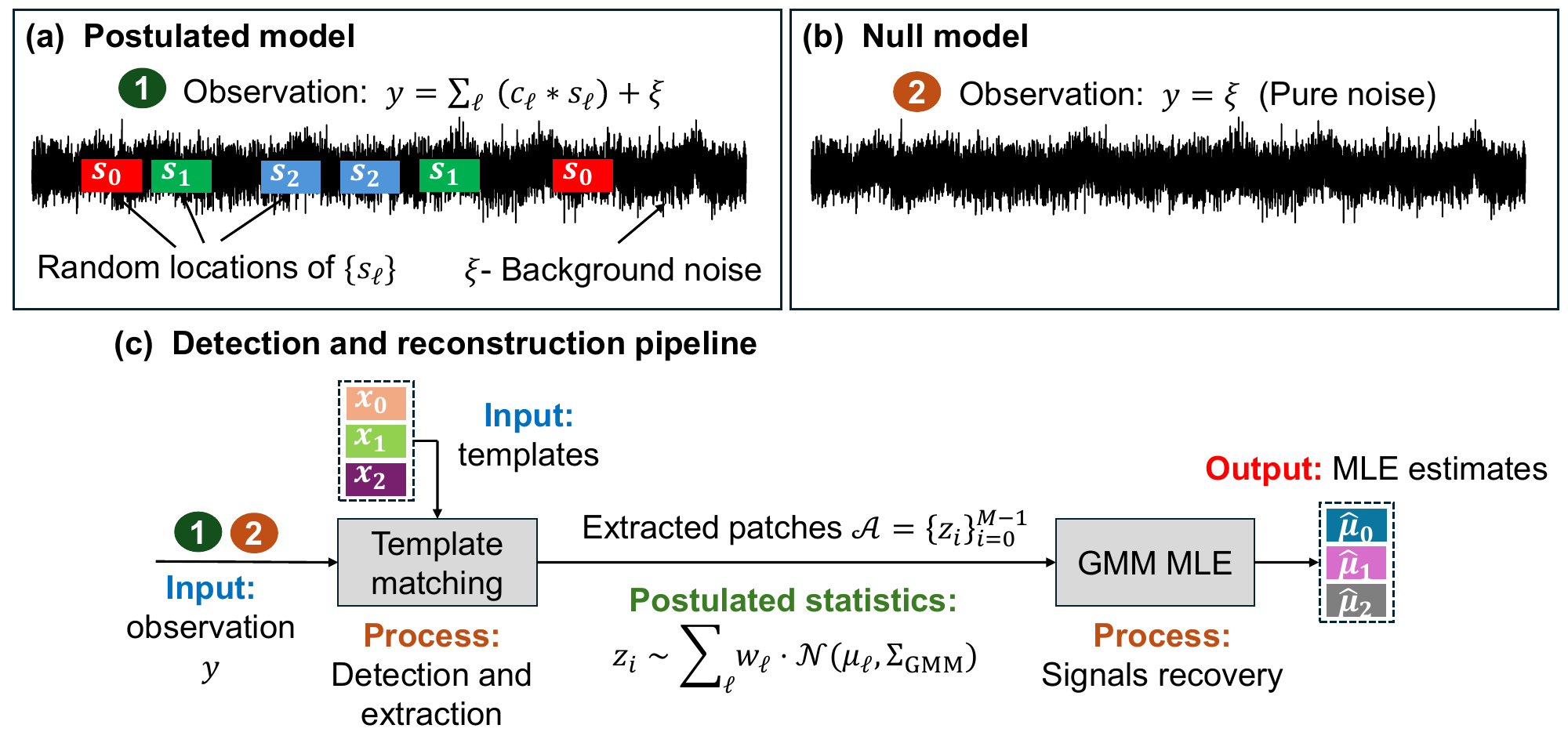}
    \caption{\textbf{Conceptual overview of the model for analyzing confirmation bias in template matching.}
    \textbf{(a)} Postulated signal-plus-noise model~\eqref{eq:general-MTD-model}: the observation contains multiple non-overlapping instances of unknown signals $\{s_\ell\}$ embedded in background noise $\xi$.
    \textbf{(b)} Null model~\eqref{eq:underlying-statistics}: the observation is pure noise.
    \textbf{(c)} Detection-and-reconstruction pipeline: assuming the postulated model, the researcher applies a two-stage procedure. First, template matching selects patches that correlate with a bank of templates $\{x_\ell\}$ (Algorithm~\ref{alg:particlePickerTemplateMatching}). Second, a downstream maximum-likelihood estimator (MLE) fits a Gaussian mixture model (GMM) to the selected patches and outputs component-mean estimates $\{\widehat{\mu}_\ell\}$ (see Section~\ref{subsec:confirmation-bias-analysis-via-downstream}). The central question of this work is whether, and how, this pipeline can yield template-dependent reconstructions even when the true data follow the null model.}
    \label{fig:10}
\end{figure}

\subsection{The template-matching selection}

The researcher generates a collection of candidate patches from the observation $y$, denoted by $\{y_i\}_{i=0}^{N-1} \subset \mathbb{R}^d$. 
Then, the \emph{template-matching selector}, as detailed in Algorithm~\ref{alg:particlePickerTemplateMatching}, operates by cross-correlating each candidate observation $\{y_i\}_{i=0}^{N-1}$ with a finite bank of normalized templates $\{x_\ell\}_{\ell=0}^{L-1} \subset \mathbb{R}^d$.\footnote{We assume $\|x_\ell\|_2 = 1$ for all $\ell$, corresponding to the standard energy normalization used in correlation-based matching. This ensures that correlation scores are directly comparable across templates and observations, avoiding trivial scaling effects.} An observation $y_i$ is selected if its correlation with at least one template exceeds a fixed threshold $T$, that is, if $\max_{0 \le \ell \le L-1} \langle y_i, x_\ell \rangle \ge T$.
The selected observations are collected into the set $\mathcal{A}$, with cardinality $|\mathcal{A}| = M$. Importantly, the elements of $\mathcal{A}$ are \emph{unlabeled}: for a given $y \in \mathcal{A}$, the template (if any) with which it has the highest correlation is not recorded by the selection rule.
In cryo-EM and cryo-ET, patches $\{y_i\}_{i=0}^{N-1}$ represent candidate boxed images (or subtomograms) extracted from a micrograph (or tomogram), i.e., each patch may or may not contain a true particle instance, and the task is to precisely decide which patches correspond to particles rather than background noise. The templates are typically derived from a reference volume; see Section~\ref{sec:applications} for these specific applications.

\begin{algorithm}[] 
  \caption{\texttt{Template-Matching Selector}}
  \label{alg:particlePickerTemplateMatching}
  \textbf{Input:} $L$ templates $\{x_\ell\}_{\ell=0}^{L-1}$, $N$ candidate patches $\{y_i\}_{i=0}^{N-1}$, and a threshold $T$. \\
  \textbf{Output:} Set $\mathcal{A}$ of selected patches.
  \begin{algorithmic}[1]
    \State Initialize $\mathcal{A} \gets \emptyset$
    \For{$i = 0$ to $N-1$}
      \If{$\underset{0 \leq \ell \leq L-1} \max \langle y_i, x_\ell \rangle \geq T$}
        \State Add $y_i$ to $\mathcal{A}$: $\mathcal{A} \gets \mathcal{A} \cup \{y_i\}$
      \EndIf
    \EndFor
    \State \Return $\mathcal{A}$
  \end{algorithmic}
\end{algorithm}

\subsection{Probabilistic models for the input patches}
Our goal is to quantify confirmation bias induced by template-based detection. Following the Einstein from Noise experiment, and to examine this effect in its most extreme form, we analyze the null regime~\eqref{eq:underlying-statistics} in which the data contain only noise and no true signal instances. Motivated by the noise statistics encountered in cryo-EM datasets, we present a hierarchy of probabilistic models for the input patches $\{y_i\}_{i=0}^{N-1}$ presented to the template-matching selector, beginning with white Gaussian noise and progressing to more realistic stationary, strongly mixing noise processes that capture spatial correlations.

\begin{model}[Independent and identically distributed Gaussian noise patches]
\label{model:whiteNoiseIID}
The patches $\{y_i\}_{i=0}^{N-1} \subset \mathbb{R}^d$ are modeled as i.i.d.\ white Gaussian noise, that is, $y_i \sim \mathcal{N}(0,\sigma^2 I_d)$, where $\mathcal{N}(0,\sigma^2 I_d)$ denotes a $d$-dimensional Gaussian distribution with zero mean and covariance matrix $\sigma^2 I_d$.
\end{model}

This idealized model is analytically convenient but not fully realistic: in practice, patches extracted from imaging data are typically dependent (due to spatial overlap), and the background noise is neither perfectly white.

\paragraph{Stationary process with strong-mixing property.}
To capture these effects, we next consider a more general model based on a stationary noise process. Assume the underlying statistics given by~\eqref{eq:underlying-statistics} so that $y = \xi$, and let $\{\xi_t\}_{t\in\mathbb{Z}}$ be a zero-mean, stationary noise process with covariance function
\begin{align}
    \mathbb{E}[\xi_t] = 0, \qquad \mathrm{Cov}(\xi_t,\xi_{t+h}) = K(h),
    \label{eq:stationary-cov-1d}
\end{align}
for every $t,h \in \mathbb{Z}$. We construct candidate patches using a sliding window of length $d$:
\begin{align}
    y_i \triangleq \big(\xi_i, \xi_{i+1}, \dots, \xi_{i+d-1}\big)^\top \in \mathbb{R}^d,
    \label{eq:ys-def-1d}
\end{align}
for every $i \in \mathbb{Z}$. By stationarity of $\{\xi_t\}$, the marginal distribution of $y_i$ does not depend on $i$, and each patch has covariance
\begin{align}
    \nonumber
    \Sigma &\triangleq \mathrm{Cov}(y_i) = \big(\Sigma_{jk}\big)_{1\le j,k\le d}, 
    \\
    \Sigma_{jk} &= \mathrm{Cov}(\xi_{i+j-1},\xi_{i+k-1}) = K(j-k).
    \label{eq:Sigma-from-K-1d}
\end{align}

In many applications, including cryo-EM and cryo-ET, it is natural to assume that distant patches are essentially independent, whereas nearby patches may be correlated but with correlations that decay with distance.
We formalize this by assuming that the process $\{\xi_t\}$ is $\alpha$-mixing (strong mixing) with respect to the index $t$~\cite{bradley2005basic}. Its $\alpha$-mixing coefficients are defined by
\begin{align}
    \alpha(r) \triangleq \sup_{\substack{A \in \sigma(\xi_t:\,t \le n),\\ B \in \sigma(\xi_t:\,t \ge n+r)}} \big|\mathbb{P}(A\cap B) - \mathbb{P}(A)\mathbb{P}(B)\big|,
    \label{eq:alpha-def-1d}
\end{align}
where $\sigma(\xi_t:\,t \le n)$ denotes the $\sigma$-algebra generated by $\{\xi_t:\, t \le n\}$ (events depending only on the past up to time $n$), and $\sigma(\xi_t:\,t \ge n+r)$ is the $\sigma$-algebra generated by $\{\xi_t:\, t \ge n+r\}$ (events depending only on the distant future).
We further assume that $\alpha(r) \to 0$ as $r \to \infty$ together with the standard summability condition~\cite{bradley2005basic}
\begin{align}
    \sum_{r=1}^\infty \alpha(r)^{\delta/(2+\delta)} < \infty
    \qquad \text{for some } \delta>0.
    \label{eq:alpha-mixing-summability-1d}
\end{align}

Informally, this condition provides a mathematical model for short-range dependence: local regions of a signal (e.g., a micrograph or tomogram) are expected to be correlated, but the dependence weakens with distance, so sufficiently separated patches behave almost independently. Assumption~\eqref{eq:alpha-mixing-summability-1d} strengthens this by requiring that the dependence decay fast enough for empirical averages over many selected patches to exhibit the same qualitative concentration behavior as in the independent case. Such assumptions are standard in the analysis of stochastic processes and are natural here for analyzing the spatially correlated noise model with limited range.

We now specify a marginal model for the patches together with a weak dependence structure across them. 

\begin{model}[Spherically symmetric noise patches]
\label{model:isotropicIID}
Let $\{y_i\}_{i=0}^{N-1} \subset \mathbb{R}^d$ be a strictly stationary sequence of random vectors with common marginal distribution $Y \in \mathbb{R}^d$. Assume that $\{y_i\}$ is $\alpha$-mixing with respect to $i$, satisfying \eqref{eq:alpha-def-1d}-\eqref{eq:alpha-mixing-summability-1d} for some $\delta>0$, and that $\mathbb{E}\|Y\|_2^{2+\delta}<\infty$.
Assume that $Y$ admits a probability density function $p_Y$ that is continuous almost everywhere, and that $\mathbb{E}[Y]=0$.
Moreover, assume that $Y$ is rotationally invariant, i.e., $p_Y(y)=\phi(\|y\|_2)$ for some function $\phi$.
In addition, for every unit vector $u\in\mathbb{S}^{d-1}$,
\begin{align}
    \lim_{T\to\infty}\frac{\mathbb{E}\!\left[\langle u,Y\rangle\,\middle|\,\langle u,Y\rangle\ge T\right]}{T}=1,
    \label{eqn:asymptoticTruncatedExpectation}
\end{align}
and for every $\varepsilon>0$,
\begin{align}
    \sup_{u\in\mathbb{S}^{d-1}} \mathbb{P}\!\left( \left\|\frac{Y}{T}-u\right\|_2>\varepsilon \,\Bigg|\, \langle u,Y\rangle\ge T \right)\xrightarrow[T\to\infty]{}0 .
    \label{eq:component-concentration}
\end{align}
\end{model}

Conditions~\eqref{eqn:asymptoticTruncatedExpectation}-\eqref{eq:component-concentration} hold for Gaussian noise and, more generally, for many rotationally invariant light-tailed distributions (e.g., sub-Gaussian and radial log-concave laws). Condition~\eqref{eqn:asymptoticTruncatedExpectation} requires that, conditioning on a large excursion of $\langle u,Y\rangle$ in direction $u$, the conditional mean grows linearly with the threshold $T$. Condition~\eqref{eq:component-concentration} strengthens this by imposing directional concentration: given $\langle u,Y\rangle\ge T$, the normalized vector $Y/T$ concentrates around $u$ uniformly over $u\in\mathbb{S}^{d-1}$. The white-Gaussian model in Model~\ref{model:whiteNoiseIID} is a special case (i.i.d.\ patches with $Y\sim\mathcal{N}(0,\sigma^2 I_d)$). 

We next introduce a more realistic correlated-noise model that is generally not spherically symmetric at the patch level. In this setting, each patch remains Gaussian but with a nontrivial covariance structure that reflects the correlations of the underlying noise process.

\begin{model}[Stationary Gaussian noise patches]
\label{model:stationaryGaussianNoisePatches} 
Suppose that the stationary noise process $\{\xi_t\}_{t\in\mathbb{Z}}$ in~\eqref{eq:stationary-cov-1d} is Gaussian. Then, each patch $y_i$ defined by the sliding window in~\eqref{eq:ys-def-1d}, is a $d$-dimensional Gaussian vector with distribution $y_i \sim \mathcal{N}(0,\Sigma)$, where the covariance matrix $\Sigma$ is given by~\eqref{eq:Sigma-from-K-1d} and is assumed to be non-singular.  
The sequence $\{y_i\}_{i=0}^{N-1}$ is therefore a strictly stationary Gaussian vector process that is $\alpha$-mixing with coefficients $\alpha(r)$ satisfying the summability condition~\eqref{eq:alpha-mixing-summability-1d}.
\end{model}

\subsection{Confirmation bias analysis via downstream reconstruction} \label{subsec:confirmation-bias-analysis-via-downstream}
Recall that, after template-based extraction, the selected patches are unlabeled. In the downstream analysis, the researcher aims to estimate and recover the underlying signals $\{s_\ell\}$ in \eqref{eq:general-MTD-model} from these extracted patches.
We adopt this same downstream mechanism as a vehicle for quantifying confirmation bias in template matching: rather than applying it to patches extracted under the postulated model \eqref{eq:general-MTD-model}, we apply it to patches extracted under the null statistics \eqref{eq:underlying-statistics} (pure noise). We then analyze how the resulting reconstructions relate to the templates used in the detection stage.

A central downstream reconstruction process is Gaussian-mixture-based class averaging, which arises, for example, in cryo-EM 2D classification. Let $\mathcal{A}=\{z_i\}_{i=0}^{M-1}\subset\mathbb{R}^d$ denote the collection of patches selected by the template-matching stage. In the downstream analysis, these patches are commonly modeled as independent samples from a GMM \cite{reynolds2009gaussian}:
\begin{align}
    z_0, z_1, \ldots, z_{M-1} \stackrel{\text{i.i.d.}}{\sim} \sum_{\ell=0}^{K-1} w_\ell \cdot \mathcal{N}(\mu_\ell, \Sigma_{\text{GMM}}),
    \label{eqn:GMMmodelMain}
\end{align}
where $\mu_\ell \in \mathbb{R}^d$ is the mean of component $\ell$, $w_\ell$ is its mixing weight (with $w_\ell \geq w_{\min} > 0$ and $\sum_{\ell=0}^{K-1} w_\ell = 1$), and $\Sigma_{\text{GMM}}$ is a shared non-singular covariance matrix.
In an ideal, unbiased setting under the postulated signal-plus-noise model \eqref{eq:general-MTD-model}, each mixture mean $\mu_\ell$ is interpreted as a class average intended to recover one of the underlying signal types $s_\ell$. Accordingly, we set the number of mixture components equal to the number of signal types, $K \triangleq L$.
However, in the null setting~\eqref{eq:underlying-statistics}, there are no true signal occurrences present, but solely noise. Consequently, any nontrivial resemblance between the reconstructed centers $\{\mu_\ell\}$ and the templates $\{x_\ell \}$ cannot be attributed to genuine signal and must arise from selection bias induced by the template-matching stage.

A standard approach to estimate the means $\{\mu_\ell\}_{\ell=0}^{L-1}$ from the selected patches is maximum-likelihood estimation \cite{wald1949note,scheres2012relion}. The log-likelihood of the GMM in \eqref{eqn:GMMmodelMain} is
\begin{align}
    \mathcal{L}\big(\{\mu_\ell\}_{\ell = 0}^{L-1}; \{z_i\}_{i = 0}^{M-1}\big) = \sum_{i=0}^{M-1} \log \Bigg(\sum_{\ell=0}^{L-1} w_\ell \, \mathcal{N}\big(z_i ; \mu_\ell, \Sigma_{\text{GMM}} \big)\Bigg),    \label{eqn:maximumLikelihoodTermMain-1}
\end{align}
where $\mathcal{N}(z_i ; \mu_\ell, \Sigma_{\text{GMM}})$ denotes the multivariate Gaussian density evaluated at $z_i$.
The GMM maximum-likelihood estimators of the component means are defined by
\begin{align}
    \{\widehat{\mu}_\ell\}_{\ell=0}^{L-1} \triangleq \argmax_{\{\mu_\ell\}_{\ell=0}^{L-1}} \mathcal{L}\big(\{\mu_\ell\}_{\ell = 0}^{L-1}; \{z_i\}_{i = 0}^{M-1}\big). \label{eqn:maximumOfLogLikelihoodMain}
\end{align}

Within our framework when no signal occurrences appear in the data, the central question is: to what extent do the estimated centers $\{\widehat{\mu}_\ell\}$ resemble the templates $\{x_\ell\}$ used in the template-matching selector? Under an unbiased procedure and in the null regime, one would expect the class averages to bear no systematic correlation with the templates. Any persistent alignment between $\{\widehat{\mu}_\ell\}$ and $\{x_\ell\}$ therefore serves as a quantitative manifestation of confirmation bias.

We emphasize that the GMM in \eqref{eqn:GMMmodelMain} provides a \emph{postulated} model for the selected patches. It does not describe the true distribution of the extracted particles by template matching (Algorithm~\ref{alg:particlePickerTemplateMatching}), which in general is highly non-Gaussian and strongly biased by the selection rule. In other words, the model is statistically \emph{misspecified}, and this mismatch must be taken into account when interpreting the resulting class averages. A more detailed discussion of this misspecification is provided in Appendix~\ref{sec:misspecifiedModels}.

\section{Theoretical results}\label{sec:model}

In this section we present our main theoretical results on confirmation bias in template-based selection. Section~\ref{sec:populationAnalysisBias} analyzes the population-level behavior of bias in the asymptotic regime $N,T \to \infty$, while Section~\ref{seubsec:finite-sample-effects} studies the effects of finite-sample for $N < \infty$. For completeness, Appendix~\ref{subsec:local-maxima} introduces a more realistic template-matching mechanism based on local maxima and shows that this refinement does not change the leading-order bias.

\subsection{Population analysis of the template-matching selection bias}\label{sec:populationAnalysisBias}

We begin by quantifying the bias induced by template matching in the setting where the candidate patches are white Gaussian (or, more generally, spherically symmetric) noise. Theorem~\ref{thm:classesCentersVersusTemplatesInformal} shows that, under Model~\ref{model:whiteNoiseIID} or Model~\ref{model:isotropicIID}, the maximum-likelihood centroids of the postulated GMM obtained from the selected patches converge, after rescaling by the threshold $T$, to the templates themselves. In the following Theorem~\ref{thm:classesCentersVersusTemplatesInformalStationary} (proved in Appendix~\ref{sec:proofOfMainThms}) we extend this result to stationary Gaussian noise with spatial correlations (Model~\ref{model:stationaryGaussianNoisePatches}), where the limiting bias is given by an anisotropic rescaling of each template determined by the noise covariance.

\begin{thm}[Bias of GMM centers with spherically symmetric noise]
\label{thm:classesCentersVersusTemplatesInformal}  
Let $\{y_i\}_{i=0}^{N-1}$ be the candidate patches drawn according to  Model~\ref{model:whiteNoiseIID} or Model~\ref{model:isotropicIID}.  
Fix $L \geq 2$, and let $\{x_\ell\}_{\ell=0}^{L-1} \subset \mathbb{R}^d$ denote the normalized input templates, with $x_{\ell_1} \neq x_{\ell_2}$ for every $\ell_1 \neq \ell_2$.
Let $T \in \mathbb{R}^{+}$ be the template-matching threshold and let $\mathcal{A}$ be the set of extracted patches produced by Algorithm~\ref{alg:particlePickerTemplateMatching}.
In the subsequent reconstruction step, the extracted patches $\mathcal{A}$ are modeled by a GMM with $L$ components as in~\eqref{eqn:GMMmodelMain}, and let $\{\widehat{\mu}_\ell\}_{\ell=0}^{L-1}$ denote the corresponding GMM maximum-likelihood estimators of the component means, as defined in~\eqref{eqn:maximumOfLogLikelihoodMain}.
Then there exists a permutation $\pi:\{0,\dots,L-1\} \to \{0,\dots,L-1\}$ such that, for every $0 \leq \ell \leq L-1$,
\begin{align}
    \lim_{T \to \infty} \lim_{N \to \infty} \frac{\widehat{\mu}_{\pi(\ell)}}{T} = x_\ell,    \label{eqn:nonVanishingEstimator2D-iid}
\end{align}
where the convergence holds in probability.
\end{thm}

Theorem~\ref{thm:classesCentersVersusTemplatesInformal} reveals a pronounced bias introduced by template-based selection: in the spherically symmetric noise setting, the estimated class centroids retain full correlation (Pearson cross-correlation of $1$) with the templates used in the selection stage, even when the input data are pure noise. This is in stark contrast to the unbiased expectation that averaging pure noise yields zero as the number of samples grows. The appearance of the permutation $\pi$ is due to the intrinsic label-switching symmetry of GMMs: the likelihood is invariant under relabeling of the mixture components, so the maximum-likelihood estimator identifies the set of centers only up to a permutation of their indices.

\begin{thm}[Bias of GMM centers with stationary Gaussian noise]
\label{thm:classesCentersVersusTemplatesInformalStationary}  
Assume the same setting as in Theorem~\ref{thm:classesCentersVersusTemplatesInformal}, except that the input candidate patches $\{y_i \}_{i=0}^{N-1}$ are drawn according to Model~\ref{model:stationaryGaussianNoisePatches} with a covariance matrix $\Sigma$ as in~\eqref{eq:Sigma-from-K-1d}, and suppose that the $\alpha$-mixing condition \eqref{eq:alpha-mixing-summability-1d} holds. Then there exists a permutation $\pi:\{0,\dots,L-1\} \to \{0,\dots,L-1\}$ such that, for every $0 \leq \ell \leq L-1$,
\begin{align}
    \lim_{T \to \infty} \lim_{N \to \infty} \frac{\widehat{\mu}_{\pi(\ell)}}{T} = \frac{\Sigma x_\ell}{x_\ell^\top \Sigma x_\ell},
    \label{eqn:nonVanishingEstimator2D-stationary}
\end{align}
where the convergence holds in probability.
\end{thm}

Theorem~\ref{thm:classesCentersVersusTemplatesInformalStationary} shows that the confirmation bias persists even when the noise is non-spherically symmetric. 
While in the spherically symmetric setting of Theorem~\ref{thm:classesCentersVersusTemplatesInformal} the estimated centroids align with the raw templates $x_\ell$ (up to a global scaling by $T$), under stationary Gaussian noise the limiting direction of the estimated centroids is given by $\frac{\Sigma x_\ell}{x_\ell^\top \Sigma x_\ell}$, which is an anisotropic rescaling of the template $x_\ell$ determined by the noise covariance $\Sigma$. 
Thus, even in the presence of realistic spatial correlations, Gaussian-mixture class averaging of patches selected by template matching remains fully biased toward the templates used in the selection stage, rather than vanishing as one might expect under a noise-only model.

\paragraph{The bias mechanism.} 
We now explain the mechanism underlying the observed bias. The key point is that template-matching is a selection operation that changes the data distribution.
Specifically, the extracted patches do not follow the original noise model; rather, they follow the conditional law of $Y$ given the event $\max_{\ell}\langle Y, x_\ell\rangle \ge T$, so the threshold $T$ induces a truncated (selection-biased) distribution.
In particular, when the noise is spherically symmetric, this conditioning breaks rotational invariance and produces a nonzero conditional mean that aligns with the template $x_\ell$.
More precisely, let $\ell^\star=\arg\max_{\ell}\langle Y,x_\ell\rangle$. Then, for patches selected with $\ell^\star=\ell$, the conditional mean is aligned with the template $x_\ell$:
\begin{align}
    \mathbb{E}\!\left[\,Y \,\middle|\, \ell^\star=\ell,\ \max_{j}\langle Y,x_j\rangle \ge T\,\right] \ \parallel\ x_\ell.
\end{align}
For correlated Gaussian noise $Y \sim \mathcal{N}(0,\Sigma)$, the same mechanism holds but the conditional mean aligns with the direction $\Sigma x_\ell$.

However, this alignment does not by itself imply that the maximum-likelihood estimators of the GMM means recover these conditional centroids.
Indeed, the distribution of the extracted patches is generally non-Gaussian, so the postulated GMM model~\eqref{eqn:GMMmodelMain} is misspecified; see Appendix~\ref{subsec:GMM-likelhood-template-selected} for further discussion and analysis.
In particular, the behavior of one center per template emerges only when $T$ is sufficiently large: as $T$ increases, the true mixture model components induced by the template-matching algorithm become effectively well-separated, so the maximum-likelihood estimation of the GMM centers places one center near each such biased mean. In contrast, at moderate thresholds, these components can overlap substantially, and the maximum-likelihood estimator need not recover these conditional means. 
Figure~\ref{fig:3}(a) illustrates this phenomenon: at low threshold values, the centroids do not align with the original templates because the extracted patches are not reliably associated with their corresponding classes. However, as the threshold increases, the estimated centroids progressively align more closely with the original templates. This reveals a trade-off in selecting the threshold during the selection stage, which we discuss further in Section~\ref{sec:discussion}.

\begin{figure}[t!]
    \centering
    \includegraphics[width=0.85\linewidth]{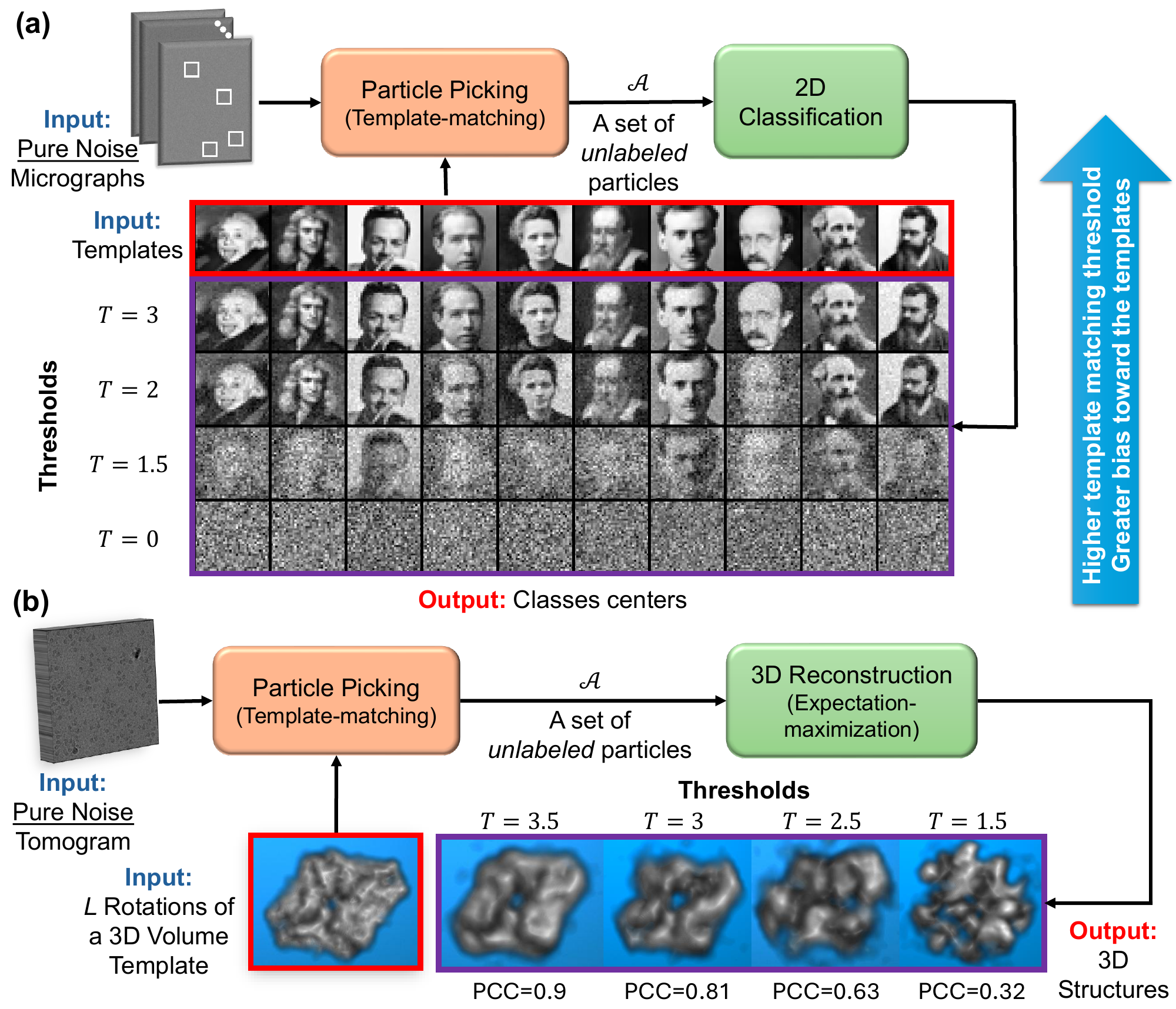}
    \caption{\textbf{Effect of template-matching thresholds on confirmation bias.}
    \textbf{(a)} \textbf{GMM means estimation ("2D classification").}
    Template matching is applied to an observation $y$ from~\eqref{eq:underlying-statistics} under noise Model~\ref{model:whiteNoiseIID} using a bank of $L=10$ templates of prominent physicists.
    The resulting (unlabeled) patches are then classified using a standard expectation-maximization algorithm with $L$ components; as the threshold $T$ increases, the estimated centers become increasingly aligned with the templates.
    Parameters: $M=10^5$ selected patches, image size $d=50\times50$.
    \textbf{(b)} \textbf{3D reconstruction (cryo-ET).}
    Template matching is applied to a pure-noise tomogram from Model~\ref{model:whiteNoiseIID} using $L=50$ rotated templates of beta-galactosidase~\cite{bartesaghi2014structure}.
    Subtomogram extraction followed by 3D refinement yields reconstructions whose similarity to the template volume increases with $T$ (PCC shown).
    Parameters: volume size $d=24\times24\times24$, $M=2\times10^4$ selected subtomograms, and noise variance $\sigma^2=1$.}
    \label{fig:3}
\end{figure}

\subsection{Finite-sample analysis of template-matching selection bias} \label{seubsec:finite-sample-effects}

The analytical results derived above are in the asymptotic regime where the number of selected patches satisfies $M \to \infty$. In practice, however, applications involve a large but finite number of patches, so it is important to understand how finite-sample effects influence the observed confirmation bias. To this end, we study the mean squared error $\mathbb{E}_{\mathcal{A}} [\|\widehat{\mu}_\ell(\mathcal{A}) - T v_\ell\|_2^2]$, which quantifies the discrepancy between the estimated class centroid $\widehat{\mu}_\ell$ and the asymptotic bias direction
\begin{align}
    v_\ell \triangleq \frac{\Sigma x_\ell}{x_\ell^\top \Sigma x_\ell}, \label{eqn:v_ell_def_2}
\end{align}
as in the right-hand side of~\eqref{eqn:nonVanishingEstimator2D-stationary}. The following proposition, proved in Appendix~\ref{sec:app_sampleComplexity}, characterizes the resulting finite-sample impact on the confirmation bias.

\begin{proposition}[Finite-sample effect on confirmation bias]
\label{prop:finite-sample-MSE}
Consider the setting of Theorem~\ref{thm:classesCentersVersusTemplatesInformalStationary}.  
Let $v_\ell$ be as in~\eqref{eqn:v_ell_def_2} and $\widehat{\mu}_\ell(\mathcal{A})$ denote the GMM maximum-likelihood estimator of the $\ell$-th class mean obtained from the extracted set of patches $\mathcal{A}$ in~\eqref{eqn:maximumOfLogLikelihoodMain}, with $|\mathcal{A}| = M$.
Then, there exist constants $M_0, T_0, C_1, C_2 > 0$, independent of $d$, $M$, and $T$, such that for every $M \geq M_0$ and $T \geq T_0$,
\begin{align}
    \mathbb{E}_{\mathcal{A}}\Big[\big\|\widehat{\mu}_\ell(\mathcal{A}) - T v_\ell\big\|_2^2\Big] \le C_1\,\frac{d}{M} + \frac{C_2}{T^2},
    \label{eq:finite-MSE-vs-Tx}
\end{align}
where the expectation is taken over the set of extracted particles $\mathcal{A}$.
\end{proposition}

The first term in the bound~\eqref{eq:finite-MSE-vs-Tx} has the classical parametric form of $d$-dimensional component means estimation under a well-separated GMM. The constant $C_1$ also captures the effect of temporal dependence between patches. Under the summability condition~\eqref{eq:alpha-mixing-summability-1d}, one can bound $C_1$ in terms of $1 + \sum_{r=1}^\infty \alpha(r)^{\delta/(2+\delta)}$. When the patches are more strongly correlated, the mixing coefficients $\alpha(r)$ decay more slowly, the sum increases, and $C_1$ increases accordingly. In the i.i.d.\ case, $\alpha(r)=0$ for all $r\ge 1$, the sum vanishes and $C_1$ reduces to the usual i.i.d.\ parametric constant. The second term, $C_2/T^2$, quantifies the residual bias between the population centers and the high-threshold limit $T v_\ell$.

An immediate consequence is that, for a fixed number of extracted patches $M$, smaller patches (i.e. lower dimension $d$) concentrate more rapidly around the biased templates and are therefore more susceptible to confirmation bias than larger patches. Further details and the proof of Proposition~\ref{prop:finite-sample-MSE} are given in Appendix~\ref{sec:app_sampleComplexity}.

\section{Applications to cryo-EM and cryo-ET}
\label{sec:applications}

In this section we specialize the framework and results of Section~\ref{sec:model} to standard single-particle cryo-EM and cryo-ET workflows. 

\subsection{2D classification}

In single-particle cryo-EM, the signals $\{s_\ell\}$ in~\eqref{eq:general-MTD-model} are not arbitrary: they arise as tomographic projections of a common 3D volume. More precisely, there exists an underlying volume $V : \mathbb{R}^3 \to \mathbb{R}$ such that $s_\ell = \Pi\big(R_\ell \cdot V\big)$, for $\ell = 0,\dots,L-1$, where $\Pi$ denotes the tomographic projection operator, $R_\ell \in \mathsf{SO}(3)$ is a 3D rotation acting via $(R_\ell \cdot V)(x) \triangleq V(R_\ell^{-1} x)$, and $L$ is the number of discretized viewing directions.
In practice, the set of admissible orientations is continuous, but for the current analysis we work with a finite discretization $\{R_\ell\}_{\ell=0}^{L-1} \subset \mathsf{SO}(3)$.
The templates used in the template-matching selector $\{x_\ell\}$, are generated in the same way from a template volume $V_{\text{template}} : \mathbb{R}^3 \to \mathbb{R}$:
\begin{align}
    x_\ell = \Pi\big(R_\ell \cdot V_{\text{template}}\big), \qquad \ell = 0,\dots,L-1.
    \label{eq:cryoEM-templates}
\end{align}

The 2D classification step in cryo-EM is commonly modeled using a GMM, exactly of the form introduced in~\eqref{eqn:GMMmodelMain}. Consequently, the results obtained for GMM-based class averaging developed in Section~\ref{sec:model} apply directly to the cryo-EM 2D classification pipeline.
We note, however, that the simplified GMM model~\eqref{eqn:GMMmodelMain} omits two standard aspects of practical cryo-EM workflows. First, in-plane alignment is typically performed, whereby each particle is rotated and translated in the image plane to best match a class average~\cite{kimanius2021new}. Second, the number of downstream 2D classes need not equal the number of picking templates. In the theory, we exclude in-plane alignment and take $K=L$ to isolate the particle-picking-induced bias in its clearest analytically tractable form. These simplifications are not essential to the qualitative mechanism: when $K\neq L$, template-based selection still induces a biased conditional distribution on the extracted patches, but the downstream GMM no longer yields a one-to-one correspondence between templates and fitted classes; if $K<L$, several template-induced modes may merge, whereas if $K>L$, a single such mode may split into multiple fitted components. As shown empirically in Section~\ref{sec:empirical}, the core confirmation-bias phenomena persist in these more realistic settings.

\subsection{3D volume reconstruction}
The theoretical framework developed above for 2D classification extends naturally to the standard 3D cryo-EM and cryo-ET reconstruction pipeline. 
In single-particle cryo-EM, particle picking is performed on 2D micrographs, whereas in cryo-ET it operates directly on 3D tomograms, yielding volumetric subtomograms. Both modalities can be described within a unified forward model: each extracted, discretized observation $z_i \in \mathbb{R}^d$ (a 2D particle image in cryo-EM or a 3D subtomogram in cryo-ET) is treated as a noisy, randomly rotated observation of an underlying 3D volume $V \in \mathbb{R}^{d_{\mathrm{vol}}}$,
\begin{align}
    z_i = \Pi\bigl(R_i \cdot V\bigr) + \epsilon_i, \qquad i = 0,\dots,M-1,
    \label{eq:unified-postulated-model-main}
\end{align}
where $\{R_i\}_{i=0}^{M-1} \subset \mathsf{SO}(3)$ are random rotations and $\{\epsilon_i\}_{i=0}^{M-1} \overset{\text{i.i.d.}}{\sim} \mathcal{N}(0,\Sigma_V)$ are additive Gaussian noise terms with covariance $\Sigma_V \succ 0$. In this unified view, cryo-ET corresponds to $\Pi = I$ (volumetric templates and subtomograms), while cryo-EM corresponds to a tomographic projection operator $\Pi$ (2D projection templates and particle images). The goal of the reconstruction step is to estimate $V$ in~\eqref{eq:unified-postulated-model-main} from the extracted observations $\{z_i\}_{i=0}^{M-1}$.

Under the postulated model~\eqref{eq:unified-postulated-model-main}, integrating over the latent rotations yields the log-likelihood
\begin{align}
    \mathcal{L}\big(V; \{z_i\}_{i=0}^{M-1}\big) = \sum_{i=0}^{M-1} \log \Bigg( \int_{R \in \mathsf{SO}(3)} \mathcal{N}\big( z_i ; \Pi(R \cdot V), \Sigma_V \big) \, d\rho(R)\Bigg),
    \label{eq:unified-likelihood-main}
\end{align}
where $\rho$ is a prior distribution on rotations (often taken to be uniform on $\mathsf{SO}(3)$). The corresponding maximum-likelihood estimator of the volume is
\begin{align}
    \widehat{V} \triangleq \argmax_{V} \mathcal{L}\big(V; \{z_i\}_{i=0}^{M-1}\big).
    \label{eq:unified-MLE-main}
\end{align}

For analytical tractability, we discretize $\mathsf{SO}(3)$ to a finite set $\{R_\ell\}_{\ell=0}^{L-1}$ and assume that the same set is used both in the template-matching stage and in the subsequent reconstruction. The template bank fed to the particle picker is then $\{x_\ell\}_{\ell=0}^{L-1}$ as defined in~\eqref{eq:cryoEM-templates}, and Algorithm~\ref{alg:particlePickerTemplateMatching} is applied with these templates and threshold $T$. In addition, we make the following assumption:

\begin{assum}\label{assump:Pi-identifiability}
Under the unified model~\eqref{eq:unified-postulated-model-main}, the map $V \mapsto \{\Pi(R_\ell\cdot V)\}_{\ell=0}^{L-1}$ determines $V$ up to a global rotation. That is, if
\begin{align}
    \bigl\{\Pi(R_\ell\cdot V_1)\bigr\}_{\ell=0}^{L-1} = \bigl\{\Pi(R_\ell\cdot V_2)\bigr\}_{\ell=0}^{L-1},
\end{align}
then there exists $R\in\mathsf{SO}(3)$ such that $V_2 = R\cdot V_1$.
\end{assum}
This property holds automatically in the cryo-ET case $\Pi = I$ if there are no internal symmetries to the molecule, and it coincides with standard identifiability assumptions in cryo-EM. The following corollary (proved in Appendix~\ref{sec:proofOfCorollary3Dreconstruction}) shows that the same confirmation bias phenomenon persists at the level of 3D reconstruction.

\begin{corollary}
[Bias of 3D volume estimate with spherically symmetric noise]
\label{thm:3DstructureInformal}
Let $\{y_i\}_{i=0}^{N-1}$ be the candidate patches drawn according to  Model~\ref{model:whiteNoiseIID} or Model~\ref{model:isotropicIID}. Fix $L \geq 2$, and let $V_{\mathrm{template}} \in \mathbb{R}^{d_{\text{vol}}}$ denote a fixed 3D template volume. Let $\{R_\ell\}_{\ell=0}^{L-1} \subset \mathsf{SO}(3)$ be a collection of distinct rotations, and define the 3D template bank $x_\ell \triangleq \Pi (R_\ell \cdot V_{\mathrm{template}})$ according to~\eqref{eq:cryoEM-templates}, with $x_{\ell_1} \neq x_{\ell_2}$ for every $\ell_1 \neq \ell_2$. Run Algorithm~\ref{alg:particlePickerTemplateMatching} with templates $\{x_\ell\}_{\ell=0}^{L-1}$ and threshold $T \in \mathbb{R}^{+}$, and let
$\mathcal{A} \subset \{y_i\}_{i=0}^{N-1}$ be the resulting set of extracted particles. In the subsequent 3D reconstruction step, model the extracted particles $\mathcal{A}$ via the likelihood~\eqref{eq:unified-likelihood-main}, and let $\widehat{V}$ denote the corresponding maximum-likelihood estimator as in~\eqref{eq:unified-MLE-main}, using the same  rotations $\{R_\ell\}_{\ell=0}^{L-1}$. Assume that Assumption~\ref{assump:Pi-identifiability} is satisfied for this collection of discretized rotations.
Then there exists a rotation $R \in \mathsf{SO}(3)$ such that
\begin{align}
    \lim_{T \to \infty} \lim_{N \to \infty} \frac{\widehat{V}}{T} = R \cdot V_{\mathrm{template}},
    \label{eq:cryoET-white-limit}
\end{align}
where the convergence holds in probability.
\end{corollary}

In words, under the white or spherically symmetric noise model, when both the threshold $T$ and the number of candidate observations $N$ are large, the reconstructed structure satisfies $\widehat{V} \approx T \,(R \cdot V_{\mathrm{template}})$ for some global rotation $R \in \mathsf{SO}(3)$. Corollary~\ref{thm:3DstructureInformal} relies on two key ingredients: (i) \emph{rotation-grid alignment}: the same discrete set $\{R_\ell\}$ must be used to generate the template bank~\eqref{eq:cryoEM-templates} and in the reconstruction likelihood~\eqref{eq:unified-likelihood-main}; and (ii) \emph{identifiability} of $V$ from the collection $\{\Pi(R_\ell \cdot V)\}_{\ell=0}^{L-1}$ up to a global rotation, as stated in Assumption~\ref{assump:Pi-identifiability}. When the rotation grids are not aligned, the maximum-likelihood estimator $\widehat{V}$ no longer reproduces the template exactly, but numerical experiments (see Figure~\ref{fig:3}(b)) indicate that $\widehat{V}$ remains highly correlated with $V_{\mathrm{template}}$.

\section{Empirical results}
\label{sec:empirical}
This section presents empirical results that corroborate and extend the theoretical findings from Sections~\ref{sec:model}--\ref{sec:applications}. We examine how deviations from the simplified model, such as the presence of true particles (i.e., rather than purely noise micrographs), and the use of in-plane alignment, affect the outcomes of 2D classification and 3D reconstruction. Additionally, we evaluate the behavior of modern particle-picking approaches, in particular deep learning-based methods, to assess their susceptibility to confirmation bias in realistic cryo-EM workflows.

Importantly, the 2D classification and 3D reconstruction methods employed in our analysis are reference-free, i.e., they rely on random initialization rather than incorporating prior templates. This design allows us to isolate the confirmation bias introduced during picking. However, we note that the non-convex nature of the downstream reconstruction tasks may also contribute to residual bias. 

\subsection{Particle picking based on template-matching}

In the previous sections, we introduced a simplified model for particle picking based on template matching, assuming that the input micrograph or tomogram consists solely of background noise. 
While this model offers important insights into the confirmation bias introduced by template-matching particle-picking, it does not fully reflect the complexity of modern cryo-EM reconstruction pipelines. To address this gap, we now turn to empirical experiments that reflect more realistic workflows. Below, we outline the experimental setup and highlight the key differences from the theoretical model. The corresponding results are presented in Figure~\ref{fig:5}.

\paragraph{Micrograph generation and particle extraction.}
In the theoretical analysis presented in Section \ref{sec:model}, we examined the case where the input micrograph consists solely of noise. This model can be considered an asymptotic case where the SNR of the micrograph approaches zero. However, in practice, true particles are located at unknown positions within the micrograph, with $\text{SNR} > 0$. 

We generate synthetic micrographs with non-overlapping true particles by iteratively placing particles at random locations and rejecting any proposed position that would overlap a previously placed particle. Each micrograph has size $2048 \times 2048$ and contains $200$ projections, randomly drawn from a bank of $50$ projections of the same structure, each of size $d = 36 \times 36$. White Gaussian noise with noise variance $\sigma^2$ is added to achieve an SNR of $\approx 1/25$, where $\text{SNR} \triangleq \frac{\mathrm{Var}(x_i)}{\sigma^2}$ and $x_i$ denotes a normalized particle projection. Using a template-matching threshold of $T = 4$, we collected $M = 2 \times 10^5$ particles in total, corresponding to about $350$ picked particles per micrograph, roughly $75 \%$  more than the $200$ true planted particles.
For the particle extraction, we employ template matching with local maxima, as outlined in Section~\ref{subsec:local-maxima}. A  detailed simulation of this procedure, including the construction of synthetic micrographs and the behavior of template matching under overlapping conditions, is presented in Appendix~\ref{sec:methods} and described explicitly in Algorithm~\ref{alg:particlePickerTemplateMatchingWithOverlapping}.

\paragraph{2D classification.} 
The theoretical model in Section~\ref{sec:model}, excludes in-plane alignment and assumes that the number of templates used during particle picking matches the number of 2D classes, an assumption that is generally not satisfied in practice. In this section, we consider a more realistic cryo-EM pipeline that includes in-plane alignment and allows the numbers of templates and 2D classes to differ.

To align with standard practice, we use RELION’s VDAM algorithm for 2D classification~\cite{kimanius2021new}. RELION is widely adopted in the cryo-EM community and VDAM, introduced in version 4.0, is a variable-metric gradient-based optimizer with adaptive moment estimation, inspired by Adam~\cite{diederik2014adam} and implemented with mini-batches for computational efficiency. Unless stated otherwise, we run VDAM with $L = 20$ classes, 200 iterations, regularization parameter 3, no CTF correction, and in-plane alignment using an angular sampling of $12^\circ$.

\paragraph{3D ab-initio modeling.} 
We also performed ab initio 3D reconstruction using RELION’s VDAM algorithm.
In this context, ab initio reconstruction refers to generating an initial 3D structure without relying on any external reference, thereby reducing potential sources of bias.
We employed RELION’s VDAM algorithm for this task with the following parameters: $L = 1$ class, 200 iterations, a regularization parameter of 4, no CTF correction, and no symmetry constraints (i.e., $C1$ symmetry). 
Unlike the idealized setting of Corollary~\ref{thm:3DstructureInformal}, this empirical reconstruction does not use the same discrete rotation grid as the one used during template matching. Rather, the experiment is intentionally carried out in a more realistic regime, in which the picking and reconstruction stages use different angular discretizations. This allows us to assess the extent to which the predicted confirmation-bias phenomenon persists under practical grid mismatch.

\paragraph{Results.} 
Figure~\ref{fig:5} presents results from the empirical setup in which the micrograph contains genuine particle projections, but the template used for particle picking corresponds to a different structure. In this scenario, the 2D class averages reflect a mixture of two competing sources: the true underlying particles and the mismatched template. As a result, while some class averages recover features of the actual particles, others resemble the template used for picking. Similarly, the ab initio 3D reconstruction derived from these particles exhibits structural elements from both the true volume and the template. This highlights the risk of introducing structural bias at the particle-picking stage, where even in the presence of a signal, the use of an incorrect template can propagate artifacts into downstream analyses and compromise the reliability of the final reconstruction.
In addition, we note that the use of a markedly dissimilar template in Figure~\ref{fig:5} is intended to make the bias mechanism visually transparent. In practice, however, the more dangerous regime may be one in which the incorrect template is structurally similar to the true particle, since the resulting bias can be substantially harder to detect.

\begin{figure}[t!]
    \centering
    \includegraphics[width=1.0\linewidth]{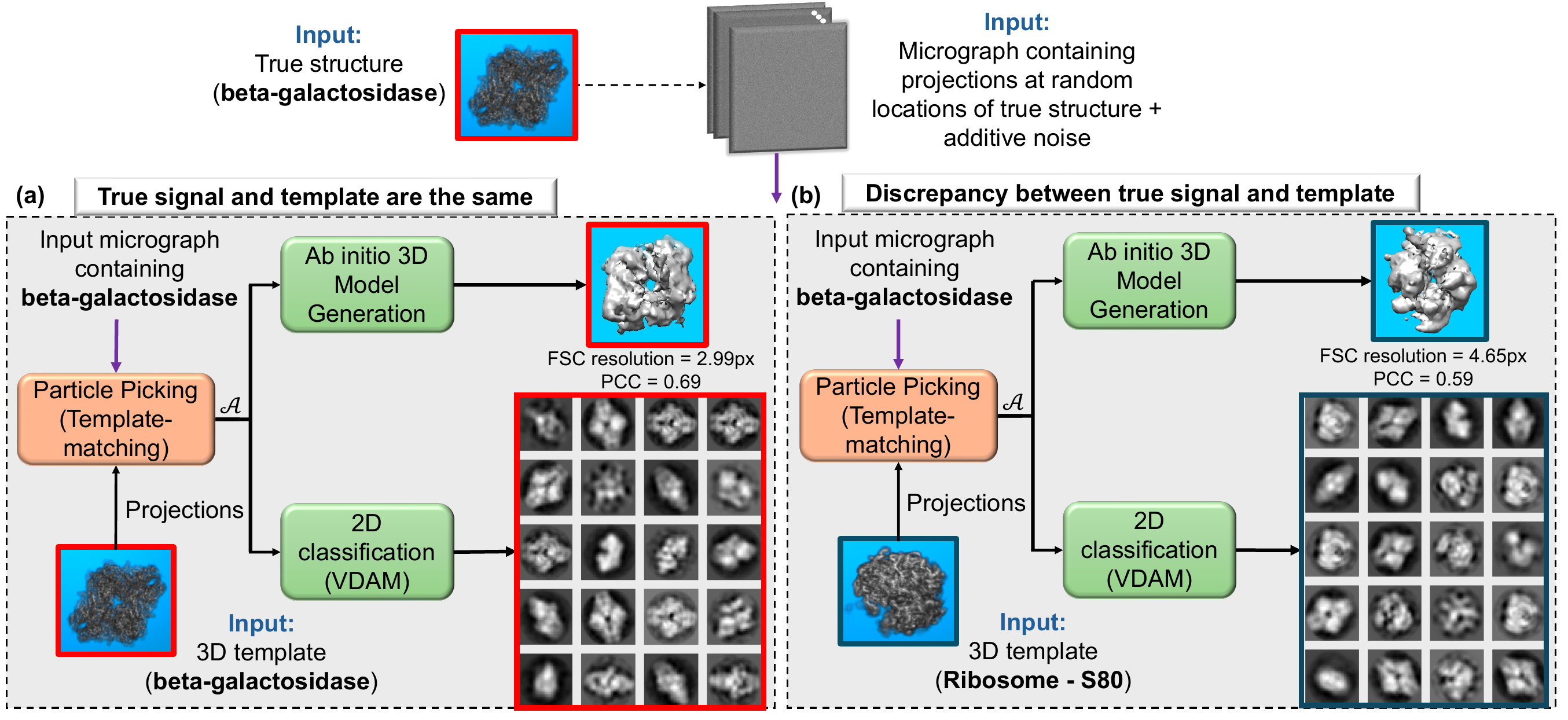}
    \caption{\textbf{Confirmation bias in particle picking when particles are present.}
    Micrographs contain true particle projections of beta-galactosidase~\cite{bartesaghi2014structure} corrupted by additive white Gaussian noise (SNR $\approx 1/25$), and particle picking is performed by template matching using either:
    \textbf{(a)} projections of the true beta-galactosidase structure, or
    \textbf{(b)} projections of an unrelated structure, the Plasmodium falciparum 80S ribosome~\cite{wong2014cryo}.
    In both cases, the extracted particles are processed by 2D classification and ab initio 3D reconstruction in RELION (VDAM)~\cite{kimanius2021new}.
    With matched templates \textbf{(a)}, the 2D classes and 3D map resemble the ground truth (PCC $=0.69$, Fourier Shell Correlation (FSC) resolution $=2.99$ pixels);
    with mismatched templates \textbf{(b)}, spurious ribosome-like classes appear and the 3D reconstruction becomes a hybrid, with degraded agreement (PCC $=0.59$,  FSC resolution $=4.65$ pixels).
    Settings: patch size $d=36\times36$, $M=2\times10^5$ extracted particles, threshold $T=4$; VDAM ran for 200 iterations without CTF correction.}
    \label{fig:5}
\end{figure}

\subsection{Particle picking based on Topaz neural network}
Deep learning–based particle pickers are now widely used in cryo-EM and cryo-ET, with Topaz being a prominent example~\cite{bepler2019positive}. Unlike template matching, which introduces bias through  more intricate mechanisms, such as explicit similarity to predefined structures, deep neural networks can inherit bias from their training data, architecture, and optimization dynamics. As a result, they may still exhibit confirmation bias by preferentially selecting particles that resemble features emphasized during training. 
To probe this effect, we consider both the pre-trained Topaz model~\cite{bepler2019positive}, used without additional training, and Topaz models trained on user-specified particle sets.

\paragraph{Pre-trained Topaz network.}
For the pre-trained Topaz network, we applied Topaz to micrographs of size $2048 \times 2048$, using an internal down-sampling factor of $s=2$. During extraction, we set the particle radius to $r=9$, corresponding to a patch size of $d = 36 \times 36$ on the original grid, and used a coordinate scaling factor of $x=2$ to account for the down-sampling.
We used a score threshold of $-1$ for the Topaz output, corresponding to a predicted particle probability of $p \ge 0.25$. Although the default threshold in Topaz and RELION is typically $0$ (that is, $p \ge 0.5$), in practice users often lower this threshold to improve recall, especially for challenging low-SNR datasets. We adopted this more permissive setting to extract a sufficiently large particle stack ($M = 2 \times 10^5$) within a reasonable runtime, enabling a direct comparison with the template-matching experiments in both dataset size and effective SNR. Thus, the observed structure-from-noise artifacts arise under realistic conditions, rather than from forcing the network to select arbitrary low-confidence patches.

Figure~\ref{fig:6}(a-c) shows that the pre-trained Topaz network induces structured bias. The mean of the selected particles by Topaz in Figure~\ref{fig:6}(a) displays a faint central circular feature, whereas under an unbiased model the mean of pure-noise patches should converge to zero. A similar central pattern appears in Figure~\ref{fig:6}(b), which shows the mean of particles selected by template matching with a large, heterogeneous template set, suggesting that both methods imprint structure on the picks. In contrast, the mean of uniformly random patches from the same noise micrographs (Figure~\ref{fig:6}(c)) is essentially zero.

To assess the impact on 3D refinement, we perform a controlled experiment in which pure-noise micrographs are split into two independent halves. Each half is processed with one of three picking strategies: template matching, Topaz, or random selection. Independent ab initio 3D reconstructions are then computed from each half using RELION’s VDAM algorithm, and their consistency is measured by the Fourier Shell Correlation (FSC), shown in Figure~\ref{fig:6}(d). Both Topaz and template matching yield high FSC curves between the halves, indicating that each method produces reproducible, method-dependent structure from noise. In contrast, random selection yields near-zero FSC, as expected under an unbiased model. Among the three, template matching exhibits the strongest bias, followed by Topaz. These results further demonstrate that the FSC gold standard can be misled by systematic errors, in line with the observations of~\cite{sorzano2022bias}.

\begin{figure}[t!]
    \centering
    \includegraphics[width=0.9\linewidth]{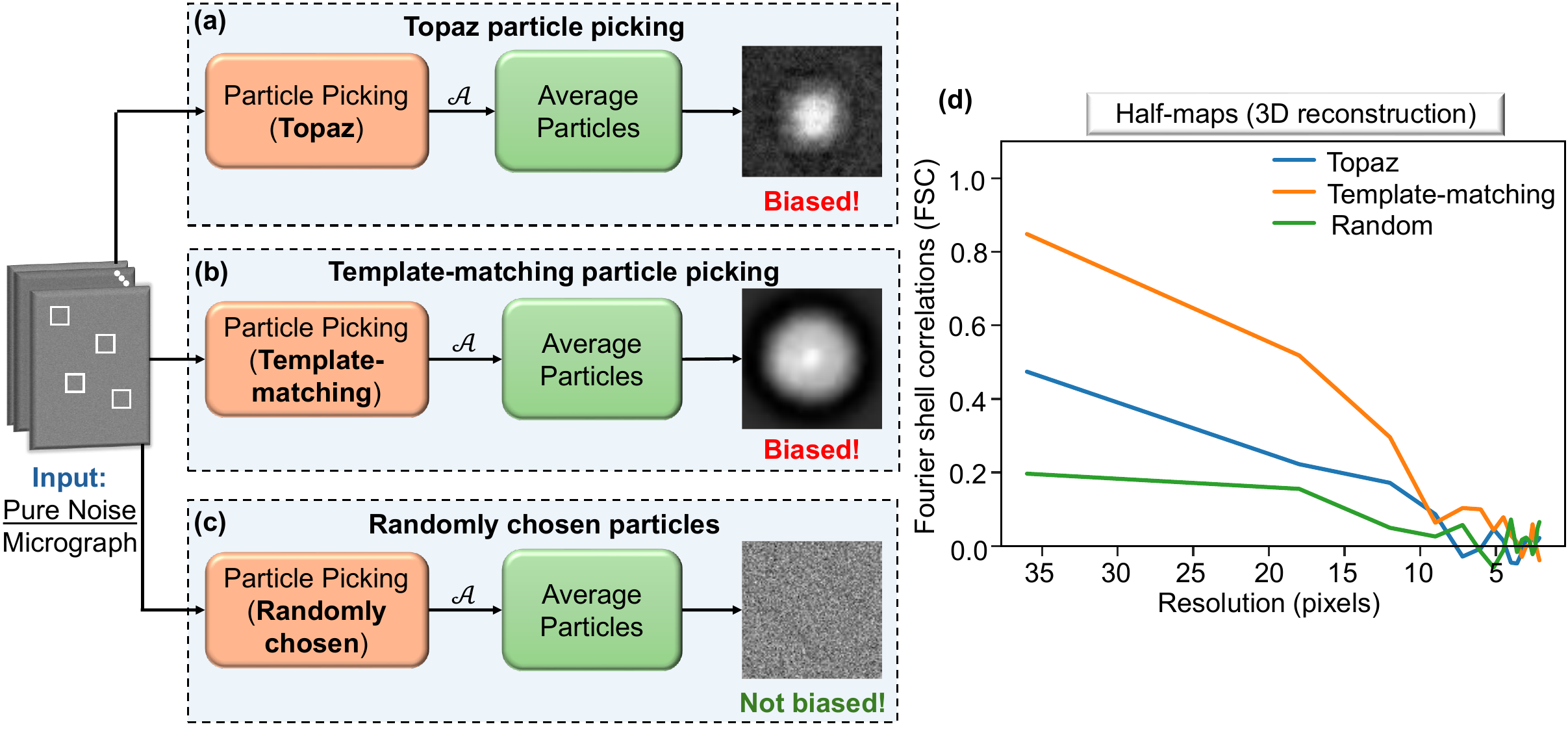}
    \caption{\textbf{Confirmation bias in pre-trained Topaz particle picking.}
    \textbf{(a--c)} Mean of picked particles from pure-noise micrographs: (a) Topaz picks yield a nonzero mean with a prominent central circular feature; (b) template matching with a large bank of Plasmodium falciparum 80S ribosome projection templates~\cite{wong2014cryo} produces a similar pattern; (c) random picks average to zero.
    \textbf{(d)} Bias assessment via half-map reproducibility: noise micrographs are split into two independent halves; on each half we pick $M=10^5$ particles using Topaz, template matching, or random selection, and reconstruct two half-maps using RELION VDAM. We report the Fourier Shell Correlation (FSC) between the half-maps. Settings: patch size $d=36\times36$; VDAM ran for 200 iterations without CTF correction. }

    \label{fig:6}
\end{figure}

\begin{figure*}[!t]
    \centering
    \includegraphics[width=1.0\linewidth]{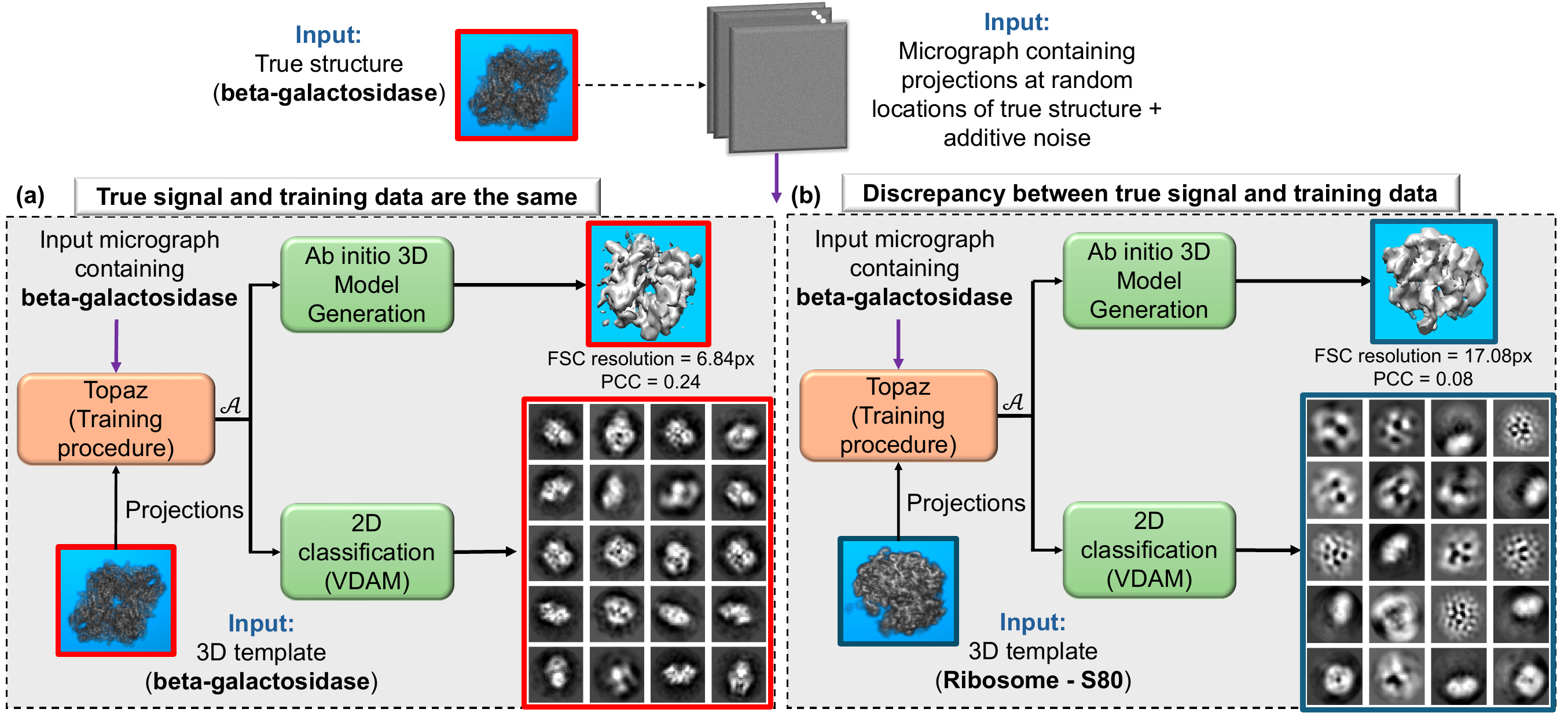}
    \caption{\textbf{Confirmation bias in trained Topaz particle picking when true particles are present.}
    Micrographs contain true particle projections of beta-galactosidase~\cite{bartesaghi2014structure} corrupted by additive white Gaussian noise. Particle picking is performed using a Topaz model trained on either:
    \textbf{(a)} projections of the true beta-galactosidase structure (matched training data), or
    \textbf{(b)} projections of an unrelated structure, the Plasmodium falciparum 80S ribosome~\cite{wong2014cryo} (mismatched training data).
    In both cases, the extracted particles are processed by 2D classification and ab initio 3D reconstruction in RELION (VDAM)~\cite{kimanius2021new}.
    With the matched training data \textbf{(a)}, the 2D classes and 3D map successfully recover the beta-galactosidase structure (PCC $= 0.24$, FSC resolution $= 6.84$ pixels).
    Conversely, when trained on the mismatched structure \textbf{(b)}, the network preferentially selects noise features that align with its learned prior. As a result, the 3D reconstruction of the true underlying signal is significantly degraded.}
    \label{fig:9}
\end{figure*}

\paragraph{Trained Topaz network.}
To compare deep-learning-based picking with classical template matching under realistic conditions, we tested whether Topaz's learned prior can override weak true signal when real particles are present. Following the setup of Figure~\ref{fig:5}, we generated synthetic micrographs containing true beta-galactosidase projections at low SNR ($\approx 1/25$). Starting from the same vanilla pre-trained Topaz model used in the previous experiment, we re-trained two networks: one on positive examples from beta-galactosidase, and one on an unrelated structure, the \textit{Plasmodium falciparum} 80S ribosome. In both cases, training was run for 10 epochs on 20 synthetic micrographs, each containing 200 particle projections embedded in white Gaussian noise at the same low-SNR level as the evaluation data.

As shown in Figure~\ref{fig:9}, when Topaz is trained on the mismatched ribosome data, the extracted particles yield 2D class averages and a 3D reconstruction with clear ribosome-like features. The reconstruction quality is also substantially worse than in the matched-training case (PCC 0.08 versus 0.24, and FSC resolution 17.08 pixels versus 6.84 pixels). This shows that, much like template matching, a dedicated neural-network picker can be strongly affected by the specific structural prior imposed during training. Rather than introducing only a generic bias toward centered, particle-like blobs, the network preferentially selects low-SNR data that resemble its training distribution, thereby propagating structure-specific hallucinations into the downstream pipeline.

\section{Discussion and Outlook} \label{sec:discussion}

To place the bias analyzed in this work in context, we compare it with other sources of bias that arise throughout the cryo-EM and cryo-ET reconstruction pipelines. We then discuss mitigation strategies and conclude with several promising directions for future research.

\subsection{Different forms of confirmation bias in cryo-EM}
\label{sec:different-statistical-models}

Several stages of the cryo-EM and cryo-ET pipeline are vulnerable to confirmation bias~\cite[Section 4.1]{sorzano2022bias}, including particle picking, 2D classification, and 3D reconstruction. The present work focuses on particle picking~\cite{scheres2015semi, heimowitz2018apple, bepler2019positive,eldar2020klt}, while related biases in low-SNR EM-based reconstruction have been analyzed in~\cite{balanov2025expectation}.
These stages introduce bias through a different statistical mechanism. Figure~\ref{fig:7} compares confirmation bias in particle picking with that in 2D classification. Below, we contrast the Einstein from Noise phenomenon from Section~\ref{sec:intro} and the bias observed in 2D classification~\cite{balanov2025confirmation} with the bias arising from particle picking.

\paragraph{The Einstein from Noise phenomenon.}  
Although frequently mentioned in the context of template-based particle picking, the Einstein from Noise effect arises from a distinct model. In the Einstein from Noise experiment, all observations are aligned and averaged with respect to a fixed template. In contrast, particle picking selects a subset of patches with the highest template correlations. As a result, particle picking produces reconstructions that asymptotically match the template up to global scaling (Theorems~\ref{thm:classesCentersVersusTemplatesInformal} and~\ref{thm:3DstructureInformal}), whereas Einstein from Noise yields an average that is only positively correlated with the template but not identical~\cite{balanov2024einstein}.

\paragraph{Confirmation bias in 2D classification and 3D volume reconstruction.}  
In later stages of cryo-EM processing, such as 2D classification and 3D volume estimation, iterative algorithms often use predefined templates for initialization. When applied to data with no true signal, these methods may still produce output that resembles the initial templates, even though the input data provides no evidence for such structure~\cite{balanov2025confirmation}.

Figure~\ref{fig:7}(b) illustrates this effect in 2D classification. Here, the input consisted entirely of pure noise, and the initial templates were $L=30$ distinct projections of the Beta-galactosidase structure of size $d=36 \times 36$. A total of $M=5 \times 10^5$ particles were processed by EM-based 2D classification for 20 iterations. Even without signal, the resulting class averages remain positively correlated with the templates. Unlike template-based particle picking, however, they do not reproduce the templates exactly, although they remain visually similar. Many modern 2D classification and 3D ab initio reconstruction methods use reference-free initialization, such as random initialization, to mitigate this bias~\cite{kimanius2021new}. Nevertheless, similar confirmation effects can still arise when templates are introduced explicitly or implicitly.

\begin{figure}[t!]
    \centering
    \includegraphics[width=0.65\linewidth]{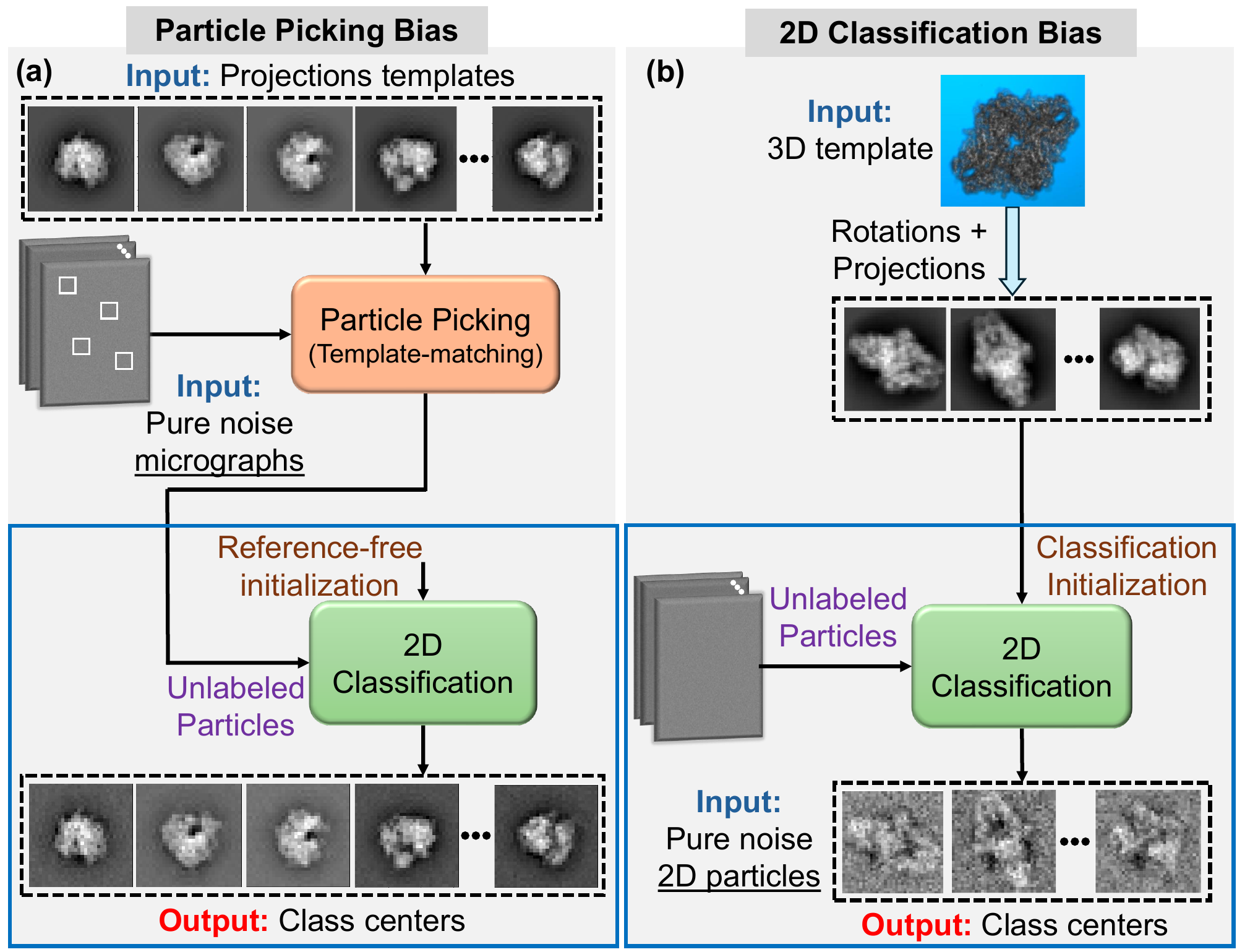}
    \caption{\textbf{Comparison of confirmation bias in particle picking versus 2D classification.} 
    \textbf{(a)} Confirmation bias during the particle-picking stage, the primary bias mechanism analyzed in this work. 
    \textbf{(b)} Confirmation bias during the 2D classification stage. Here, the input consists solely of pure noise particles, and the classification is initialized using predefined templates. Despite the absence of signal in the data, the resulting class averages exhibit a bias toward the initial templates, indicating a positive correlation (i.e., confirmation bias) \cite{balanov2025confirmation}. Notably, the class centers in Panel (a) appear to more closely resemble the initial templates than those in Panel (b).}
    \label{fig:7}
\end{figure}

\subsection{Strategies for mitigating confirmation bias}

Previous studies have explored the detection and mitigation of various forms of bias in cryo-EM. For example, recent work~\cite{sorzano2022bias} has proposed several practical strategies to reduce reconstruction bias, including identifying bias at the level of estimated parameters, performing multiple independent reconstructions, averaging parameter estimates to reduce variance, applying phase randomization at the image level, and validating algorithms across diverse preprocessing and reconstruction settings. Similarly, an additional work~\cite{lucas2023baited} provided an empirical analysis of template bias in high-resolution 2D template matching and proposed practical conditions and diagnostics to reduce template-induced bias in reconstruction. Despite these advances, the particle-picking stage remains particularly vulnerable. Here, we focus specifically on computational strategies to mitigate bias during particle picking, incorporating both practical considerations and new theoretically grounded approaches.

\paragraph{Statistical control of selection thresholds.}

Despite its importance, the threshold $T$ is typically chosen heuristically or manually, without rigorous statistical justification. A more principled alternative is to set thresholds using multiple hypothesis testing frameworks that control the false discovery rate (FDR)~\cite{benjamini1995controlling}. These techniques guarantee that the expected proportion of false positives remains below a predefined level, while typically achieving high statistical power. As a result, they help reduce the selection of spurious particles that arise due to random correlations with the templates. However, these methods are often specifically tailored to well-specified statistical models and may not generalize easily to more complex or poorly understood settings~\cite{eldar2024object}. Recent developments, such as knockoff-based procedures, offer a promising direction for extending FDR control to more flexible and model-agnostic scenarios by introducing synthetic variables that mimic the data's dependence structure~\cite{barber2015controlling}.

\paragraph{Template filtering and design.}
Several tools in current cryo-EM pipelines aim to reduce bias by preprocessing the templates themselves~\cite{scheres2012relion}. For example, low-pass filtering the templates reduces their high-frequency features, which are most susceptible to overfitting by noise. This encourages the selection of broader structural motifs rather than fine-grained details, mitigating high-frequency confirmation bias. Our theoretical model (see Section~\ref{sec:model}) shows that this practice leads to reconstructions that converge toward low-pass-filtered versions of the templates. While not a complete remedy, such filtering is routinely employed and is supported by empirical evidence, including the analysis presented in~\cite{balanov2024einstein}.

\paragraph{Template-free picking and alternative similarity measures.}
A natural way to mitigate confirmation bias is to reduce or avoid structure-specific templates during particle picking. In practical cryo-EM workflows, this is often done by starting with template-free candidate generation, such as blob- or Laplacian-of-Gaussian (LoG)-based picking, followed by 2D classification. The resulting class averages can then be reused as internal references in later picking rounds. Unlike classical template matching with an external structural reference, blob/LoG methods rely only on weak, generic, low-resolution shape priors rather than molecule-specific templates. Accordingly, based on the results of this work, one expects these methods to favor coarse, blob-like, low-frequency structure rather than detailed structural features.

More broadly, several template-free approaches have been proposed for unsupervised particle detection without external references~\cite{heimowitz2018apple}. Such methods reduce the direct influence of externally imposed structural priors and therefore offer a potentially safer alternative to template-based picking. However, they are also expected to become unreliable below some SNR threshold.

A complementary mitigation strategy within the template-based paradigm is to refine the similarity measure used during matching. Rather than relying on standard cross-correlation, one may consider more robust criteria, such as cross-entropy-based measures~\cite{sorzano2010clustering}. While promising, the design, analysis, and empirical validation of such alternatives remain largely open challenges.

\paragraph{Bypassing particle picking.}
Below a certain SNR threshold, any particle-picking method is fundamentally doomed to fail~\cite{dadon2024detection}. One way to circumvent this limitation is to eliminate particle picking altogether. Recently, it was shown, both theoretically and empirically, that one can reconstruct the 3D structure directly from raw micrographs at arbitrarily low SNR by bypassing particle picking and treating the unknown particle locations as nuisance variables~\cite{balanov2025orbit}. Earlier direct-structure methods, including approaches based on the method of moments~\cite{bendory2023toward} and approximations to EM~\cite{kreymer2023stochastic}, likewise aim to recover 3D volumes directly from micrographs without explicit particle selection. In principle, these approaches avoid selection bias altogether.

\subsection{Future work}
Several directions remain open for future work.

\paragraph{Beyond the present picking model.}
A natural next step is to extend the present theory beyond the pure-noise regime to the practically relevant low-SNR setting, where genuine particles are present. Such an analysis could quantify how genuine signal and template-induced bias interact in downstream reconstruction, and how it depends on the SNR, the density of the true particles, the picking threshold, and the mismatch between the true structure and the picking template.

Another important direction is to develop a theoretical framework for deep particle pickers. Unlike hand-crafted template-matching rules, deep pickers rely on learned scoring functions whose inductive biases are shaped by the training data, architecture, and optimization procedure. A rigorous analysis of such learned selection rules would broaden the scope of the present work and help clarify how modern data-driven picking pipelines may amplify or mitigate bias in low-SNR regimes.

A related extension is to move beyond controlled synthetic micrographs to realistic experimental cryo-EM data. Experimental micrographs contain structured background effects, such as ice gradients, carbon support edges, contamination, and variable ice thickness, that are absent from the idealized pure-noise model studied here. 

\paragraph{CTF-aware modeling.}
For analytical tractability, the present work does not model the CTF explicitly. In the pure-noise regime considered here, a fixed CTF acts as a linear filter, so Gaussian background noise remains Gaussian after CTF filtering, but with a modified covariance structure. From this viewpoint, CTF effects are partially absorbed into the correlated-noise model through the covariance matrix in Theorem~\ref{thm:classesCentersVersusTemplatesInformalStationary}. In practice, however, template matching is often CTF-aware and the CTF may vary across micrographs due to defocus changes. Extending the framework to explicit CTF-aware scoring and heterogeneous defocus is therefore an important direction for future work.

\paragraph{Cryo-ET validation and missing-wedge effects.}
An important direction for future work is broader empirical and theoretical validation in a practical cryo-ET setting. While the present work includes controlled cryo-ET illustrations on synthetic pure-noise tomograms, these experiments are intended as proof-of-concept demonstrations of the bias mechanism rather than fully realistic subtomogram-averaging studies. Incorporating key experimental features such as missing-wedge effects~\cite{chen2019complete}, tomographic reconstruction artifacts, and validation on real cryo-ET datasets therefore remains for future work.

\section*{Data Availability}
The detailed implementation and code are available at \href{https://github.com/AmnonBa/particle-picking-confirmation-bias}{https://github.com/AmnonBa/particle-picking-confirmation-bias}.

\section*{Acknowledgment}
T.B. is supported in part by BSF under Grant 2020159, in part by NSF-BSF under Grant 2019752, in part by ISF under Grant 1924/21, and in part by a grant from The Center for AI and Data Science at Tel Aviv University (TAD). We thank Amitay Eldar for insightful discussions and for his contributions to the initial formulation of the framework.

\bibliographystyle{plain}

\begin{appendices}

{\centering{\section*{Appendix}}}

\paragraph{Appendix organization.} 
Appendix~\ref{sec:problemFormulation} provides the preliminaries for the appendix. 
Appendix~\ref{sec:2Dclass} develops the theoretical framework and proofs for our general bias analysis of the GMM mean maximum-likelihood estimators, and its relation to the templates used in the template-matching selector.
Appendix~\ref{sec:3D-reconstruction} extends this framework to the 3D reconstruction process in both single-particle cryo-EM and cryo-ET.
In Appendix~\ref{sec:methods}, we describe in detail the empirical setup and procedures used for the simulations presented in Section~\ref{sec:empirical}.
Appendix~\ref{sec:different-statistical-models} compares the confirmation bias induced by template-matching selection, as analyzed in this work, with other sources of bias arising in single-particle cryo-EM and cryo-ET.

\paragraph{Notations.} 
Throughout, $\xrightarrow[]{\mathcal{D}}$, $\xrightarrow[]{\mathcal{P}}$, $\xrightarrow[]{\text{a.s.}}$, and $\xrightarrow[]{\mathcal{L}^p}$ denote convergence in distribution, in probability, almost surely, and in $\mathcal{L}^p$, respectively, for sequences of random variables. The Euclidean inner product is written as either $\langle a,b\rangle$ or $a^\top b$ when $a,b\in\mathbb R^d$, and as $\langle a,b\rangle\triangleq a^{\ast} b$ when $a,b\in\mathbb C^d$.
We denote by $\s{SO}(3)$ the special orthogonal group representing 3D rotations. All norms $\| \cdot \|$ are Frobenius norms. $\mathbbm{1}$ represents the indicator function. 

We employ standard asymptotic notation. For nonnegative sequences ${a_n}$ and ${b_n}$, we write $a_n=O(b_n)$ if there exists a constant $C<\infty$ and $n_0$ such that $a_n\le C b_n$ for all $n\ge n_0$, and $a_n=o(b_n)$ if $a_n/b_n\to 0$. We write $a_n=\omega(b_n)$ if $a_n/b_n\to\infty$, and $a_n=\Theta(b_n)$ if both $a_n=O(b_n)$ and $b_n=O(a_n)$ hold. 

For random sequences, let $\{X_n\}$ be random variables and $\{b_n\}$ a positive deterministic sequence. We write $X_n = O_{\mathbb{P}}(b_n)$ if the family $\{X_n/b_n\}$ is bounded in probability, that is, for every $\varepsilon>0$ there exists $M<\infty$ such that $
    \sup_{n} \mathbb{P}\bigl(|X_n/b_n| > M\bigr) \le \varepsilon.$
We write $X_n = o_{\mathbb{P}}(b_n)$ if $X_n/b_n \to 0$ in probability, that is, for every $\varepsilon>0$,
$ \mathbb{P}\bigl(|X_n/b_n| > \varepsilon\bigr) \longrightarrow 0$ as $n\to\infty.$

\section{Preliminaries}
\label{sec:problemFormulation}

The particle picker's role is to classify each patch as either a particle or noise assuming no overlap between patches. 
Each micrograph patch, or equivalently, candidate particle is represented as $\{y_i\}_{i=0}^{N-1} \in \mathbb{R}^d$, where $N$ is the number of candidate particles, and $d$ is their dimension (which may correspond to the number of pixels or voxels in 2D image or 3D volume).
Accordingly, the particle-picking task is defined by a binary function $f: \mathbb{R}^d \to \{0,1\}$:  
\begin{align}
    f(y_i) = 
    \begin{cases} 
        1, & \text{if } y_i \text{ is a particle}, \\
        0, & \text{otherwise}.
    \end{cases}
\end{align}

To analyze confirmation bias in particle picking, we consider two noise models for the candidate particles $\{y_i\}_{i=0}^{N-1}$. In the Gaussian setting, we assume $y_i \sim \mathcal{N}(0,\Sigma)$, which covers both Model~\ref{model:whiteNoiseIID} (i.i.d.\ with $\Sigma = \sigma^2 I_{d}$) and Model~\ref{model:stationaryGaussianNoisePatches} (stationary Gaussian process with covariance $\Sigma$ and $\alpha$-mixing coefficients satisfying Assumption~\eqref{eq:alpha-mixing-summability-1d}).
Alternatively, in the spherically symmetric i.i.d.\ setting (Model~\ref{model:isotropicIID}), we assume $y_i$ are i.i.d.\ with density $p_{\mathrm{sph}}(y)$, which by the spherical symmetry has the radial form
\begin{align}
    p_{\mathrm{sph}}(y) = \phi(\|y\|_2), \qquad y \in \mathbb{R}^d.
\end{align}

We consider a particle picker based on \textit{template matching}, as outlined in Algorithm~\ref{alg:particlePickerTemplateMatching}. This picker operates by cross-correlating each candidate particle $\{y_i\}_{i=0}^{N-1}$ with $L$ predefined templates, denoted $\{x_\ell\}_{\ell=0}^{L-1}$. Templates are assumed to be normalized so that $\|x_\ell\|_2 = 1$ for all $\ell \in \{0, \ldots, L-1\}$. If the correlation with at least one template exceeds a threshold $T$, the observation is classified as a particle and added to the set of selected particles, $\mathcal{A}$. We denote by $M =\abs{\mathcal{A}}$ the number of selected particles. 

To examine the confirmation bias introduced during the particle selection stage, it is necessary to establish a metric to compare the extracted particles with the initial templates used in the template selection process. Specifically, we analyze two approaches applied to the particles extracted during this stage: one is based on the 2D classification output (relevant for cryo-EM) and the other is based on the 3D reconstruction process (relevant for sub-tomogram averaging in cryo-ET and 3D reconstruction process in single-particle cryo-EM).
In Appendix \ref{sec:2Dclass}, we focus on the 2D classification output, modeled using Gaussian mixture models (GMM). We analyze the relationship between the maximum-likelihood estimates of the GMM means and the predefined templates employed during particle picking.
In Appendix~\ref{sec:3D-reconstruction}, we turn to 3D reconstruction process and analyze the connection between the 3D template volume used in the particle-picking stage and the maximum-likelihood estimator derived from the 3D reconstruction process output. This analysis aims to uncover how the initial templates shape the reconstruction results and whether they contribute to biases in the reconstruction workflow.

We state the following known result for truncated Gaussian distributions which will be used through the proofs.
\begin{lem}[\cite{durrett2019probability}] \label{lem:C4}
Let $X \sim \mathcal{N}(0, \sigma^2)$. Then, for any threshold $T \in \mathbb{R}$,
\begin{align}
    \mathbb{E}[X \mid X \geq T] = \sigma \cdot \frac{\varphi(T/\sigma)}{Q(T/\sigma)}, \label{eqn:app_C13}
\end{align}
where $\varphi(z) = \frac{1}{\sqrt{2\pi}} e^{-z^2/2}$ is the standard normal density function,
$Q(z) = \int_z^\infty \varphi(t) \, dt = 1 - \Phi(z)$ is the standard Gaussian upper tail probability, and $\Phi$ is the CDF of the standard normal distribution.
In addition, the variance of $X$ conditioned on $X \geq T$ is given by:
\begin{align}
    \operatorname{Var}(X \mid X \geq T) = \sigma^2 \left(1 + \frac{T}{\sigma} \cdot \frac{\varphi(T/\sigma)}{Q(T/\sigma)} - \left(\frac{\varphi(T/\sigma)}{Q(T/\sigma)}\right)^2\right). \label{eqn:app_C14}
\end{align}
Asymptotically, as $T \to \infty$, we have, \footnote{Here, the symbol $f(t) \sim g(t)$ denotes asymptotic equivalence, i.e., $\lim_{t \to \infty} \frac{f(t)}{g(t)} = 1.$}
\begin{align}
    \mathbb{E}[X \mid X \geq T] &\sim T + \frac{\sigma^2}{T}, \\
    \operatorname{Var}(X \mid X \geq T) &\sim \sigma^2 \left(1 - \frac{\sigma^2}{T^2}\right).
\end{align}
\end{lem}

\section{2D classification in cryo-EM} \label{sec:2Dclass}

\subsection{2D classification process overview} 

In cryo-EM, following the particle-picking stage, the extracted particles are typically subjected to a 2D classification process. This step groups particles into clusters based on their similarity, with each cluster represented by a centroid, which is the average of the observations assigned to that cluster. These centroids serve two primary purposes: they help assess the overall quality of the dataset and provide a basis for iteratively sorting out false positives or poorly aligned particles. In practice, the appearance of particle-like features in the centroids is often interpreted by researchers as evidence that the selected particles contain meaningful structural information.

\paragraph{2D Classification based on Gaussian mixture models.} 
During the 2D classification process, particles are often modeled using a GMM, a widely used and effective probabilistic framework for describing particle heterogeneity. While alternative models exist, we focus on the GMM due to its analytical tractability. In this framework, the observed data are assumed to be generated from a mixture of Gaussian distributions, with each component representing a distinct subpopulation within the dataset. Formally, the realizations of the extracted particles are modeled as:
\begin{align}
     (\text{Postulated statistics}) \qquad z_0, z_1, \ldots, z_{M-1} \stackrel{\text{i.i.d.}}{\sim} \sum_{\ell=0}^{L-1} w_\ell \cdot \mathcal{N}(\mu_\ell, \Sigma_\ell), \label{eqn:GMMmodel}
\end{align}
where $z_0, z_1, \ldots, z_{M-1}$ are the realizations of the \textit{postulated} GMM distribution of the extracted particles, $w_\ell$, $\mu_\ell$, and $\Sigma_\ell$ are the weight, mean, and covariance matrix of the $\ell$-th Gaussian component, respectively. The weights $\ppp{w_\ell}_{\ell=0}^{L-1}$ satisfy $w_\ell \geq 0$, and $\sum_{\ell=0}^{L-1} w_\ell = 1$, and are assumed to be known. For the purpose of theoretical modeling, we assume that the number of templates used in the particle-picking process is equal to the number of classes in the GMM, denoted by $L$.  this assumption. 
This assumption is formally stated below as Assumption~\ref{assump:0}.

\begin{assum}\label{assump:0}
The number of templates used in the particle extraction algorithm (Algorithm \ref{alg:particlePickerTemplateMatching}) and the number of classes assumed in the GMM during the 2D classification process \eqref{eqn:GMMmodel} are equal to $L$.
\end{assum}

\begin{remark}
    Throughout the Appendix, we use the notation $\{y_i\}_{i=0}^{N-1}$ to denote the candidate particles provided as input to the template-matching particle picker. The resulting set of selected particles, which serves as the input for downstream tasks such as 2D classification and 3D reconstruction, is denoted by $\{z_i\}_{i=0}^{M-1}$.
\end{remark}

\paragraph{The log-likelihood function of GMM.}
In cryo-EM, it is common to assume that all components of the GMM share the same covariance matrix $\Sigma_{\text{GMM}}$. 
The objective of the 2D classification process is to estimate the maximum-likelihood centroids $\ppp{\mu_\ell}_{\ell=0}^{L-1}$ in \eqref{eqn:GMMmodel}.
In this context, the likelihood of the GMM specified in \eqref{eqn:GMMmodel} is given by:
\begin{align}
    \nonumber \mathcal{L}\p{\ppp{\mu_\ell}_{\ell \in \pp{L}}; \mathcal{A} = \ppp{z_i}_{i \in \pp{M}}} & = \prod_{i=0}^{M-1} f \p{z_i ; \ppp{\mu_\ell}_{\ell=0}^{L-1}}
    \\ & = \prod_{i=0}^{M-1} \sum_{\ell=0}^{L-1} w_\ell \cdot \mathcal{N}\p{z_i ; \mu_\ell, \Sigma_{\text{GMM}}}, \label{eqn:maximumLikelihoodTermMain}
\end{align}
where we have defined, 
\begin{align}
    f\p{z ; \ppp{\mu_\ell}_{\ell=0}^{L-1}} \triangleq \sum_{\ell=0}^{L-1} w_\ell  f_{\ell}(z) = \sum_{\ell=0}^{L-1} w_\ell \cdot \mathcal{N}\p{z ; \mu_\ell, \Sigma_{\text{GMM}}}. \label{eqn:GMMmodelExcplicitMain}
\end{align}
Here, $\mathcal{N}(z_i ; \mu_\ell, \Sigma_{\text{GMM}})$ represents the multivariate Gaussian density for the $\ell$-th component, and the weights $w_\ell$ are the mixture probabilities for each component. The objective of the 2D classification process is to estimate the centroids $\ppp{\mu_\ell}_{\ell = 0}^{L-1}$, using the maximum-likelihood estimator that maximizes the likelihood \eqref{eqn:maximumLikelihoodTermMain}. 
The maximum-likelihood centroids, denoted by $\widehat{\mu}_\ell$, are formally defined as
\begin{align}
    \ppp{\widehat{\mu}_\ell}_{\ell=0}^{L-1} = \argmax_{\ppp{\mu_\ell}_{\ell=0}^{L-1}} \mathcal{L}(\ppp{\mu_\ell}_{\ell=0}^{L-1}; \mathcal{A}). \label{eqn:maximumOfLogLikelihood}
\end{align}

Thus, the aim of the 2D classification is to estimate the centroids of the GMM, defined in \eqref{eqn:GMMmodel} using the extracted unlabeled particles $\mathcal{A} = \ppp{z_i}_{i=0}^{M-1}$ 
(i.e., using the output particles of Algorithm~\ref{alg:particlePickerTemplateMatching}). The main question we address in this subsection is how the estimated centroids $\ppp{\widehat{\mu}_{\ell}}_{\ell=0}^{L-1}$ relate to the templates $\ppp{x_{\ell}}_{\ell=0}^{L-1}$ used in the particle-picker template-matching (Algorithm \ref{alg:particlePickerTemplateMatching}), when the input micrograph is pure noise.

\begin{remark}
We use the notation $\Sigma_{\s{GMM}}$ to denote the covariance matrix of the postulated GMM model described in~\eqref{eqn:GMMmodel}, and $\Sigma$ to denote the covariance matrix of the noise in the input to the particle picker, as described in Algorithm \ref{alg:particlePickerTemplateMatching}. These two covariance matrices are not necessarily equal in our analysis. 
\end{remark}

\paragraph{The extracted particles form a non-Gaussian mixture model.} \label{sec:misspecifiedModels}

Although the extracted particles are often assumed to follow a Gaussian distribution, as defined in \eqref{eqn:GMMmodel}, the statistical properties of the particle picking process described in Algorithm~\ref{alg:particlePickerTemplateMatching} deviate from this assumption. Formally, the law of a particle $Z$ extracted by Algorithm~\ref{alg:particlePickerTemplateMatching} is a finite mixture
\begin{align}
    g^{}(z) = \sum_{\ell=0}^{L-1} \pi_\ell \, g_\ell^{(T)}(z), \qquad z \in \mathbb{R}^d,
    \label{eqn:mixtureModelParticlePicker}
\end{align}
where $L$ is the number of templates, the mixing weights satisfy $\pi_\ell \ge 0$ and $\sum_{\ell=0}^{L-1} \pi_\ell = 1$, and $g_\ell^{(T)}:\mathbb{R}^d \to [0,\infty)$ denotes the density associated with template $x_\ell$.

In the Gaussian setting, where the underlying patches satisfy
$y_i \sim \mathcal{N}(0,\Sigma)$, each component has the truncated-Gaussian form
\begin{align}
    g_\ell^{(T)}(z) = C_\ell^{(T)}\, \exp\!\left(-\frac{1}{2} z^\top \Sigma^{-1} z\right) \,\mathbbm{1}_{\{\langle z, x_\ell\rangle \ge T\}},    \label{eqn:componentOfMixtureModelParticlePicker}
\end{align}
where $T$ is the template–matching threshold and $C_\ell^{(T)}$ is the
normalizing constant chosen so that
$\int_{\mathbb{R}^d} g_\ell^{(T)}(z)\,dz = 1$, that is, 
\begin{align}
    C_\ell^{(T)} = \biggl(\int_{\{\langle z, x_\ell\rangle \ge T\}}\exp\!\left(-\frac{1}{2} z^\top \Sigma^{-1} z\right) dz\biggr)^{-1}.
\end{align}
In the spherically symmetric i.i.d. noise model (Model~\ref{model:isotropicIID}), with base density $p_{\mathrm{sph}}$, the corresponding components are
\begin{align}
    g_\ell^{(T)}(z) = C_\ell^{(T)}\, p_{\mathrm{sph}}(z)\,  \mathbbm{1}_{\{\langle z, x_\ell\rangle \ge T\}}, \label{eqn:componentOfMixtureModelParticlePicker-2}
\end{align}
with $C_\ell^{(T)}$ again defined by the normalization condition $\int_{\mathbb{R}^d} g_\ell^{(T)}(z)\,dz = 1$.

\begin{remark}
In both cases, we denote the resulting component densities by the same symbols $g_\ell^{(T)}$ and $C_\ell^{(T)}$. Thus, $g_\ell^{(T)}$ and $C_\ell^{(T)}$ should be understood as model-dependent: in the Gaussian setting, they are defined via the Gaussian base density, whereas in the spherically symmetric setting, they are defined using $p_{\mathrm{sph}}$. We keep the same notation to avoid introducing separate symbols for each model.
\end{remark}

In other words, $g_\ell^{(T)}$ represents the distribution from which a particle associated with the $\ell$-th template is selected, meaning its correlation with the template $x_\ell$ exceeds the threshold $T$. Clearly, $g_\ell^{(T)}(z)$ does not follow a Gaussian distribution as it is strictly positive only for $z \in \mathbb{R}^d$ values which satisfy $\langle z, x_\ell \rangle \geq T$. 
Figure~\ref{fig:appendix_B1} illustrates the probability density function of this truncated Gaussian distribution in one and two dimensions for the case $\Sigma = I_d$.

\begin{figure}[t!]
    \centering
    \includegraphics[width=0.9\linewidth]{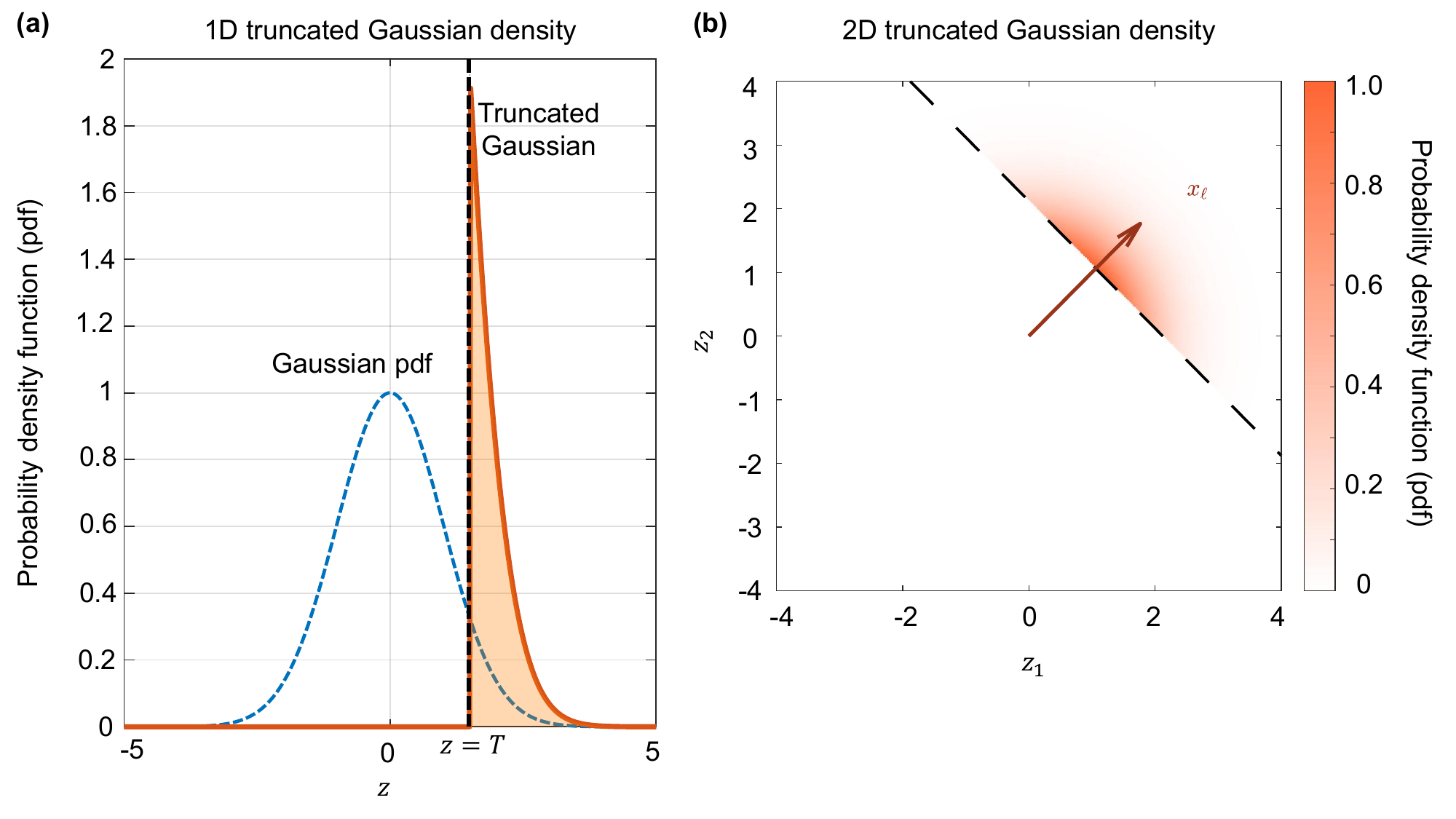}
    \caption{\textbf{Template-matching-induced truncation of a Gaussian distribution.} 
    \textbf{(a)} In one dimension, the original Gaussian density (dashed) is truncated by a hard threshold $T$, resulting in a re-normalized density $g_\ell(z)$ supported only on the region where $z \geq T$. 
    \textbf{(b)} In two dimensions, the original isotropic Gaussian is truncated to the half-space defined by $\langle z, x_\ell \rangle \geq T$, where $x_\ell$ is the template direction (shown as a red arrow) and the dashed line denotes the threshold boundary.}
    \label{fig:appendix_B1}
\end{figure}

\paragraph{Misspecification of true statistics and the GMM model.} 
It is important to recognize that the statistical assumptions underlying the particles, specifically, the GMM statistics defined in \eqref{eqn:GMMmodel}, do not accurately reflect the true statistics of the particle picking process based on template matching, as described in \eqref{eqn:mixtureModelParticlePicker} and \eqref{eqn:componentOfMixtureModelParticlePicker}. Consequently, \textit{the models are misspecified}, and this must be considered when estimating the centroids.

\paragraph{Relationship between 2D classification centers and templates.}  
The model presented here raises two central questions: 
\begin{enumerate}
    \item What is the relationship between the mean of each component $g_\ell^{(T)}$ in the true mixture model \eqref{eqn:componentOfMixtureModelParticlePicker} and the corresponding template $x_\ell$ when the number of particles $M \to \infty$? (answered in Proposition \ref{thm:prop1} and Proposition~\ref{prop:truncated-mean-isotropic})
    \item In the 2D classification process based on a GMM model, how do the templates $\{x_\ell\}_{\ell=0}^{L-1}$ relate to the GMM maximum-likelihood estimators of the means $\{\widehat{\mu}_\ell\}_{\ell=0}^{L-1}$, as specified in \eqref{eqn:maximumOfLogLikelihood}? (answered in Proposition \ref{thm:high-threshold-MLE-alignment})
\end{enumerate}

\subsection{Bias of per-template picked particles}
We begin by presenting a theoretical result that establishes the relationship between the mean of the $\ell$-th component of the mixture model, as defined in \eqref{eqn:componentOfMixtureModelParticlePicker}, and the corresponding template $x_\ell$.

\paragraph{The centers of the template matching process. } Using the definitions in Algorithm~\ref{alg:particlePickerTemplateMatching}, let $\widehat{m}_\ell^{(T)}$ be defined as:
\begin{align}
    \widehat{m}_\ell^{(T)} = \frac{1}{|\mathcal{A}_\ell^{(T)}|} \sum_{z_i \in \mathcal{A}_\ell^{(T)}} z_i, \label{eq:mu_l_def}
\end{align}
where
\begin{align}
    \mathcal{A}_\ell^{(T)} = \{z_i \in \mathcal{A} \mid \langle z_i, x_\ell \rangle \geq T\}. \label{eq:calA_ell}
\end{align}
Here, $\widehat{m}_\ell^{(T)}$ represents the centroid of the particles whose correlation with the template $x_\ell$ exceeds the threshold $T$. The key question we address is: what is the relationship between $\widehat{m}_\ell^{(T)}$ and the corresponding template $x_\ell$?

To this end, it is convenient to define, for every $\ell \in \pp{L}$, the following sets:
\begin{align}
    \mathcal{V}_\ell^{(T)} \triangleq \ppp{v \in \mathbb{R}^d: {{\langle{v, x_\ell\rangle}} \geq T } }, \label{eqn:VlDef}
\end{align}
which collect all vectors whose correlation with $x_\ell$ exceeds the threshold $T$. Note that $\{\mathcal{V}_\ell^{(T)}\}_{\ell\in[L]}$ are deterministic and are not necessarily disjoint.

We now characterize the limiting behavior of the empirical centers $\widehat{m}_\ell^{(T)}$ under two noise models: general stationary Gaussian noise and spherically symmetric noise.

\paragraph{Stationary Gaussian noise.}
We first consider the case where the candidate particles arise from a stationary Gaussian process. The next proposition shows that, under mild mixing conditions, the empirical centers align with the templates $\Sigma x_\ell$; the proof is given in Appendix~\ref{sec:proofOfStaionaryGaussianMeans}.

\begin{proposition}
\label{thm:prop1}
Let $\{y_i\}_{i=0}^{N-1}$ denote the input candidate particles to Algorithm~\ref{alg:particlePickerTemplateMatching}, drawn according to either Model~\ref{model:whiteNoiseIID} or Model~\ref{model:stationaryGaussianNoisePatches}. In both cases, each $y_i$ is marginally Gaussian with mean zero and non-singular covariance matrix $\Sigma \in \mathbb{R}^{d \times d}$, and in the stationary model the process is $\alpha$-mixing with coefficients satisfying the summability condition~\eqref{eq:alpha-mixing-summability-1d}.

Let $\{x_\ell\}_{\ell=0}^{L-1}$ denote the normalized templates, $T \in \mathbb{R}^{+}$ the template-matching threshold, and $\mathcal{A} \subset \{y_i\}_{i=0}^{N-1}$ the output set of selected particles of Algorithm~\ref{alg:particlePickerTemplateMatching}.
For each $\ell$, let $\mathcal{A}_\ell^{(T)}$ be the corresponding subset defined in~\eqref{eq:calA_ell}, and let $\widehat{m}_\ell^{(T)}$ be its empirical mean as in~\eqref{eq:mu_l_def}. Then:

\begin{enumerate}
    \item For every fixed $T < \infty$ and every $0 \le \ell \le L-1$,
    \begin{align}
        \widehat{m}_\ell^{(T)} = \frac{1}{|\mathcal{A}_\ell^{(T)}|} \sum_{y_i \in \mathcal{A}_\ell^{(T)}} y_i \xrightarrow[N\to\infty]{\mathrm{a.s.}}  m_\ell^{(T)} ,
    \end{align}
    where
    \begin{align}
         m_\ell^{(T)}  = \mathbb{E}\bigl[y_0 \,\big|\, \langle y_0, x_\ell \rangle \ge T\bigr] = \lambda_\ell(T)\, \frac{\Sigma x_\ell}{x_\ell^\top \Sigma x_\ell},
    \end{align}
    for some scalar $\lambda_\ell(T) > 0$. In particular, $ m_\ell^{(T)} $ lies in the one-dimensional subspace spanned by $\Sigma x_\ell$.

    \item As $T \to \infty$, we have
    \begin{align}
        \lim_{T \to \infty} \frac{ m_\ell^{(T)} }{T} = \frac{\Sigma x_\ell}{x_\ell^\top \Sigma x_\ell},
    \end{align}
    and hence
    \begin{align}
        \lim_{T \to \infty} \lim_{N \to \infty}    \frac{\widehat{m}_\ell^{(T)}}{T} = \frac{\Sigma x_\ell}{x_\ell^\top \Sigma x_\ell},
    \end{align}
    where the convergence holds almost surely, for every $\ell \in \{0,\dots,L-1\}$.
\end{enumerate}
\end{proposition}

\paragraph{Spherically symmetric noise.}
We next specialize to the spherically symmetric setting, where the conditional means align exactly with the templates themselves; the proof of the following proposition is given in Appendix~\ref{sec:proofOfIsotropicNoise}.

\begin{proposition}
\label{prop:truncated-mean-isotropic}
Let $\{y_i\}_{i=0}^{N-1}$ denote the input candidate particles to 
Algorithm~\ref{alg:particlePickerTemplateMatching}, drawn according to 
Model~\ref{model:isotropicIID}, distributed with common distribution $Y \in \mathbb{R}^d$. Let $\{x_\ell\}_{\ell=0}^{L-1}$ denote the normalized templates, $\|x_\ell\|_2=1$, $T\in\mathbb{R}^{+}$ the template–matching threshold, and $\mathcal{A}\subset\{y_i\}_{i=0}^{N-1}$ the output set of selected particles of Algorithm~\ref{alg:particlePickerTemplateMatching}.
For each $\ell$, let $\mathcal{A}_\ell^{(T)}$ be the corresponding subset defined in~\eqref{eq:calA_ell}, and let $\widehat{m}_\ell^{(T)}$ be its empirical mean as in~\eqref{eq:mu_l_def}. Then:
\begin{enumerate}
    \item For every fixed $T<\infty$ and every $0\le\ell\le L-1$,
    \begin{align}
        \widehat{m}_\ell^{(T)} \xrightarrow[N\to\infty]{\mathrm{a.s.}}  m_\ell^{(T)} 
        \triangleq \mathbb{E}\bigl[Y \,\big|\, \langle Y,x_\ell\rangle \ge T\bigr],
    \end{align}
    and there exists a scalar $\beta(T)\geq T$, independent of $\ell$, such that for all $\ell \in \{0,1, \ldots, L-1 \}$,
    \begin{align}
         m_\ell^{(T)} = \beta(T)\,x_\ell.
    \end{align}
    In particular, each conditional mean lies in the one-dimensional subspace 
    spanned by the corresponding template $x_\ell$.

    \item As $T\to\infty$, we have
    \begin{align}
        \lim_{T\to\infty}\frac{ m_\ell^{(T)} }{T} = x_\ell,
    \end{align}
    and hence
    \begin{align}
        \lim_{T\to\infty}\lim_{N\to\infty}\frac{\widehat{m}_\ell^{(T)}}{T} = x_\ell,
    \end{align}
    where the convergence holds almost surely, for every $\ell\in\{0,\dots,L-1\}$.
\end{enumerate}
\end{proposition}

This result shows that, under Model~\ref{model:isotropicIID}, averaging the noise observations $\{y_i\}_{i=0}^{N-1}$ whose correlation with the template $x_\ell$ exceeds the threshold $T$ converges almost surely to the template $x_\ell$, scaled by a factor $\beta(T) \ge T$. Moreover, the second part of the result implies that for asymptotically large thresholds $T$, the ratio $\beta(T) / T$ approaches $1$, so that $\beta(T)$ attains its lower bound asymptotically.

In essence, Propositions~\ref{thm:prop1} and~\ref{prop:truncated-mean-isotropic} establish that the centers of the mixture components induced by template matching resemble the templates (or filtered templates) up to a scalar factor. This is in contrast with the ``Einstein from noise'' phenomenon and the associated confirmation bias in 2D classification, where the resulting estimates are positively correlated with the initialization but do not fully reproduce it~\cite{balanov2024einstein}.

\subsubsection{Proof of proposition~\ref{thm:prop1}} \label{sec:proofOfStaionaryGaussianMeans}
We treat both noise models in a unified way. In either Model~\ref{model:whiteNoiseIID} or Model~\ref{model:stationaryGaussianNoisePatches}, each candidate particle $y_i$ is marginally Gaussian with mean zero and covariance $\Sigma$, where $\Sigma$ is non-singular. In the white-noise model the $y_i$ are i.i.d.; in the stationary model the sequence $\{y_i\}$ is strictly stationary and $\alpha$-mixing with coefficients satisfying the summability condition~\eqref{eq:alpha-mixing-summability-1d}.

Fix $\ell$ and $T$, and recall that
\begin{align}
    \mathcal{A}_\ell^{(T)} = \bigl\{y_i : \langle y_i, x_\ell \rangle \ge T\bigr\},
    \qquad
    \widehat{m}_\ell^{(T)} = \frac{1}{|\mathcal{A}_\ell^{(T)}|} \sum_{y_i \in \mathcal{A}_\ell^{(T)}} y_i.
\end{align}
We can rewrite $\widehat{m}_\ell^{(T)}$ as a ratio of empirical averages:
\begin{align}
    \widehat{m}_\ell^{(T)} &= \frac{\frac{1}{N} \sum_{i=0}^{N-1} y_i\, \mathbbm{1}_{\{\langle y_i, x_\ell \rangle \ge T\}}}{\frac{1}{N} \sum_{i=0}^{N-1} \mathbbm{1}_{\{\langle y_i, x_\ell \rangle \ge T\}}}.
    \label{eq:ratio-representation}
\end{align}

\paragraph{Step 1: Law of large numbers (i.i.d. and $\alpha$-mixing cases).}
In the white-noise model, the terms $y_i\,\mathbbm{1}_{\{\langle y_i, x_\ell \rangle \ge T\}}$ and $\mathbbm{1}_{\{\langle y_i, x_\ell \rangle \ge T\}}$ are i.i.d. and integrable, so the classical SLLN applies.

In the stationary case, both sequences $\{y_i\,\mathbbm{1}_{\{\langle y_i, x_\ell \rangle \ge T\}}\}$ and $\{\mathbbm{1}_{\{\langle y_i, x_\ell \rangle \ge T\}}\}$ are strictly stationary with $\alpha$-mixing satisfying $\sum_{r=1}^\infty \alpha(r)^{\delta/(2+\delta)} < \infty$, and they have finite moments of all orders because the underlying field is Gaussian. Hence, by the strong law of large numbers for $\alpha$-mixing sequences (see, e.g.,~\cite{bradley2005basic}), we have almost surely, as $N\to\infty$,
\begin{align}
    \frac{1}{N} \sum_{i=0}^{N-1} y_i\,\mathbbm{1}_{\{\langle y_i, x_\ell \rangle \ge T\}} &\xrightarrow[]{\mathrm{a.s.}} \mathbb{E}\bigl[y_0\,\mathbbm{1}_{\{\langle y_0, x_\ell \rangle \ge T\}}\bigr], \\
    \frac{1}{N} \sum_{i=0}^{N-1}\mathbbm{1}_{\{\langle y_i, x_\ell \rangle \ge T\}} &\xrightarrow[]{\mathrm{a.s.}} \mathbb{P}\bigl(\langle y_0, x_\ell \rangle \ge T\bigr) > 0,
\end{align}
for every finite $T$. Combining these limits with \eqref{eq:ratio-representation} and applying the continuous mapping theorem yields
\begin{align}
    \widehat{m}_\ell^{(T)} \xrightarrow[N\to\infty]{\mathrm{a.s.}}  m_\ell^{(T)}  \triangleq \frac{\mathbb{E}\bigl[y_0\,\mathbbm{1}_{\{\langle y_0, x_\ell \rangle \ge T\}}\bigr]}{\mathbb{P}\bigl(\langle y_0, x_\ell \rangle \ge T\bigr)}
    =
    \mathbb{E}\bigl[y_0 \,\big|\, \langle y_0, x_\ell \rangle \ge T\bigr].
    \label{eq:mu-ell-T-def}
\end{align}
This proves the first convergence in part (1).

\paragraph{Step 2: Conditional mean for Gaussian distribution.}
We now compute $ m_\ell^{(T)} $ explicitly for a Gaussian vector $y_0 \sim \mathcal{N}(0,\Sigma)$, with $\Sigma$ non-singular. Set
\begin{align}
    s \triangleq \langle y_0, x_\ell \rangle = x_\ell^\top y_0.
\end{align}
Then $(y_0, s)$ is jointly Gaussian with $\mathbb{E}[y_0] = 0$, $\mathbb{E}[s] = 0$, and
\begin{align}
    \mathrm{Var}(s) &= x_\ell^\top \Sigma x_\ell, \\
    \mathrm{Cov}(y_0,s) &= \mathbb{E}[y_0 s] = \mathbb{E}[y_0 y_0^\top] x_\ell = \Sigma x_\ell.
\end{align}
By the standard formula for the conditional expectation of a Gaussian vector given a linear observation~\cite{anderson1958introduction}, we have
\begin{align}
    \mathbb{E}[y_0 \mid s] = \Sigma x_\ell \,(x_\ell^\top \Sigma x_\ell)^{-1} s.
\end{align}
Therefore, by iterated conditioning,
\begin{align}
     m_\ell^{(T)} = \mathbb{E}[y_0 \mid s \ge T] = \mathbb{E}\bigl[\mathbb{E}[y_0 \mid s]\mid s \ge T\bigr] = \Sigma x_\ell \,(x_\ell^\top \Sigma x_\ell)^{-1} \mathbb{E}[s \mid s \ge T].
    \label{eq:mu-ell-T-gaussian}
\end{align}
Since $s \sim \mathcal{N}\bigl(0, x_\ell^\top \Sigma x_\ell\bigr)$, we may
write $s = \sqrt{x_\ell^\top \Sigma x_\ell}\,Z$, where $Z\sim\mathcal{N}(0,1)$.
Thus
\begin{align}
    \mathbb{E}[s \mid s \ge T] = \sqrt{x_\ell^\top \Sigma x_\ell}\, \mathbb{E}\!\left[Z \,\big|\, Z \ge\frac{T}{\sqrt{x_\ell^\top \Sigma x_\ell}}\right].
\end{align}
Combining this with~\eqref{eq:mu-ell-T-gaussian}, we obtain
\begin{align}
     m_\ell^{(T)}  = \lambda_\ell(T)\, \frac{\Sigma x_\ell}{x_\ell^\top \Sigma x_\ell},
\end{align}
where
\begin{align}
    \lambda_\ell(T) \triangleq \mathbb{E}[s \mid s \ge T] = \sqrt{x_\ell^\top \Sigma x_\ell}\, \mathbb{E}\!\left[Z \,\big|\, Z \ge \frac{T}{\sqrt{x_\ell^\top \Sigma x_\ell}}
    \right] > 0.
\end{align}
This shows that $ m_\ell^{(T)} $ lies in the one-dimensional subspace spanned by $\Sigma x_\ell$, completing part (1).

\paragraph{Step 3: Asymptotics as $T\to\infty$.}
It remains to analyze the behavior of $ m_\ell^{(T)} /T$ as $T\to\infty$. Let $Z \sim \mathcal{N}(0,1)$ with density $\varphi$ and distribution function $\Phi$. Then, by Lemma~\ref{lem:C4}
\begin{align}
    \mathbb{E}[Z \mid Z \ge t] = \frac{\int_t^\infty z \,\varphi(z)\,dz}{ \int_t^\infty \varphi(z)\,dz} = \frac{\varphi(t)}{1-\Phi(t)}.
\end{align}
and using the Mills ratio asymptotics,
\begin{align}
    \frac{\varphi(t)}{1-\Phi(t)} = t + \frac{1}{t} + O\!\left(\frac{1}{t^3}\right) \qquad (t\to\infty),
\end{align}
we obtain, in particular,
\begin{align}
    \mathbb{E}[Z \mid Z \ge t] = t + \frac{1}{t} + O\!\left(\frac{1}{t^3}\right) \sim t \qquad \text{as } t\to\infty.
\end{align}
Taking $t = T/\sqrt{x_\ell^\top \Sigma x_\ell}$ in the expression for $\lambda_\ell(T)$, we obtain
\begin{align}
    \frac{\lambda_\ell(T)}{T} = \frac{1}{T}\, \sqrt{x_\ell^\top \Sigma x_\ell}\, \mathbb{E}\!\left[Z \,\big|\, Z \ge \frac{T}{\sqrt{x_\ell^\top \Sigma x_\ell}} \right] = \frac{1}{t}\,\mathbb{E}[Z\mid Z\ge t],
\end{align}
with $t = T/\sqrt{x_\ell^\top \Sigma x_\ell} \to \infty$ as $T\to\infty$. Thus
\begin{align}
    \lim_{T\to\infty} \frac{\lambda_\ell(T)}{T} = 1.
\end{align}
Substituting into~\eqref{eq:mu-ell-T-gaussian} yields
\begin{align}
    \lim_{T \to \infty} \frac{ m_\ell^{(T)} }{T} = \frac{\Sigma x_\ell}{x_\ell^\top \Sigma x_\ell}.
\end{align}
Combining this with the almost sure convergence
$\widehat{m}_\ell^{(T)} \to  m_\ell^{(T)} $ from~\eqref{eq:mu-ell-T-def} gives
\begin{align}
    \lim_{T\to\infty} \lim_{N\to\infty} \frac{\widehat{m}_\ell^{(T)}}{T} = \frac{\Sigma x_\ell}{x_\ell^\top \Sigma x_\ell},
\end{align}
which completes the proof.

\subsubsection{Proof of proposition~\ref{prop:truncated-mean-isotropic}} \label{sec:proofOfIsotropicNoise}
\paragraph{Step 1: Law of large numbers.} 
Fix $\ell$ and $T$. As in~\eqref{eq:mu-ell-T-def}, we rewrite $\widehat{m}_\ell^{(T)}$ as a ratio of empirical averages. By the strong law of large numbers for $\alpha$-mixing sequences with a finite $(2+\delta)$-moment (see~\eqref{eq:alpha-mixing-summability-1d} and Model~\ref{model:isotropicIID}), we obtain, for each fixed $T$,
\begin{align}
    \widehat{m}_\ell^{(T)} \xrightarrow[N\to\infty]{\mathrm{a.s.}} m_\ell^{(T)} \triangleq \frac{\mathbb{E}\bigl[Y\,\mathbbm{1}_{\{\langle Y, x_\ell \rangle \ge T\}}\bigr]}{\mathbb{P}\bigl(\langle Y, x_\ell \rangle \ge T\bigr)} = \mathbb{E}\bigl[Y \,\big|\, \langle Y, x_\ell \rangle \ge T\bigr].
    \label{eq:mu-ell-T-def-2}
\end{align}
Since the one-dimensional marginal $\langle Y,x_\ell\rangle$ has a continuous density and non-compact support (by spherical symmetry and the assumptions in Model~\ref{model:isotropicIID}), we have $\mathbb{P}\bigl(\langle Y, x_\ell \rangle \ge T\bigr) > 0$ for every finite $T$, so the denominator in~\eqref{eq:mu-ell-T-def-2} is strictly positive. 

\paragraph{Step 2: Direction of the conditional mean.}
We now show that $ m_\ell^{(T)} $ is colinear with $x_\ell$. This condition can be equivalently expressed algebraically: for any vector $u \in \mathbb{R}^d$ orthogonal to $x_\ell$ (i.e., $\langle u, x_\ell \rangle = 0$), the mean $ m_\ell^{(T)} $ must also be orthogonal to $u$ i.e.,
\begin{align}
    \langle u,  m_\ell^{(T)}  \rangle = 0. \label{eqn:orthogonalIsZero}
\end{align}
To prove~\eqref{eqn:orthogonalIsZero}, we state the following Lemma.
\begin{lem} \label{lemma:1}
Let $\mathcal{P} \p{y} = p_{\mathrm{sph}}(y)$ be the distribution of the random variable $Y$ according to Model~\ref{model:isotropicIID}. Assume $u \in \mathbb{R}^d$, satisfying $\langle u, x_\ell \rangle = 0$. Then, for every $T$,
\begin{align}
    \langle u,  m_\ell^{(T)}  \rangle \, \mathbb{P} \, [\langle Y, x_\ell \rangle \geq T] = \int_{\mathbb{R}^d} \langle y, u \rangle \mathcal{P} \p{y} \mathbbm{1}_{\langle y, x_\ell \rangle \geq T} \ dy = 0. \label{eqn:SLLNsingleMixtureMean}
\end{align}
\end{lem}

In other words, Lemma~\ref{lemma:1} shows that $x_\ell$ and $ m_\ell^{(T)} $ are collinear, which is equivalent to
\begin{align}
     m_\ell^{(T)}  = \alpha_\ell^{(T)} \, x_{\ell},
\end{align}
where,
\begin{align}
    \alpha_\ell^{(T)} \triangleq \langle  m_\ell^{(T)} , x_{\ell} \rangle \in \mathbb{R}.
\end{align}
In Lemma \ref{lemma:2}, stated below, we prove that $\alpha_\ell^{(T)} \geq T$, and that $\alpha_\ell^{(T)}$ is independent of $\ell$, i.e., $\alpha_{\ell_1}^{(T)} = \alpha_{\ell_2}^{(T)}$, for every $\ell_1, \ell_2 \in \pp{L}$, which completes the proof of the first part of the proposition.

\begin{lem} \label{lemma:2}
Recall the definition of $ m_\ell^{(T)} $ in \eqref{eq:mu-ell-T-def-2}. For every $\ell_1, \ell_2 \in \pp{L}$, the following holds,
\begin{align}
    \alpha_{\ell_1}^{(T)} = \langle  m_{\ell_1}^{(T)} , x_{\ell_1} \rangle = \langle  m_{\ell_2}^{(T)} , x_{\ell_2} \rangle = \alpha_{\ell_2}^{(T)} \geq T.
\end{align}
\end{lem}
Following Lemma~\ref{lemma:2}, we denote by $\beta(T) \triangleq \alpha_\ell^{(T)}$, which is independent of $\ell$.

\paragraph{Step 3: Asymptotics as $T\to\infty$.}
From \eqref{eq:mu-ell-T-def-2} and the discussion above, we have
\begin{align}
    \beta(T) = \mathbb{E}[Z \mid Z>T],
\end{align}
where $Z$ is the common one–dimensional marginal of $\langle Y,u\rangle$ for any
unit vector $u$. By assumption in Model~\ref{model:isotropicIID},
\begin{align}
    \lim_{T\to\infty}\frac{\beta(T)}{T} = \lim_{T\to\infty}\frac{\mathbb{E}[Z\mid Z>T]}{T} = 1.
\end{align}
Therefore,
\begin{align}
    \lim_{T\to\infty}\frac{ m_\ell^{(T)} }{T} = \lim_{T\to\infty}\frac{\beta(T)}{T}\,x_\ell = x_\ell,
\end{align}
for every $\ell$. Combining this with the almost sure convergence $\widehat{m}_\ell^{(T)}\to m_\ell^{(T)} $ from \eqref{eq:mu-ell-T-def-2}, we obtain
\begin{align}
    \lim_{T\to\infty}\lim_{N\to\infty}\frac{\widehat{m}_\ell^{(T)}}{T} = x_\ell,
\end{align}
for each $0\le\ell\le L-1$. This proves part (2) and completes the proof.
It remains to prove the lemmas.

\begin{proof}[Proof of Lemma~\ref{lemma:1}]
Fix $u \in \mathbb{R}^d$, such that $\langle u, x_\ell \rangle = 0$. First, we show that for every $v \in \mathcal{V}_\ell^{(T)}$, there is a unique corresponding $v' \in \mathcal{V}_\ell^{(T)}$ such that $\langle u, v \rangle = - \langle u,v' \rangle$, and $\mathcal{P}\p{v} = \mathcal{P}\p{v'}$.
Every $v \in \mathbb{R}^d$ can be decomposed uniquely by $v = \gamma_v u + w$, where $w \in \mathbb{R}^d$ orthogonal to $u$, $\langle u, w \rangle = 0$.  
Define $v' = -\gamma_v u + w$. Clearly, the map $v \to v'$ is unique, and due to the orthogonality between $u$ and $w$, we have, $\langle u, v \rangle = \gamma_v$, and $\langle u, v' \rangle = -\gamma_v$; thus, $\langle u, v \rangle = - \langle u,v' \rangle$. In addition, $v \in \mathcal{V}_\ell^{(T)}$ if and only if $v' \in \mathcal{V}_\ell^{(T)}$, since,
\begin{align}
    v \in \mathcal{V}_\ell^{(T)}  & \Longleftrightarrow   \langle v, x_\ell \rangle \geq T \Longleftrightarrow \langle \gamma_v u + w, x_\ell \rangle \geq T  
    \nonumber \\ & \Longleftrightarrow \langle w, x_\ell \rangle \geq T \label{eqn:iffBelongToVell}
    \\ &  \Longleftrightarrow \langle -\gamma_v u + w, x_\ell \rangle \geq T \Longleftrightarrow \langle v', x_\ell \rangle \geq T \Longleftrightarrow  v' \in \mathcal{V}_\ell^{(T)}, \nonumber
\end{align}
where the transition in \eqref{eqn:iffBelongToVell} follows from the orthogonality between $u$ and $x_\ell$.

Next, we show that $\mathcal{P}\p{v} = \mathcal{P}\p{v'}$. As $\mathcal{P}(v)$ is a spherically symmetric function that depends solely on the norm of $v$. Thus, we prove below that $\norm{v} = \norm{v'}$, which would prove that $\mathcal{P} \p{v} = \mathcal{P} \p{v'}$: \begin{align}
    \norm{v}^2 = \norm{\gamma_v u + w}^2 = \gamma_v^2\norm{u}^2 + \norm{w}^2 = \norm{-\gamma_v u + w}^2 = \norm{v'}^2,
\end{align}
where in the second equality we have used the orthogonality between $u$ and $w$. Therefore, $\norm{v} = \norm{v'}$, which shows that $\mathcal{P} \p{v} = \mathcal{P} \p{v'}$.

Next, note that~\eqref{eqn:SLLNsingleMixtureMean} can be rewritten as follows,
\begin{align}
    \int_{\mathbb{R}^d} \langle y, u \rangle \mathcal{P} \p{y} & \mathbbm{1}_{\langle y, x_\ell \rangle \geq T} \ dy  = \int_{y \in \mathcal{V}_\ell^{(T)}} \langle y, u \rangle \mathcal{P} \p{y} \ dy \nonumber
    \\ & = \int_{y \in \mathcal{V}_\ell^{(T)}: \langle y, u \rangle > 0} \langle y, u \rangle \mathcal{P} \p{y} \ dy + \int_{y \in \mathcal{V}_\ell^{(T)}: \langle y, u \rangle < 0} \langle y, u \rangle \mathcal{P} \p{y} \ dy. \label{eqn:Appx_C13}
\end{align}
Thus, since for every $v \in \mathcal{V}_\ell^{(T)}$ there is a unique corresponding $v' \in \mathcal{V}_\ell^{(T)}$ such that ${\langle{v, u\rangle}} = -{\langle{v', u\rangle}}$ and $\mathcal{P}\p{v} = \mathcal{P}\p{v'}$, it is clear that the r.h.s. of \eqref{eqn:Appx_C13} is zero, which completes the proof.
\end{proof}

\begin{proof}[Proof of Lemma \ref{lemma:2}]
First, we show that for every $\ell_1, \ell_2 \in \pp{L}$, $\langle x_{\ell_1}, {m_{\ell_1}^{(T)}} \rangle = \langle x_{\ell_2}, {m_{\ell_2}^{(T)}} \rangle$, which proves that $\alpha$ is independent of $\ell$. As the templates $\ppp{x_\ell}_{\ell=0}^{L-1}$ are normalized, there is a orthogonal matrix $O_{\ell_1, \ell_2}$, such that $O_{\ell_1, \ell_2} x_{\ell_1} = x_{\ell_2}$. Then, following from the right-hand-side of \eqref{eqn:SLLNsingleMixtureMean}, we have,
\begin{align}
    \langle x_{\ell_1}, { m_{\ell_1}^{(T)} } \rangle \, \mathbb{P} \, [\langle Y, x_{\ell_1} \rangle \geq T] & = 
    \int_{\mathbb{R}^d} \langle y, x_{\ell_1} \rangle \mathcal{P} \p{y} \mathbbm{1}_{\langle y, x_{\ell_1} \rangle \geq T} \ dy \label{eqn:app_C19}
    \\ & = \int_{\mathbb{R}^d} \langle y, O_{\ell_2, \ell_1} x_{\ell_2} \rangle \mathcal{P} \p{y} \mathbbm{1}_{\langle y, O_{\ell_2, \ell_1} x_{\ell_2} \rangle \geq T} \ dy \label{eqn:app_C20}
    \\ & = \int_{\mathbb{R}^d} \langle O_{\ell_2, \ell_1} y, O_{\ell_2, \ell_1} x_{\ell_2} \rangle \mathcal{P} \p{O_{\ell_2, \ell_1} y} \mathbbm{1}_{\langle O_{\ell_2, \ell_1} y, O_{\ell_2, \ell_1} x_{\ell_2} \rangle \geq T} \ dy \label{eqn:app_C21}
    \\ & = \int_{\mathbb{R}^d} \langle y, x_{\ell_2} \rangle \mathcal{P} \p{O_{\ell_2, \ell_1} y} \mathbbm{1}_{\langle y, x_{\ell_2} \rangle \geq T} \ dy \label{eqn:app_C22}
    \\ & = \int_{\mathbb{R}^d} \langle y, x_{\ell_2} \rangle \mathcal{P} \p{y} \mathbbm{1}_{\langle y, x_{\ell_2} \rangle \geq T} \ dy \label{eqn:app_C23}
    \\ & = \langle x_{\ell_2}, { m_{\ell_2}^{(T)} } \rangle \, \mathbb{P} \, [\langle Y, x_{\ell_2} \rangle \geq T], \label{eqn:app_C24}
\end{align}
where \eqref{eqn:app_C19} follows from the definition of $ m_\ell^{(T)} $; \eqref{eqn:app_C20} follows from the orthonormal operator definition of $O_{\ell_2, \ell_1}$; \eqref{eqn:app_C21} follows from substitution of variables $y \to O_{\ell_2, {\ell_1}} y$; \eqref{eqn:app_C22} follows from the property of orthonormal operator $O$, which satisfies, $\langle O  a, O b \rangle = \langle a, b \rangle$ for every $a,b \in \mathbb{R}^d$; \eqref{eqn:app_C23} follows from the spherical symmetry property of $\mathcal{P} \p{y} = \mathcal{P} \p{O y}$, that is for every orthonormal operator $O$ and for every $y \in \mathbb{R}^d$, $\norm{O y} = \norm{y}$; and \eqref{eqn:app_C24} follows from the definition of $ m_\ell^{(T)} $. In addition, by  the spherical symmetry of $Y$, we have,
\begin{align}
    \mathbb{P} \, [\langle Y, x_{\ell_1} \rangle \geq T] = \mathbb{P} \, [\langle Y, x_{\ell_2} \rangle \geq T]
\end{align}

Next, we show that $\alpha_\ell \geq T$. Similar to Lemma~\ref{lemma:1}, for every $y \in \mathcal{V}_\ell^{(T)}$, we can decompose it into $y = \beta_y x_\ell + \gamma_y w_y$, where $\langle w_y, x_\ell  \rangle = 0$, and $\beta_y \geq T$, by the definition of the domain of the integral. Thus, we can decompose the integral in \eqref{eqn:SLLNsingleMixtureMean} as follows:
\begin{align}
    \int_{\mathbb{R}^d} \langle y, x_\ell \rangle \mathcal{P} \p{y} \mathbbm{1}_{\langle y, x_\ell \rangle \geq T} \ dy & = \int_{\mathbb{R}^d} \langle \beta_y x_\ell, x_\ell \rangle \mathcal{P} \p{y} \mathbbm{1}_{\langle y, x_\ell \rangle \geq T} \ dy \nonumber
    \\ &  \quad + \int_{\mathbb{R}^d} \langle \gamma_y w_y, x_\ell \rangle \mathcal{P} \p{y} \mathbbm{1}_{\langle y, x_\ell \rangle \geq T} \ dy \label{eqn:app_D21}
\end{align}

By Lemma~\ref{lemma:1}, there exists $y' = \beta_y x_\ell - \gamma_y w_y$, such that $\norm{y'} = \norm{y}$, and $\langle y', x_\ell \rangle \geq T$. Thus, as $\mathcal{P}$ is spherically symmetric, that is, $\mathcal{P}(y) = \mathcal{P}(y')$ for $\norm{y} = \norm{y'}$, the last term in \eqref{eqn:app_D21} vanishes:
\begin{align}
    \int_{\mathbb{R}^d} \langle \gamma_y w_y, x_\ell \rangle \mathcal{P} \p{y} \mathbbm{1}_{\langle y, x_\ell \rangle \geq T} \ dy = 0. \label{eqn:app_D22}
\end{align}
Substituting \eqref{eqn:app_D22} into \eqref{eqn:app_D21} results,
\begin{align}
    \int_{\mathbb{R}^d} \langle y, x_\ell \rangle \mathcal{P} \p{y} \mathbbm{1}_{\langle y, x_\ell \rangle \geq T} \ dy & = \int_{\mathbb{R}^d} \langle \beta_y x_\ell, x_\ell \rangle \mathcal{P} \p{y} \mathbbm{1}_{\langle y, x_\ell \rangle \geq T} \ dy \nonumber
     \\ & = \int_{\mathbb{R}^d} \beta_y \|x_\ell\|^2 \mathcal{P} \p{y} \mathbbm{1}_{\langle y, x_\ell \rangle \geq T} \ dy \nonumber
     \\ & = \int_{\mathbb{R}^d} \beta_y  \mathcal{P} \p{y} \mathbbm{1}_{\langle y, x_\ell \rangle \geq T} \ dy,
     \label{eqn:app_D23}
\end{align}
where we have used $\|x_\ell\| = 1$. As $\beta_y \geq T$ for every $y \in \mathcal{V}_\ell^{(T)}$, we have,
\begin{align}
    \alpha_\ell = \frac{\int_{\mathbb{R}^d} \beta_y  \mathcal{P} \p{y} \mathbbm{1}_{\langle y, x_\ell \rangle \geq T} \ dy}{\int_{\mathbb{R}^d} \mathcal{P} \p{y} \mathbbm{1}_{\langle y, x_\ell \rangle \geq T} \ dy} \geq T, \label{eqn:app_D24_1}
\end{align}
which completes the proof of the lemma.
\end{proof}

\subsection{GMM maximum likelihood on template-selected particles} \label{subsec:GMM-likelhood-template-selected}
In the previous section, we showed that the particle-picking process induces a mixture model with $L$ components, where $L$ is the number of templates. The corresponding component means $\{m_\ell^{(T)}\}_{\ell=0}^{L-1}$ reproduce the templates (up to a scaling factor that depends on the threshold $T$ and the noise model). However, this does not automatically imply that the output of the maximum-likelihood estimation for a GMM, recovers these centroids. The key difficulty is that the GMM is only a \emph{postulated} model: the true distribution of the picked particles is not Gaussian.

In general, under model correctness, maximum-likelihood estimators converge to the true parameter values~\cite{newey1994large}. In particular, if the extracted particles were in fact drawn from a GMM with component means $\{m_\ell^{(T)}\}$, then the maximum-likelihood estimators $\{\widehat\mu_\ell\}_{\ell=0}^{L-1}$ in~\eqref{eqn:maximumOfLogLikelihoodMain} would converge to $\{m_\ell^{(T)}\}$. In our setting, however, the postulated model is an $L$-component GMM, while the \emph{actual} statistics of the picked particles are described by the mixture in~\eqref{eqn:mixtureModelParticlePicker}-\eqref{eqn:componentOfMixtureModelParticlePicker}, which is non-Gaussian. We must therefore analyze the behavior of the GMM maximum-likelihood estimator under model misspecification.

The connection between Proposition~\ref{thm:prop1} and the GMM mean estimators is illustrated in Figure~\ref{fig:appendix_C1}. In the labeled setting of Proposition~\ref{thm:prop1}, where the particle-template assignments are known, the empirical means converge directly to scaled versions of the templates. In contrast, in the unlabeled setting underlying GMM-based mean estimation, the component means converge to the templates only in the high-threshold regime.

\begin{figure}[t!]
    \centering
    \includegraphics[width=0.9\linewidth]{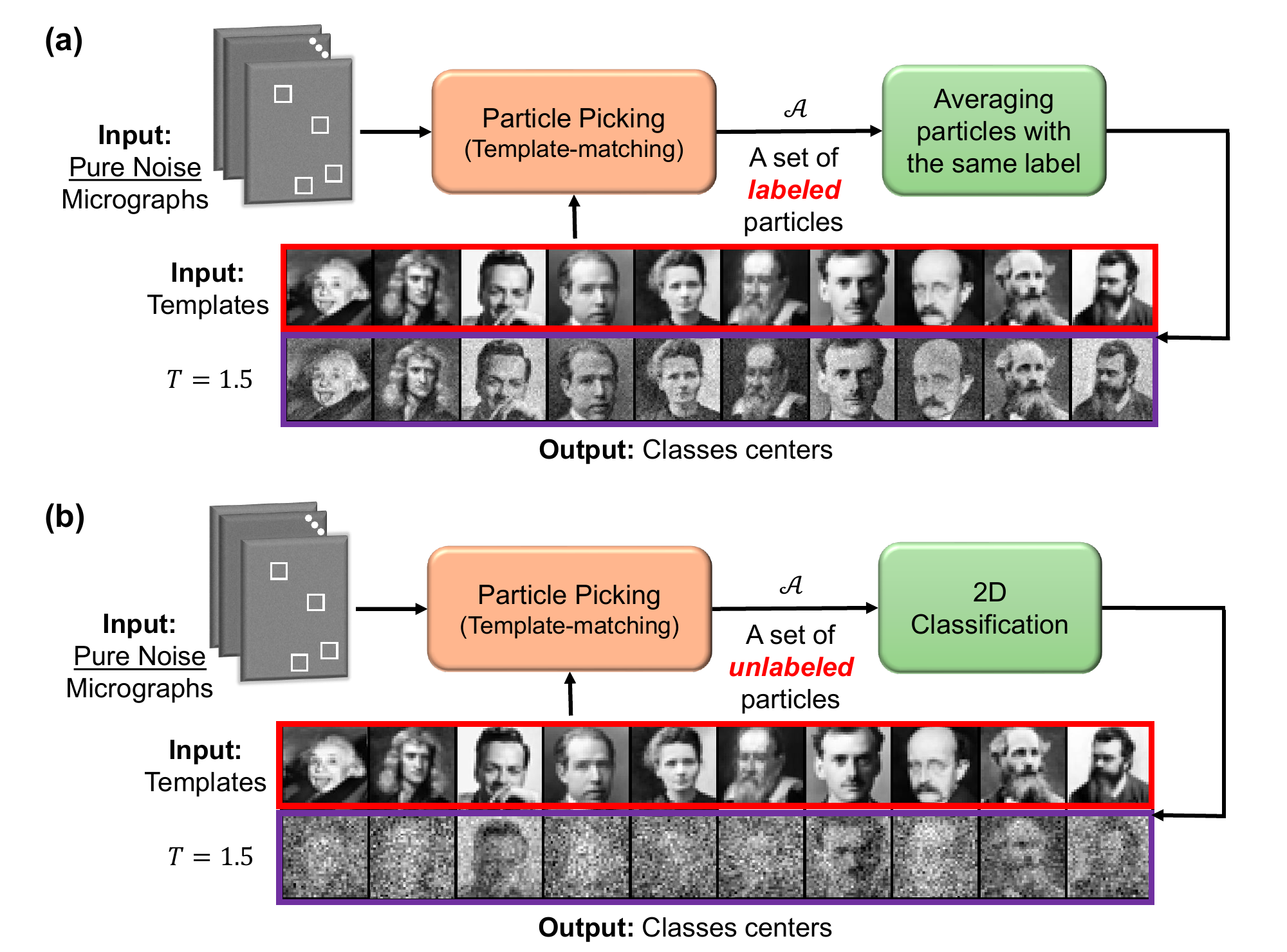}
    \caption{\textbf{Illustration of the relationship between the means for labeled particles (Proposition~\ref{thm:prop1}) and unlabeled particles (Proposition~\ref{thm:high-threshold-MLE-alignment}).} Starting from pure-noise observations, template-based particle picking selects highly correlated candidates. \textbf{(a)} When the particles are labeled, that is, each selected particle is assigned to a known template, Proposition~\ref{thm:prop1} characterizes the asymptotic behavior of the mean of particles within each label. These means converge to the corresponding templates up to a scaling factor,    and this convergence holds for any threshold $T$.
    \textbf{(b)} When the particles are unlabeled, i.e., the problem involves latent variables, the 2D classification step treats the data as a GMM. Proposition~\ref{thm:high-threshold-MLE-alignment} shows that as the threshold increases, the estimated class centers increasingly align with the original templates (see Figure~\ref{fig:3}). However, at low thresholds, this property does not necessarily occur as can be seen in the example. This simulation uses $M = 10^5$ particles collected during the particle-picking process and image sizes of $d = 50 \times 50$. The data generation model for this experiment is the same as presented in Model~\ref{model:whiteNoiseIID}.}
    \label{fig:appendix_C1}
\end{figure}

Recall the postulated GMM model from~\eqref{eqn:GMMmodelExcplicitMain}
\begin{align}
    f\bigl(z; \{\mu_\ell\}_{\ell=0}^{L-1}\bigr) = \sum_{\ell=0}^{L-1} w_\ell f_\ell(z) = \sum_{\ell=0}^{L-1} w_\ell \,\mathcal{N}\bigl(z;\mu_\ell,\Sigma_{\mathrm{GMM}}\bigr),
    \label{eqn:GMMmodelExplicit}
\end{align}
and the true mixture model of the extracted particles
\begin{align}
    g\bigl(z; \{m_\ell^{(T)}\}_{\ell=0}^{L-1}\bigr) = \sum_{\ell=0}^{L-1} \pi_\ell\, g_\ell^{(T)}(z),
    \label{eqn:trueMixtureStatistics-updated}
\end{align}
where $g_\ell^{(T)}$ is defined by~\eqref{eqn:componentOfMixtureModelParticlePicker} in the Gaussian case and by~\eqref{eqn:componentOfMixtureModelParticlePicker-2} in the spherical symmetric case. The vectors $m_\ell^{(T)} \in \mathbb{R}^d$ are the component means of the true mixture,
\begin{align}
    m_\ell^{(T)} \triangleq \mathbb{E}_{g_\ell^{(T)}}[Z] = \int_{\mathbb{R}^d} z\,g_\ell^{(T)}(z)\,dz. \label{eqn:mixture-models-means}
\end{align}
By Proposition~\ref{thm:prop1} and Proposition~\ref{prop:truncated-mean-isotropic},
\begin{align}
    \lim_{T\to\infty} \frac{m_\ell^{(T)}}{T} = v_\ell \neq 0,
\end{align}
where
\begin{align}
    v_\ell =
    \begin{cases}
        x_\ell, & \text{under Model~\ref{model:isotropicIID}},\\[0.1cm]
        \dfrac{\Sigma x_\ell}{x_\ell^\top \Sigma x_\ell}, & \text{under } y_i\sim \mathcal{N}(0,\Sigma).
    \end{cases} \label{eqn:v_ell_def}
\end{align}
Moreover, for every $\varepsilon>0$,
\begin{align}
    \mathbb{P}\!\left( \left\| \frac{Y}{T} - v_\ell \right\|_2 > \varepsilon\,\Bigg|\, \langle x_\ell,Y\rangle \ge T \right) \xrightarrow[T\to\infty]{} 0,
    \label{eq:component-concentration-appendix-updated}
\end{align}
which follows from~\eqref{eq:component-concentration} for the spherically symmetric model and can be verified directly in the Gaussian case. Intuitively, conditional on exceeding a large threshold in direction $x_\ell$, the picks concentrate in a narrow cone around the direction $v_\ell$ with radius $O(T)$.

The output of Algorithm~\ref{alg:particlePickerTemplateMatching} is the set of extracted particles $ \mathcal{A} = \{z_i\}_{i=0}^{M-1}\subset\mathbb{R}^d$, which we model as a stationary sequence with marginal density $g$. The GMM maximum-likelihood estimators of the component means are defined by
\begin{align}
    \{\widehat\mu_\ell\}_{\ell=0}^{L-1} \triangleq \argmax_{\{\mu_\ell\}_{\ell=0}^{L-1}} \prod_{i=0}^{M-1} f\bigl(z_i;\{\mu_\ell\}_{\ell=0}^{L-1}\bigr).
    \label{eqn:maximumLikelihoodTerm-updated}
\end{align}
For fixed $T<\infty$, it is convenient to introduce the empirical average log-likelihood
\begin{align}
    \mathcal{L}_{M,T}(\{\mu_\ell\}) \triangleq \frac{1}{M}\sum_{i=0}^{M-1} \log f\bigl(z_i;\{\mu_\ell\}\bigr),
    \label{eq:def-empirical-loglik-updated}
\end{align}
and the corresponding population objective
\begin{align}
    \mathcal{H}_T(\{\mu_\ell\}) \triangleq \mathbb{E}_{Z\sim g(\cdot;\{m_\ell^{(T)}\})} \bigl[\log f(Z;\{\mu_\ell\})\bigr].
    \label{eq:def-population-loglik-updated}
\end{align}

The next lemma is the standard misspecified maximum-likelihood estimation fact: the GMM maximum-likelihood estimator converges to the KL minimizer between $g$ and the GMM family $\{f(\cdot;\{\mu_\ell\})\}$.

\begin{lem}[maximum-likelihood estimator under model misspecification]
\label{lemma:KL-MLE}
Let $\{\widehat\mu_\ell\}_{\ell=0}^{L-1}$ be the maximum-likelihood estimators defined in~\eqref{eqn:maximumLikelihoodTerm-updated} for the postulated GMM model with density $f$ in~\eqref{eqn:GMMmodelExplicit}, and let $g$ be the true density of the extracted particles, given by~\eqref{eqn:trueMixtureStatistics-updated}. Assume that the underlying process $\{z_i\}$ is strictly stationary and $\alpha$-mixing with summable coefficients~\eqref{eq:alpha-mixing-summability-1d}. Then, as $M\to\infty$,
\begin{align}
    \{\widehat\mu_\ell\}_{\ell=0}^{L-1} \xrightarrow[\;M\to\infty\;]{\mathrm{a.s.}} \argmin_{\{\mu_\ell\}_{\ell=0}^{L-1}} D_{\mathrm{KL}}\!\Bigl(g(\cdot;\{m_\ell^{(T)}\}) \,\Big\|\, f(\cdot;\{\mu_\ell\})\Bigr),
    \label{eqn:KLdiv-updated}
\end{align}
where $D_{\mathrm{KL}}$ denotes the Kullback–Leibler divergence. In particular, any sequence of maximizers of $\mathcal{L}_{M,T}$ converges almost surely to a maximizer of $\mathcal{H}_T$.
\end{lem}

Lemma~\ref{lemma:KL-MLE} shows that, for each fixed threshold $T$, the GMM maximum-likelihood estimator finds in the limit $M \to \infty$ the Gaussian mixture that is closest in KL divergence to the true mixture $g(\cdot;\{m_\ell^{(T)}\})$. We now analyze this KL minimizer in the high-threshold regime.

\begin{proposition}
[High-threshold alignment of GMM centers]
\label{thm:high-threshold-MLE-alignment}
For each $(N,T)$, let $\{z_i\}_{i=0}^{M-1}$ be the extracted particles, modeled as a stationary sequence with marginal density~\eqref{eqn:trueMixtureStatistics-updated}, where $M = M(N,T)$ is the number of picks produced by Algorithm~\ref{alg:particlePickerTemplateMatching}. Recall the definition of $v_\ell$ from~\eqref{eqn:v_ell_def}. Let
\begin{align}
    \{\widehat\mu_\ell(N,T)\}_{\ell=0}^{L-1} \in \argmax_{\{\mu_\ell\}} \prod_{i=0}^{M-1} f\bigl(z_i;\{\mu_\ell\}\bigr)
    \label{eq:def-MLE-centers-updated}
\end{align}
denote the corresponding GMM maximum-likelihood estimators of the component means. Assume that the underlying process is $\alpha$-mixing with summable coefficients~\eqref{eq:alpha-mixing-summability-1d}. Then there exists a permutation $\pi$ of $\{0,\dots,L-1\}$ such that
\begin{align}    
    \lim_{T\to\infty} \lim_{N\to\infty} \frac{\widehat\mu_{\pi(\ell)}(N,T)}{T} = v_\ell, \qquad 0\le\ell\le L-1,
    \label{eq:double-limit-centers-updated}
\end{align}
where the convergence holds in probability.
\end{proposition}

\begin{proof} [Proof of Proposition~\ref{thm:high-threshold-MLE-alignment}] We prove in steps.

\paragraph{Step 1: Population limit for fixed $T$.}
By Lemma~\ref{lemma:KL-MLE} and the assumption $M(N,T)\to\infty$ as $N\to\infty$, any sequence of GMM maximum-likelihood estimators $\{\widehat\mu_\ell(N,T)\}$ converges almost surely, for each fixed $T$, to a maximizer $\{\mu_\ell^\star(T)\}$ of the population objective $\mathcal{H}_T$, i.e.
\begin{align}
    \{\mu_\ell^\star(T)\} \in \argmax_{\{\mu_\ell\}} \mathcal{H}_T(\{\mu_\ell\}).
    \label{eq:mu-star-maximizer-updated}
\end{align}

\paragraph{Step 2: High-threshold structure of the population objective.}
We next show that there exists a permutation $\pi$, such that the following holds as $T \to \infty$:
\begin{align}
    \max_{0\le\ell\le L-1} \|\mu_{\pi(\ell)}^\star(T) -m_\ell^{(T)}\|_2 = o(T).
\end{align}
Using the mixture representation~\eqref{eqn:trueMixtureStatistics-updated}, we can decompose
\begin{align}
    \mathcal{H}_T(\{\mu_\ell\}) = \sum_{\ell=0}^{L-1} \pi_\ell\, \mathbb{E}_{Z\sim g_\ell^{(T)}} \bigl[\log f(Z;\{\mu_\ell\})\bigr].
    \label{eq:HT-decomposition-updated}
\end{align}
By~\eqref{eq:component-concentration-appendix-updated}, if $Z\sim g_\ell^{(T)}$ then
\begin{align}
    \frac{Z}{T} \xrightarrow[T\to\infty]{\mathbb{P}} v_\ell, \qquad \frac{Z - m_\ell^{(T)}}{T} \xrightarrow[T\to\infty]{\mathbb{P}} 0.
    \label{eq:Z-concentration-updated}
\end{align}
The log-density of the $k$-th Gaussian component of $f$ satisfies
\begin{align}
    \log \mathcal{N}(Z;\mu_k,\Sigma_{\mathrm{GMM}}) = -\frac{1}{2} (Z-\mu_k)^\top \Sigma_{\mathrm{GMM}}^{-1}(Z-\mu_k) + C,
\end{align}
for a constant $C$ independent of $Z$ and $\mu_k$. 
Expanding the quadratic form at $Z = m_\ell^{(T)} + R_T$ with $R_T/T\to 0$ in probability, we obtain
\begin{align}
    (Z-\mu_k)^\top \Sigma_{\mathrm{GMM}}^{-1}(Z-\mu_k) &= (m_\ell^{(T)}-\mu_k)^\top\Sigma_{\mathrm{GMM}}^{-1}(m_\ell^{(T)}-\mu_k)
    \nonumber
    \\ &\quad
    + 2 R_T^\top\Sigma_{\mathrm{GMM}}^{-1}(m_\ell^{(T)}-\mu_k) + R_T^\top\Sigma_{\mathrm{GMM}}^{-1}R_T.
\end{align}
Since the cross term is $O_{\mathbb{P}}(\|R_T\|_2 T)=o_{\mathbb{P}}(T^2)$ and
the remainder term is $O_{\mathbb{P}}(\|R_T\|_2^2)=o_{\mathbb{P}}(T^2)$, we have, 
\begin{align}
    (Z-\mu_k)^\top \Sigma_{\mathrm{GMM}}^{-1}(Z-\mu_k) = (m_\ell^{(T)}-\mu_k)^\top\Sigma_{\mathrm{GMM}}^{-1}(m_\ell^{(T)}-\mu_k) + o_{\mathbb{P}}(T^2),
\end{align}
uniformly over configurations with $\|m_\ell^{(T)}-\mu_k\|_2=O(T)$. Consequently, the leading contribution to $\mathbb{E}_{Z\sim g_\ell^{(T)}}[\log f(Z;\{\mu_\ell\})]$ is of order $-T^2$ and is minimized (i.e., the log-likelihood is maximized) by choosing at least one center $\mu_k$ close to $m_\ell^{(T)}$ on the scale $T$. Thus, for $T$ large enough, any maximizer $\{\mu_\ell^\star(T)\}$ of $\mathcal{H}_T$ must satisfy
\begin{align}
    \min_k \|\mu_k^\star(T)-m_\ell^{(T)}\|_2 = o(T), \qquad 0\le\ell\le L-1.
\end{align}
Since the directions $\{v_\ell\}$ are distinct, the high-threshold cones around $\{m_\ell^{(T)}\}$ become disjoint on the scale $T$, and we can index the centers so that there exists a permutation $\pi$ with
\begin{align}
    \max_{0\le\ell\le L-1} \|\mu_{\pi(\ell)}^\star(T) -m_\ell^{(T)}\|_2 = o(T).
    \label{eq:mu-star-close-to-m-updated}
\end{align}
Dividing~\eqref{eq:mu-star-close-to-m-updated} by $T$ and using $m_\ell^{(T)}/T \to v_\ell$ yields
\begin{align}
    \max_{0 \le \ell \le L-1} \left\| \frac{\mu_{\pi(\ell)}^\star(T)}{T} - v_\ell \right\|_2 \xrightarrow[T\to\infty]{} 0.
\end{align}
Thus, we obtain for every $0\le\ell\le L-1$,
\begin{align}
    \frac{\mu_{\pi(\ell)}^\star(T)}{T} \xrightarrow[T\to\infty]{} v_\ell.
    \label{eq:pop-center-limit-updated}
\end{align}

\paragraph{Step 3: Combining the limits.}
Fix $\ell$. For each $T$, ~\eqref{eq:mu-star-maximizer-updated} and Lemma~\ref{lemma:KL-MLE} imply that
\begin{align}
    \widehat\mu_{\pi(\ell)}(N,T) \xrightarrow[N\to\infty]{\mathrm{a.s.}} \mu_{\pi(\ell)}^\star(T),
\end{align}
while~\eqref{eq:pop-center-limit-updated} gives
\begin{align}
    \frac{\mu_{\pi(\ell)}^\star(T)}{T} \xrightarrow[T\to\infty]{} v_\ell.
\end{align}
Combining these yields the double limit
\begin{align}
    \lim_{T\to\infty}\;\lim_{N\to\infty} \frac{\widehat\mu_{\pi(\ell)}(N,T)}{T} = v_\ell,
\end{align}
with convergence in probability, which is~\eqref{eq:double-limit-centers-updated}.
\end{proof}

\subsection{Proof of Theorem~\ref{thm:classesCentersVersusTemplatesInformal} and Theorem~\ref{thm:classesCentersVersusTemplatesInformalStationary}}
\label{sec:proofOfMainThms}

The proof is a direct consequence of Proposition~\ref{thm:prop1}, Proposition~\ref{prop:truncated-mean-isotropic} and Proposition~\ref{thm:high-threshold-MLE-alignment}.
For each $(N,T)$, let $\{\widehat\mu_\ell(N,T)\}_{\ell=0}^{L-1}$ denote the GMM maximum-likelihood estimators of the component means based on the extracted set $\mathcal{A}$, as in~\eqref{eq:def-MLE-centers-updated}.
By Proposition~\ref{thm:prop1} and Proposition~\ref{prop:truncated-mean-isotropic}, the means $\{m_\ell^{(T)}\}$ of the true mixture components satisfy
\begin{align}
    \lim_{T\to\infty} \frac{m_\ell^{(T)}}{T} = v_\ell \neq 0,
\end{align}
with
\begin{align}
    v_\ell =
    \begin{cases}
    x_\ell, & \text{under Model~\ref{model:isotropicIID}},\\[0.1cm]
    \dfrac{\Sigma x_\ell}{x_\ell^\top \Sigma x_\ell}, & \text{under } y_i\sim \mathcal{N}(0,\Sigma).
    \end{cases}
\end{align}
In the spherically symmetric setting of the theorem (Model~\ref{model:whiteNoiseIID}  with $\Sigma \propto I_d$ or Model~\ref{model:isotropicIID}), we have $v_\ell = x_\ell$ for all $\ell$.

By Proposition~\ref{thm:high-threshold-MLE-alignment}, that there exists a permutation $\pi$ of $\{0,\dots,L-1\}$ such that
\begin{align}
    \lim_{T\to\infty} \lim_{N\to\infty}
    \frac{\widehat\mu_{\pi(\ell)}(N,T)}{T} = v_\ell,
    \qquad 0\le \ell\le L-1,
\end{align}
where the convergence holds in probability.
This proves the two theorems for the spherically symmetric case with $v_\ell = x_\ell$ and for the stationary Gaussian case with $v_\ell = \frac{\Sigma x_\ell}{x_\ell^\top \Sigma x_\ell}$.

\subsection{Proof of Proposition~\ref{prop:finite-sample-MSE}}  \label{sec:app_sampleComplexity}
Fix $L \ge 2$ and templates $\{x_\ell\}_{\ell=0}^{L-1}\subset\mathbb R^d$ with $x_{\ell_1}\neq x_{\ell_2}$ for $\ell_1\neq\ell_2$.
Let $Y\sim\mathcal{N}(0,\Sigma)$ with $\Sigma\succ 0$, and consider the template--matching selection rule at threshold $T>0$ as in Algorithm~\ref{alg:particlePickerTemplateMatching}. Recall the definition of $g_T(z)$ from~\eqref{eqn:componentOfMixtureModelParticlePicker}, and recall the postulated Gaussian mixture model $f(z;\{\mu_k\})$ from~\eqref{eqn:GMMmodelExcplicitMain}.
Recall the definition of $\{\mu_k^\star(T)\}$ from~\eqref{eq:mu-star-maximizer-updated}
\begin{align}
    \{\mu_k^\star(T)\} \in \argmin_{\{\mu_k\}} D_{\mathrm{KL}}\bigl(g_T \,\|\, f(\cdot;\{\mu_k\})\bigr).
    \label{eq:def-pseudotrue-means}
\end{align}
and recall that $\widehat{\mu}_\ell(\mathcal{A})$ denotes the mixture-model maximum-likelihood estimator of the $\ell$-th class mean obtained from the extracted set of patches $\mathcal{A}$.
Set
\begin{align}
    v_\ell \triangleq \frac{\Sigma x_\ell}{x_\ell^\top\Sigma x_\ell}, \qquad \sigma_\ell^2 \triangleq x_\ell^\top\Sigma x_\ell.
    \label{eq:def-v-ell}
\end{align}

We first establish auxiliary lemmas and then prove the proposition. Lemma~\ref{lem:high-threshold-bias-Gaussian} is proved in the end of the end of the subsection.
\begin{lem}
\label{lem:high-threshold-bias-Gaussian}
There exist constants $T_0>0$ and $C_{\mathrm{bias}}>0$ (depending only on $\Sigma$, $\{x_\ell\}$, and $\{w_k\}$) and a permutation $\pi$ of $\{0,\dots,L-1\}$ such that, for all $T\ge T_0$ and all $0\le\ell\le L-1$,
\begin{align}
    \big\|\mu_{\pi(\ell)}^\star(T) - T\,v_\ell\big\|_2 \le \frac{C_{\mathrm{bias}}}{T}.
    \label{eq:mu-star-high-threshold-bias}
\end{align}
\end{lem}

\begin{lem}
\label{lem:variance-bound-GMM}
Under the assumption of Proposition~\ref{prop:finite-sample-MSE} There exist constants $M_0 \ge 1$ and $C_{\mathrm{var}}>0$, independent of $d,M,T$, such that for all $M\ge M_0$ and $T\ge T_0$,
\begin{align}
    \mathbb{E}\big\|\widehat{\mu}_\ell(\mathcal{A}) - \mu_\ell^\star(T)\big\|_2^2 \le \frac{C_{\mathrm{var}}\,d}{M}, \qquad 0\le \ell \le L-1. \label{eq:lem-MSE-mu-variance-bound}
\end{align}
\end{lem}

Lemma~\ref{lem:variance-bound-GMM} is a standard result in the literature on well-separated GMMs; see, for example,~\cite{newey1994large, moitra2010settling, dasgupta1999learning}. In our setting, the high-threshold regime $T \to \infty$ ensures that the mixture components are well separated, so these results apply for sufficiently large $T$.

\begin{proof}[Proof of Proposition~\ref{prop:finite-sample-MSE}]
Fix $\ell\in\{0,\dots,L-1\}$ and $T\ge T_0$, and let $\pi$ be the permutation from Lemma~\ref{lem:high-threshold-bias-Gaussian}.
We decompose
\begin{align}
    \widehat{\mu}_{\pi(\ell)}(\mathcal{A}) - T v_\ell = \big(\widehat{\mu}_{\pi(\ell)}(\mathcal{A}) - \mu_{\pi(\ell)}^\star(T)\big) + \big(\mu_{\pi(\ell)}^\star(T) - T v_\ell\big).
    \label{eq:finite-MSE-vector-decomp}
\end{align}
Using $\|a+b\|_2^2 \le 2\|a\|_2^2 + 2\|b\|_2^2$ and taking expectations in
\eqref{eq:finite-MSE-vector-decomp} gives
\begin{align}
    \mathbb{E}\Big[\big\|\widehat{\mu}_{\pi(\ell)}(\mathcal{A}) - T v_\ell\big\|_2^2\Big] \le 2\,\mathbb{E}\Big[\big\|\widehat{\mu}_{\pi(\ell)}(\mathcal{A}) - \mu_{\pi(\ell)}^\star(T)\big\|_2^2\Big] + 2\,\big\|\mu_{\pi(\ell)}^\star(T) - T v_\ell\big\|_2^2.
    \label{eq:finite-MSE-decomposition}
\end{align}
By the variance bound of Lemma~\ref{lem:variance-bound-GMM},
\begin{align}
    \mathbb{E}\Big[\big\|\widehat{\mu}_{\pi(\ell)}(\mathcal{A}) - \mu_{\pi(\ell)}^\star(T)\big\|_2^2\Big] \le \frac{C_{\mathrm{var}}\,d}{M}, \qquad M\ge M_0,\; T\ge T_0,
    \label{eq:finite-MSE-variance-part}
\end{align}
while the high-threshold bias control of
Lemma~\ref{lem:high-threshold-bias-Gaussian} yields
\begin{align}
    \big\|\mu_{\pi(\ell)}^\star(T) - T v_\ell\big\|_2^2 \le \frac{C_{\mathrm{bias}}^2}{T^2}, \qquad T\ge T_0.
    \label{eq:finite-MSE-bias-part}
\end{align}
Substituting \eqref{eq:finite-MSE-variance-part} and \eqref{eq:finite-MSE-bias-part} into \eqref{eq:finite-MSE-decomposition}, we obtain
\begin{align}
    \mathbb{E}\Big[\big\|\widehat{\mu}_{\pi(\ell)}(\mathcal{A}) - T v_\ell\big\|_2^2\Big] &\le 2\,\frac{C_{\mathrm{var}}\,d}{M} + 2\,\frac{C_{\mathrm{bias}}^2}{T^2} \nonumber
    \\
    & \le C_1\,\frac{d}{M} + \frac{C_2}{T^2},
    \label{eq:finite-MSE-final}
\end{align}
for all $M\ge M_0$ and $T\ge T_0$, where
$C_1 \triangleq 2C_{\mathrm{var}}$ and $C_2 \triangleq 2C_{\mathrm{bias}}^2$ are
independent of $d,M,T$, completing the proof.
\end{proof}

\begin{proof}[Proof of Lemma~\ref{lem:high-threshold-bias-Gaussian}]
Fix $\ell$ and consider the scalar projection $U_\ell \triangleq \langle x_\ell, Y\rangle$. Then $U_\ell\sim\mathcal{N}(0,\sigma_\ell^2)$ with $\sigma_\ell^2$ as in~\eqref{eq:def-v-ell}.
For the one--dimensional Gaussian tail, the conditional mean admits the standard Mills--ratio expansion (Lemma~\ref{lem:C4}):
\begin{align}
    \mathbb{E}[U_\ell \mid U_\ell \ge T] = \sigma_\ell\,\frac{\varphi(a_T)}{1-\Phi(a_T)}, \qquad a_T \triangleq \frac{T}{\sigma_\ell},
    \label{eq:trunc-1d-mean-def}
\end{align}
where $\varphi$ and $\Phi$ denote the standard normal pdf and cdf, respectively.
As $T \to \infty$,
\begin{align}
    \frac{\varphi(a_T)}{1-\Phi(a_T)} = a_T + \frac{1}{a_T} + O\!\left(\frac{1}{a_T^3}\right), 
    \label{eq:Mills-ratio}
\end{align}
so that
\begin{align}
    \mathbb{E}[U_\ell \mid U_\ell \ge T] &= \sigma_\ell\left(a_T + \frac{1}{a_T} + O\!\left(\frac{1}{a_T^3}\right)\right) \nonumber\\
    &= T + \frac{\sigma_\ell^2}{T} + O\!\left(\frac{1}{T^3}\right).
    \label{eq:trunc-1d-mean-asymptotic}
\end{align}

Next, decompose $Y$ along $v_\ell$ and its orthogonal complement. By Gaussian regression,
\begin{align}
    \mathbb{E}[Y \mid U_\ell = u] = \frac{u}{\sigma_\ell^2}\,\Sigma x_\ell = u\,v_\ell.
    \label{eq:Gaussian-regression}
\end{align}
Recall the definition of $m_\ell^{(T)}$ from ~\eqref{eqn:mixture-models-means}. Conditioning first on $\{U_\ell\ge T\}$ and then on $U_\ell$, we obtain from~\eqref{eq:Gaussian-regression}--\eqref{eq:trunc-1d-mean-asymptotic}
\begin{align}
    m_\ell^{(T)} &= \mathbb{E}[Y \mid U_\ell \ge T] = \mathbb{E}[\,\mathbb{E}[Y \mid U_\ell]\mid U_\ell \ge T] \nonumber\\
    &= \mathbb{E}[U_\ell \mid U_\ell \ge T]\,v_\ell \nonumber
    \\
    &= \left(T + \frac{\sigma_\ell^2}{T} + O\!\left(\frac{1}{T^3}\right) \right)v_\ell, \qquad T\to\infty.
    \label{eq:m-ell-T-expansion}
\end{align}
In particular, there exist constants $T_1>0$ and $C_1>0$ such that for all $T\ge T_1$,
\begin{align}
    \big\|m_\ell^{(T)} - T v_\ell\big\|_2 \le \frac{C_1}{T}.
    \label{eq:trunc-mean-1-over-T}
\end{align}

By definition, the KL--minimizer $\{\mu_k^\star(T)\}$ in~\eqref{eq:def-pseudotrue-means} also maximizes the population log--likelihood $
\mathcal{H}_T(\{\mu_k\}) \triangleq
\mathbb{E}_{Z\sim g_T}\big[\log f(Z;\{\mu_k\})\big]$. The gradient of $\mathcal{H}_T$ with respect to $\mu_k$ is
\begin{align}
    \nabla_{\mu_k}\mathcal{H}_T(\{\mu_j\}) = \mathbb{E}_{Z\sim g_T}\Big[\gamma_k(Z;\{\mu_j\})\,\Sigma_{\text{GMM}}^{-1}(Z-\mu_k)\Big],
    \label{eq:grad-pop-loglik}
\end{align}
where
\begin{align}
    \gamma_k(z;\{\mu_j\}) \triangleq    \frac{w_k\,\mathcal{N}(z; \mu_k; \Sigma_{\text{GMM}})} {\sum_{m=0}^{L-1} w_m\,\mathcal{N}(z; \mu_m; \Sigma_{\text{GMM}})},
    \label{eq:def-responsibility}
\end{align}
is the posterior responsibility of component $k$ under the model $f$.
At the true parameter we have the stationarity condition
\begin{align}
\nabla_{\mu_k}\mathcal{H}_T(\{\mu_j^\star(T)\}) = 0,
\qquad k=0,\dots,L-1.
\label{eq:stationary-condition}
\end{align}

Decompose the expectation in~\eqref{eq:grad-pop-loglik} over the true mixture $g_T$:
\begin{align}
    0  &= \mathbb{E}_{Z\sim g_T}\Big[
    \gamma_k(Z;\{\mu_j^\star(T)\})\,\Sigma_{\text{GMM}}^{-1}(Z-\mu_k^\star(T))
    \Big] \nonumber
    \\ &= \sum_{\ell=0}^{L-1} \pi_\ell
    \mathbb{E}_{Z\sim g_\ell^{(T)}}\Big[
    \gamma_k(Z;\{\mu_j^\star(T)\})\,\Sigma_{\text{GMM}}^{-1}(Z-\mu_k^\star(T))\Big].
    \label{eq:grad-decomposition}
\end{align}
For large $T$, the component distributions $\{g_\ell^{(T)}\}$ are concentrated in disjoint high-dimensional caps around $\{m_\ell^{(T)}\}$, and the distance between points of $g_\ell^{(T)}$ and $g_{k}^{(T)}$ for $k \neq \ell$ is of order $T$. Consequently, there exist constants $c,C_2>0$ such that, for all $k\neq\ell$,
\begin{align}
    \mathbb{E}_{Z\sim g_\ell^{(T)}}\big[\gamma_k(Z;\{\mu_j^\star(T)\})\big] \le C_2 e^{-cT^2},
    \label{eq:responsibility-off-diagonal}
\end{align}
while
\begin{align}
    \mathbb{E}_{Z\sim g_\ell^{(T)}}\big[\gamma_\ell(Z;\{\mu_j^\star(T)\})\big] = 1 + O(e^{-cT^2}).
    \label{eq:responsibility-diagonal}
\end{align}
Plugging \eqref{eq:responsibility-off-diagonal}--\eqref{eq:responsibility-diagonal} into \eqref{eq:grad-decomposition} and isolating the dominant term with $\ell=k$ yields
\begin{align}
    \pi_k \mathbb{E}_{Z\sim g_k^{(T)}}\big[\Sigma_{\text{GMM}}^{-1}(Z-\mu_k^\star(T))\big] = O(e^{-cT^2}).
\end{align}
Equivalently, as $\pi_k > 0$ for every $k \in \{0, 1, \ldots L-1 \}$,
\begin{align}
    \mu_k^\star(T) - m_k^{(T)} =
    O(e^{-cT^2}).
    \label{eq:mu-star-minus-m-exp}
\end{align}
Thus there exist constants $T_2>0$ and $C_2'>0$ such that for all $T\ge T_2$,
\begin{align}
    \big\|\mu_k^\star(T) - m_k^{(T)}\big\|_2 \le C_2' e^{-cT^2}.
    \label{eq:mu-star-minus-m-explicit}
\end{align}

Combining~\eqref{eq:trunc-mean-1-over-T} and~\eqref{eq:mu-star-minus-m-explicit} via the triangle inequality, we obtain for $T\ge T_0\triangleq\max\{T_1,T_2\}$,
\begin{align}
    \big\|\mu_\ell^\star(T) - T v_\ell\big\|_2 & \le \big\|\mu_\ell^\star(T) - m_\ell^{(T)}\big\|_2
    + \big\|m_\ell^{(T)} - T v_\ell\big\|_2 \nonumber\\
    &\le C_2' e^{-cT^2} + \frac{C_1}{T}.
    \label{eq:mu-star-bias-two-terms}
\end{align}
Since $e^{-cT^2} \le 1/T$ for all sufficiently large $T$, we may absorb the exponential term into the $1/T$ term: there exists $C_{\mathrm{bias}}>0$ and possibly a larger $T_0$ such that, for all $T\ge T_0$,
\begin{align}
    \big\|\mu_\ell^\star(T) - T v_\ell\big\|_2 \le \frac{C_{\mathrm{bias}}}{T},
\end{align}
for each $\ell$. This completes the proof.

\end{proof}

\subsection{Template matching with local maxima}
\label{subsec:local-maxima}

In practical implementations, template matching and the subsequent extraction are typically restricted to local maxima of the correlation score: rather than selecting every location where the cross-correlation exceeds a threshold $T$, one first constructs a correlation heatmap and then retains only those locations that are \emph{local maxima} in a prescribed neighborhood.

To connect this to our theoretical model, for each template $x_\ell$ and patch center $s\in\mathbb{Z}$, define the template-matching score
\begin{align}
    C_\ell(s) \triangleq \langle x_\ell, y_s\rangle,
    \label{eq:Cl-def-1d}
\end{align}
where $y_s$ is the patch extracted at location $s$ as in~\eqref{eq:ys-def-1d}. The corresponding heatmap is
\begin{align}
    C(s) \triangleq \max_{0\le \ell \le L-1} C_\ell(s),
    \label{eq:C-max-def-1d}
\end{align}
and a patch is retained at $s$ if $C(s)\ge T$ and $s$ is a local maximum of $C$ in a prescribed neighborhood.

To formalize local maxima, fix an integer radius $r\ge 0$ and, for each $s\in\mathbb{Z}$, define
\begin{align}
    N(s) \triangleq \{t\in \mathbb{Z} : |t-s|\le r\}.
    \label{eq:Ns-example}
\end{align}
We say that $s$ is a local maximum of $C$ if $C(s) \ge C(s')$ for all $s'\in N(s)$.
Intuitively, in the high-threshold regime $T\to\infty$, exceedances above $T$ become rare. In this regime, the local-maximum constraint has a negligible effect on the leading-order behavior: conditional on having at least one exceedance in a finite window, it is overwhelmingly likely that there is exactly one high exceedance, and hence at most one high local maximum. Consequently, in the high-threshold regime, the non-local-maximum suppression step does not change the leading-order bias induced by template matching.

\begin{proposition}
\label{prop:isolated-exceedances-1d}
Let $\{\xi_t\}_{t\in\mathbb{Z}}$ be a stationary Gaussian noise process satisfying~\eqref{eq:stationary-cov-1d} and the $\alpha$-mixing condition~\eqref{eq:alpha-mixing-summability-1d}, and let $C(s)$ be the template-matching heatmap defined in~\eqref{eq:Cl-def-1d}-\eqref{eq:C-max-def-1d} for a finite collection of templates $\{x_\ell\}_{\ell=0}^{L-1}$. Fix a finite index window $W\subset\mathbb{Z}$ and define
\begin{align}
    N_T(W) \triangleq \bigl|\{s\in W : C(s)\ge T\}\bigr|
\end{align}
to be the number of exceedances of level $T$ in $W$.
Then, as $T\to\infty$,
\begin{align}
    \mathbb{P}\big(N_T(W)\ge 2 \,\big|\, N_T(W)\ge 1\big) = \frac{\mathbb{P}(N_T(W)\ge 2)}{\mathbb{P}(N_T(W)\ge 1)} \longrightarrow 0.
    \label{eq:ratio-exceedances}
\end{align}
\end{proposition}

\begin{proof}[Proof of Proposition~\ref{prop:isolated-exceedances-1d}]
By stationarity of $\{C(s)\}_{s\in\mathbb Z}$,
\begin{align}
    \lambda_T \triangleq \mathbb{E}[N_T(W)] = \sum_{s\in W} \mathbb{P}\big(C(s)\ge T\big) = |W|\,\mathbb{P}\big(C(0)\ge T\big).
\end{align}
Since $C(0)$ is the maximum of finitely many centered Gaussian random variables $\{C_\ell(0)\}_{\ell=0}^{L-1}$ (see~\eqref{eq:Cl-def-1d}–\eqref{eq:C-max-def-1d}), it has a sub-Gaussian tail: there exist constants $c_1>0$ and $K_1<\infty$ such that
\begin{align}
    \mathbb{P}\big(C(0)\ge T\big) \le K_1 e^{-c_1 T^2}, \qquad T\ge 0,
\end{align}
by a union bound and standard Gaussian tail estimates~\cite{durrett2019probability}. Hence $\lambda_T \to 0$ exponentially fast as $T\to\infty$.

Because $W$ is finite, we can bound the probability of multiple exceedances by a union bound over pairs of indices:
\begin{align}
    \mathbb{P}\big(N_T(W)\ge 2\big) \le \sum_{\substack{s,t\in W\\ s\neq t}} \mathbb{P}\big(C(s)\ge T,\; C(t)\ge T\big).
    \label{eq:pair-union}
\end{align}
Fix $s\neq t$. By definition $C(s) = \max_{0\le \ell\le L-1} C_\ell(s)$ and $C(t)=\max_{0\le k\le L-1} C_k(t)$ with $C_\ell(s) = \langle x_\ell, y_s\rangle$, where each $C_\ell(s)$ is a centered Gaussian linear functional of the underlying field $\{\xi_u\}$. Thus,
\begin{align}
    \mathbb{P}\big(C(s)\ge T,\; C(t)\ge T\big) \le \sum_{\ell,k=0}^{L-1} \mathbb{P}\big(C_\ell(s)\ge T,\; C_k(t)\ge T\big).
    \label{eq:max-to-components}
\end{align}
For each fixed pair $(\ell,k)$ and $s\neq t$, the vector $(C_\ell(s),C_k(t))$ is a non-degenerate bivariate Gaussian. In particular, its correlation coefficient satisfies $|\rho_{\ell,k}(s,t)|<1$ (a consequence of strict positive-definiteness of the covariance in~\eqref{eq:stationary-cov-1d}). Standard bivariate Gaussian tail bounds~\cite{leadbetter2012extremes} then yield
\begin{align}
    \mathbb{P}\big(C_\ell(s)\ge T,\; C_k(t)\ge T\big) \le K_2 e^{-c_2 T^2}, \qquad T\ge 0,
\end{align}
for some constants $K_2<\infty$ and $c_2>c_1$ that are uniform over $(\ell,k)$ and $s\neq t$ in the finite window $W$. Combining this with~\eqref{eq:max-to-components} gives
\begin{align}
    \mathbb{P}\big(C(s)\ge T,\; C(t)\ge T\big) \le L^2 K_2 e^{-c_2 T^2}.
\end{align}
Substituting into~\eqref{eq:pair-union}, we obtain
\begin{align}
    \mathbb{P}\big(N_T(W)\ge 2\big) \le |W|(|W|-1)L^2 K_2 e^{-c_2 T^2} = O\big(e^{-c_2 T^2}\big).
\end{align}

On the other hand, by Markov’s inequality and the definition of $\lambda_T$,
\begin{align}
    \mathbb{P}\big(N_T(W)\ge 1\big) \le \lambda_T = |W|\,\mathbb{P}\big(C(0)\ge T\big) \le |W|K_1 e^{-c_1 T^2},
    \label{eq:Nge1-upper}
\end{align}
while a lower bound is obtained by the union bound with a correction for double counting:
\begin{align}
    \mathbb{P}\big(N_T(W)\ge 1\big) &\ge \sum_{s\in W} \mathbb{P}\big(C(s)\ge T\big) - \sum_{\substack{s,t\in W\\ s\neq t}} \mathbb{P}\big(C(s)\ge T,\; C(t)\ge T\big)
    \nonumber\\
    &= \lambda_T - O\big(e^{-c_2 T^2}\big).
    \label{eq:Nge1-lower}
\end{align}
Since $c_2>c_1$, the error term $e^{-c_2 T^2}$ is $o(e^{-c_1 T^2})$, and hence $O(e^{-c_2 T^2}) = o(\lambda_T)$ as $T\to\infty$. Combining~\eqref{eq:Nge1-upper} and~\eqref{eq:Nge1-lower} gives
\begin{align}
    \mathbb{P}\big(N_T(W)\ge 1\big) \sim \lambda_T.
\end{align}
Therefore, as $T \to \infty$,
\begin{align}
    \mathbb{P}\big(N_T(W)\ge 2 \,\big|\, N_T(W)\ge 1\big) = \frac{\mathbb{P}(N_T(W)\ge 2)}{\mathbb{P}(N_T(W)\ge 1)} = \frac{O(e^{-c_2 T^2})}{\lambda_T}
    \longrightarrow 0,
\end{align}
because $\lambda_T = \Theta(e^{-c_1 T^2})$ with $c_1 < c_2$. This proves~\eqref{eq:ratio-exceedances}.
\end{proof}

\section{3D reconstruction}
\label{sec:3D-reconstruction}

\subsection{Unified reconstruction model for cryo-EM and cryo-ET}
\label{sec:app_recon_unified}

\paragraph{Forward model.}
We describe cryo-EM and cryo-ET within a common framework.
Let $V \in \mathbb{R}^{d_{\mathrm{vol}}}$ denote the underlying 3D volume, and let
$\{R_i\}_{i=0}^{M-1} \subset \mathsf{SO}(3)$ be random rotations.
Each extracted datum $z_i \in \mathbb{R}^d$ (with $d$ the number of pixels or voxels of the observed object) is modeled as
\begin{align}
    z_i = \Pi\bigl(R_i \cdot V\bigr) + \epsilon_i, \qquad i = 0,\dots,M-1,
    \label{eqn:unified-cryo-model}
\end{align}
where $\Pi$ is a linear imaging operator and 
$\{\epsilon_i\}_{i=0}^{M-1} \overset{\text{i.i.d.}}{\sim} \mathcal{N}(0,\Sigma_V)$ are additive Gaussian noise terms with covariance $\Sigma_V \succ 0$.

\begin{itemize}
    \item \emph{Cryo-ET:} $\Pi$ is the identity operator on volumes, so $z_i$ is a 3D subtomogram and $d = d_{\mathrm{vol}}$.
    \item \emph{Cryo-EM:} $\Pi$ is the tomographic projection operator (optionally including CTF and other imaging effects), so $z_i$ is a 2D particle image and $d < d_{\mathrm{vol}}$.
\end{itemize}

In both modalities, a particle-picking stage first detects and extracts $z_i$ from the raw data (micrographs in cryo-EM, tomograms in cryo-ET), and a subsequent reconstruction stage estimates $V$ from the extracted observations $\{z_i\}_{i=0}^{M-1}$.

\begin{remark}
We use $\Sigma_V$ to denote the noise covariance in the \emph{postulated} reconstruction model~\eqref{eqn:unified-cryo-model}, and $\Sigma$ to denote the covariance of the noise in the input to the particle picker, as in Algorithm~\ref{alg:particlePickerTemplateMatching}. These covariance matrices need not coincide in our analysis.
\end{remark}

\paragraph{Template generation and particle picking.}
During particle picking, a bank of templates is generated from a single template volume $V_{\mathrm{template}}$ by applying a finite set of rotations $\{R_\ell\}_{\ell=0}^{L-1} \subset \mathsf{SO}(3)$:
\begin{align}
    x_\ell \triangleq \Pi\bigl(R_\ell \cdot V_{\mathrm{template}}\bigr), \qquad \ell = 0,\dots,L-1.
\end{align}
In cryo-ET, these are 3D templates (since $\Pi = I$); in cryo-EM, they are 2D projection templates.
The templates $\{x_\ell\}$ are used as inputs to Algorithm~\ref{alg:particlePickerTemplateMatching}, with a correlation threshold $T$.
When Algorithm~\ref{alg:particlePickerTemplateMatching} is applied to pure noise whose candidate patches satisfy $y_i \sim \mathcal{N}(0,\Sigma)$ (or more generally Model~\ref{model:isotropicIID}), the law of a picked datum $Z$ is a finite mixture
\begin{align}
    g^{(T)}(z)\triangleq \sum_{\ell=0}^{L-1} \pi_\ell\, g_\ell^{(T)}(z), \qquad z\in\mathbb{R}^d,
    \label{eqn:actualUnifiedStatistics}
\end{align}
where $\pi_\ell \ge 0$, $\sum_{\ell=0}^{L-1}\pi_\ell = 1$ reflect the relative frequencies of picks associated with template $x_\ell$, and $g_\ell^{(T)}$ is the component density corresponding to template $x_\ell = \Pi(R_\ell\cdot V_{\mathrm{template}})$.

In the Gaussian setting $y_i \sim \mathcal{N}(0,\Sigma)$, each component is a truncated Gaussian,
\begin{align}
    g_\ell^{(T)}(z) = C_\ell^{(T)} \exp\!\left(-\frac{1}{2} z^\top \Sigma^{-1} z\right) \mathbbm{1}_{\{\langle z, \Pi(R_\ell \cdot V_{\text{template}})\rangle \ge T\}},     \label{eqn:componentOfMixtureModelParticlePickerUnified}
\end{align}
where $C_\ell^{(T)}$ is the normalizing constant chosen so that $\int_{\mathbb{R}^d} g_\ell^{(T)}(z)\,dz = 1$.
In the spherically symmetric i.i.d.\ noise model (Model~\ref{model:isotropicIID}) with base density $p_{\mathrm{sph}}$, the components take the form
\begin{align}
    g_\ell^{(T)}(z) = C_\ell^{(T)}\, p_{\mathrm{sph}}(z)\, \mathbbm{1}_{\{\langle z, \Pi(R_\ell \cdot V_{\text{template}}) \rangle \ge T\}},       \label{eqn:componentOfMixtureModelParticlePickerUnified-2}
\end{align}
with $C_\ell^{(T)}$ again defined by the normalization $\int_{\mathbb{R}^d} g_\ell^{(T)}(z)\,dz = 1$.

Thus, in both cryo-EM and cryo-ET, the output of the template-matching particle picker with templates $x_\ell = \Pi(R_\ell\cdot V_{\mathrm{template}})$ and threshold $T$ is a stationary sequence of observations whose marginal distribution is the mixture~\eqref{eqn:actualUnifiedStatistics}.

\paragraph{Reconstruction likelihood.}
The goal of the reconstruction step is to estimate $V$ in the postulated model~\eqref{eqn:unified-cryo-model} from the extracted observations $\mathcal{A} = \{z_i\}_{i=0}^{M-1}$.
Let $f(z;V)$ denote the likelihood of a single \emph{postulated} observation:
\begin{align}
    f(z;V) = \frac{1}{(2\pi)^{d/2}\det (\Sigma_V)^{1/2}} \int_{R\in\mathsf{SO}(3)} \exp\!\left(-\frac{(\,z - \Pi(R\cdot V)\,)^\top \Sigma_V^{-1}(\,z - \Pi(R\cdot V)\,)}{2}\right)\,d\rho(R),      \label{eqn:distributionUnifiedAlign}
\end{align}
where $\rho$ is a prior distribution on rotations.

The corresponding log-likelihood of $V$ given the extracted data is
\begin{align}
    \mathcal{L}\big(V;\mathcal{A}\big) &\triangleq \frac{1}{M} \sum_{i=0}^{M-1}\log \pp{\int_{R\in\mathsf{SO}(3)} \exp\!\left( -\frac{(\,z_i - \Pi(R\cdot V)\,)^\top \Sigma_V^{-1}(\,z_i - \Pi(R\cdot V)\,)}{2}\right)\,d\rho(R)} +  \text{const}. \label{eqn:assumedLikeEstUnifiedIntegral}
\end{align}
In practice, the integral over $\mathsf{SO}(3)$ is approximated by a finite set of rotations. For analytical convenience (and to match the picking stage), we consider a discretization $\{R_\ell\}_{\ell=0}^{L-1}\subset\mathsf{SO}(3)$ with associated weights $\{w_\ell\}_{\ell=0}^{L-1}$ satisfying $w_\ell\ge 0$ and $\sum_{\ell=0}^{L-1} w_\ell = 1$. Then, the single-particle density is approximated by
\begin{align}
    f(z;V) = \frac{1}{(2\pi)^{d/2}\det (\Sigma_V)^{1/2}} \sum_{\ell=0}^{L-1} w_\ell\, \exp\!\left(-\frac{(\,z - \Pi(R_\ell \cdot V)\,)^\top \Sigma_V^{-1}(\,z - \Pi(R_\ell \cdot V)\,)}{2} \right),    
    \label{eqn:distributionUnifiedAlignDisc}
\end{align}
and the discretized likelihood becomes
\begin{align}
    \mathcal{L}\big(V;\mathcal{A}\big) \triangleq \frac{1}{M} \sum_{i=0}^{M-1}\log \pp{\sum_{\ell = 0}^{L-1} w_\ell\,\exp\!\left( -\frac{(\,z_i - \Pi(R_\ell\cdot V)\,)^\top \Sigma_V^{-1}(\,z_i - \Pi(R_\ell\cdot V)\,)}{2} \right)}  +  \text{const}.  
    \label{eqn:assumedLikeEstUnified}
\end{align}
The maximum-likelihood estimator of the volume under this discretized model is
\begin{align}
    \widehat{V} \triangleq \argmax_{V} \mathcal{L}(V;\mathcal{A}).    \label{eqn:maximumOfLogLikelihoodUnified}
\end{align}

Note that the \emph{postulated} statistics of the extracted data, $f(z;V)$ in~\eqref{eqn:distributionUnifiedAlignDisc}, differ from the \emph{actual} statistics of the picked observations, $g^{(T)}(z)$ in~\eqref{eqn:actualUnifiedStatistics}. Thus, in both cryo-EM and cryo-ET, the reconstruction model is misspecified, and the high-threshold analysis developed in the GMM setting applies with the unified forward operator $\Pi$ (identity for cryo-ET, tomographic projection for cryo-EM).

\begin{assum}\label{assump:1}
The same finite set of rotations $\{R_\ell\}_{\ell=0}^{L-1}\subset\mathsf{SO}(3)$ is used both to generate the templates in the particle-picking stage,
\begin{align*}
    x_\ell = \Pi\bigl(R_\ell \cdot V_{\mathrm{template}}\bigr),
    \qquad \ell = 0,\dots,L-1,
\end{align*}
and as the discretization of $\mathsf{SO}(3)$ in the reconstruction model~\eqref{eqn:distributionUnifiedAlignDisc}.
\end{assum}

The main question is how the reconstructed volume $\widehat{V}$ in~\eqref{eqn:maximumOfLogLikelihoodUnified} relates to the template volume $V_{\mathrm{template}}$ used for picking, when the underlying candidate observations are drawn from a zero-mean noise model (Model~\ref{model:whiteNoiseIID} or Model~\ref{model:isotropicIID}). At first glance, one might expect $\widehat{V}$ to converge to zero, since the input data have zero mean before selection. However, Corollary~\ref{thm:3DstructureInformal} shows that, as the number of extracted observations $M \to \infty$ and the threshold $T\to\infty$, the reconstructed volume $\widehat{V}$ instead converges (up to scaling) to the template volume $V_{\mathrm{template}}$, rotated by an element of $\mathsf{SO}(3)$.

\subsection{Proof of Corollary~\ref{thm:3DstructureInformal}}\label{sec:proofOfCorollary3Dreconstruction}

Let $\widehat{V}(N,T)$ denote any maximum-likelihood estimator of the volume under the discretized unified model \eqref{eqn:distributionUnifiedAlignDisc}--\eqref{eqn:maximumOfLogLikelihoodUnified}, based on the $M = M(N,T)$ extracted observations produced by Algorithm~\ref{alg:particlePickerTemplateMatching}.
We reduce the statement to the high-threshold alignment of GMM centers (Proposition~\ref{thm:high-threshold-MLE-alignment}).

\paragraph{Step 1: Unified likelihood as a constrained GMM.}
Under the unified model~\eqref{eqn:distributionUnifiedAlignDisc}, a single observation $z$ has density
\begin{align}
    f(z;V) = \frac{1}{(2\pi)^{d/2}\det (\Sigma_V)^{1/2}} \sum_{\ell=0}^{L-1} w_\ell\, \exp\!\left(-\frac{(\,z - \Pi(R_\ell \cdot V)\,)^\top \Sigma_V^{-1}(\,z - \Pi(R_\ell \cdot V)\,)}{2} \right),
\end{align}
which is exactly a Gaussian mixture with component means
\begin{align}
    \mu_\ell(V) \;\triangleq\; \Pi\bigl(R_\ell \cdot V\bigr), \qquad \ell=0,\dots,L-1,
\end{align}
shared covariance $\Sigma_V$, and fixed weights $\{w_\ell\}$.
Thus the reconstruction log-likelihood $\mathcal{L}(V;\mathcal{A})$ in~\eqref{eqn:assumedLikeEstUnified} is the GMM log-likelihood
\begin{align}
    \mathcal{L}(V;\mathcal{A}) = \frac{1}{M} \sum_{i=0}^{M-1} \log f(z_i;V) = \mathcal{L}_{\mathrm{GMM}}\bigl(\{\mu_\ell(V)\};\mathcal{A}\bigr),
\end{align}
evaluated on the special parametrization $\mu_\ell = \Pi(R_\ell \cdot V)$.

For each $(N,T)$, define the induced centers
\begin{align}
    \widehat{\mu}_\ell(N,T) \triangleq \Pi\bigl(R_\ell \cdot \widehat{V}(N,T)\bigr), \qquad \ell=0,\dots,L-1.
\end{align}
By construction, $\{\widehat{\mu}_\ell(N,T)\}$ is a maximizer of the empirical GMM log-likelihood over the restricted parameter set
\begin{align}
    \mathcal{M}_{\mathrm{orbit}} \triangleq \Bigl\{(\mu_0,\dots,\mu_{L-1}) : \exists\, V\in\mathbb{R}^{d_{\mathrm{vol}}} \text{ such that } \mu_\ell = \Pi(R_\ell\cdot V),\; 0\le\ell\le L-1 \Bigr\}.
\end{align}

\paragraph{Step 2: Population limit for fixed $T$.}
By the same misspecified maximum-likelihood estimation argument used in Lemma~\ref{lemma:KL-MLE}, for each fixed $T$, as $N\to\infty$ we have
\begin{align}
    \widehat{V}(N,T) \xrightarrow[N\to\infty]{\mathrm{a.s.}} V^\star(T),
    \label{eqn:KL-convergence-unified}
\end{align}
where $V^\star(T)$ maximizes the population objective
\begin{align}
    \mathcal{H}_T^{\mathrm{unified}}(V) \triangleq \mathbb{E}_{Z\sim g^{(T)}}\bigl[\log f(Z;V)\bigr],
\end{align}
with $g^{(T)}$ the picking-induced mixture~\eqref{eqn:actualUnifiedStatistics}.
Equivalently, the induced centers
\begin{align}
    \mu_\ell^\star(T) \triangleq \Pi\bigl(R_\ell \cdot V^\star(T)\bigr)
\end{align}
form a maximizer of the GMM population objective
\begin{align}
    \mathcal{H}_T(\{\mu_\ell\}) \triangleq \mathbb{E}_{Z\sim g^{(T)}}\bigl[\log f(Z;\{\mu_\ell\})\bigr],
\end{align}
restricted to $\mathcal{M}_{\mathrm{orbit}}$:
\begin{align}
    \{\mu_\ell^\star(T)\} \in \argmax_{\{\mu_\ell\}\in\mathcal{M}_{\mathrm{orbit}}} \mathcal{H}_T(\{\mu_\ell\}).
\end{align}

\paragraph{Step 3: High-threshold structure and alignment with the template orbit.}
By Proposition~\ref{thm:prop1} and Proposition~\ref{prop:truncated-mean-isotropic},
the component means $m_\ell^{(T)}$ of the true mixture $g^{(T)}$ satisfy, under spherical symmetric noise,
\begin{align}
    \frac{m_\ell^{(T)}}{T} \xrightarrow[T\to\infty]{} x_\ell = \Pi\bigl(R_\ell \cdot V_{\mathrm{template}}\bigr), \qquad 0\le\ell\le L-1.
\end{align}
Thus, in the notation of Proposition~\ref{thm:high-threshold-MLE-alignment}, we have
\begin{align}
    v_\ell = x_\ell = \Pi\bigl(R_\ell \cdot V_{\mathrm{template}}\bigr).
\end{align}

Applying Proposition~\ref{thm:high-threshold-MLE-alignment} to the same mixture $g^{(T)}$, we obtain that there exists a permutation $\pi$ of $\{0,\dots,L-1\}$ such that any unconstrained population maximizer $\{\tilde{\mu}_\ell^\star(T)\}$ of $\mathcal{H}_T$ satisfies
\begin{align}
    \frac{\tilde{\mu}_{\pi(\ell)}^\star(T)}{T} \xrightarrow[T\to\infty]{} v_\ell = \Pi\bigl(R_\ell \cdot V_{\mathrm{template}}\bigr), \qquad 0\le\ell\le L-1.
    \label{eq:unconstrained-centers-limit-unified}
\end{align}

The high-threshold limiting configuration
\begin{align}
   \lim_{T\to\infty} \frac{m_\ell^{(T)}}{T} = \Pi\bigl(R_\ell \cdot V_{\mathrm{template}}\bigr)
\end{align}
lies in the restricted manifold $\mathcal{M}_{\mathrm{orbit}}$, since
\begin{align}
    \bigl(\Pi(R_0\cdot V_{\mathrm{template}}), \dots, \Pi(R_{L-1}\cdot V_{\mathrm{template}}) \bigr) = \{\mu_\ell(V_{\mathrm{template}})\}.
\end{align}
Therefore, the constrained maximizers $\{\mu_\ell^\star(T)\}$
cannot do better asymptotically (at scale $T^2$) than the unconstrained maximizers
$\{\tilde{\mu}_\ell^\star(T)\}$, and the same high-threshold argument as in
Step~2 of the proof of Proposition~\ref{thm:high-threshold-MLE-alignment}
implies that there exists a rotation $R\in\mathsf{SO}(3)$ such that
\begin{align}
    \frac{\Pi(R_\ell \cdot V^\star(T))}{T} = \frac{\mu_\ell^\star(T)}{T} \xrightarrow[T\to\infty]{} \Pi\bigl(R_\ell \cdot (R \cdot V_{\mathrm{template}})\bigr),
    \qquad 0\le\ell\le L-1.
\end{align}
Under Assumption~\ref{assump:Pi-identifiability}, the map
$V \mapsto \{\Pi(R_\ell\cdot V)\}_{\ell=0}^{L-1}$ identifies $V$ up to a global rotation. Hence the above convergence is only possible if
\begin{align}
    \frac{V^\star(T)}{T} \xrightarrow[T\to\infty]{} R \cdot V_{\mathrm{template}},
    \label{eq:population-volume-limit-unified}
\end{align}
for some $R\in\mathsf{SO}(3)$.

\paragraph{Step 4: Combining the limits.}
Fix $T$. From~\eqref{eqn:KL-convergence-unified} we have
\begin{align}
    \widehat{V}(N,T) \xrightarrow[N\to\infty]{\mathrm{a.s.}} V^\star(T),
\end{align}
and hence
\begin{align}
    \frac{\widehat{V}(N,T)}{T} \xrightarrow[N\to\infty]{\mathrm{a.s.}} \frac{V^\star(T)}{T}.
\end{align}
Combining this with the population limit~\eqref{eq:population-volume-limit-unified} and taking the
iterated limits $N\to\infty$, then $T\to\infty$, yields
\begin{align}
    \lim_{T\to\infty}\lim_{N\to\infty} \frac{\widehat{V}(N,T)}{T} =  R \cdot V_{\mathrm{template}},
\end{align}
with convergence in probability, which is exactly~\eqref{eq:cryoET-white-limit}.

\section{Experimental methods}
\label{sec:methods}

\subsection{Setup of the numerical experiments} 
In the theoretical model presented in Section~\ref{sec:model}, candidate particles $\{y_i\}_{i=0}^{N-1}$ were assumed to be drawn i.i.d. from an underlying distribution. However, in practical particle picking from micrographs, candidates are not sampled i.i.d.; instead, they are extracted using a sliding-window approach that introduces spatial overlap between neighboring regions and induces statistical dependencies.
Furthermore, when selecting particles, it is desirable to avoid redundancy by ensuring that the selected patches do not overlap. To achieve this, the particle detection procedure must enforce a non-overlapping constraint: only one candidate patch should be selected from any set of overlapping regions. The algorithm used to implement this constraint is described in detail in Algorithm~\ref{alg:particlePickerTemplateMatchingWithOverlapping}.

This algorithm identifies candidate particles in a collection of 2D micrographs $\{Y_i\}_{i=0}^{R-1} \subset \mathbb{R}^{k \times k}$ by computing their cross-correlation with a set of template images $\{x_\ell\}_{\ell=0}^{L-1} \subset \mathbb{R}^{d \times d}$. For each micrograph-template pair, a 2D circular cross-correlation (assuming periodic boundary conditions) is computed to produce a heatmap $C_{\ell,i}$, where each pixel indicates the similarity between a local patch of the micrograph and the template. Each $C_{\ell,i}$ has the same spatial resolution as the micrograph.
To consolidate the response across templates, the algorithm computes a combined heat map $C_i$ by taking the maximum pixel over all $L$ templates. This results in a score map where each entry reflects the strongest template match at that location. The pixels in $C_i$ are then flattened and sorted in descending order by score.
The algorithm then iterates through the ranked pixels and selects particle candidates that (i) exceed a predefined threshold $T$, and (ii) do not spatially overlap with any previously selected patch. A binary mask is maintained to track used regions and enforce non-overlap. Only the highest-scoring, non-overlapping patches are added to the output set $\mathcal{A}$.
This non-overlapping template matching strategy is commonly used in cryo-EM; see, for example,~\cite{eldar2024object}.

\begin{algorithm}[t!]
\caption{\texttt{Template-Matching Particle Picker with Overlap}}
\label{alg:particlePickerTemplateMatchingWithOverlapping}
\textbf{Input:} Templates $\{x_\ell\}_{\ell=0}^{L-1} \subset \mathbb{R}^{d \times d}$,
         Micrographs $\{Y_i\}_{i=0}^{R-1} \subset \mathbb{R}^{k \times k}$,
         Threshold $T \in \mathbb{R}$ \\
\textbf{Output:} Set $\mathcal{A}$ of extracted particle patches
\begin{algorithmic}[1]
\State $\mathcal{A} \gets \emptyset$
\For{$i = 0$ to $R - 1$} \Comment{Loop over micrographs}
    \State $C_i \gets \text{zero matrix in } \mathbb{R}^{k \times k}$ \Comment{Heatmap storing the maximum cross-correlation}
    \For{$\ell = 0$ to $L - 1$} \Comment{Loop over templates}
        \State $C_{\ell,i} \gets x_\ell \star Y_i$ \Comment{2D circular cross-correlation}
        \ForAll{pixels $(u,v)$}
            \State $C_i[u,v] \gets \max(C_i[u,v], C_{\ell,i}[u,v])$ \Comment{Maximum of the heatmap}
        \EndFor
    \EndFor
    \State $(S_i, I_i) \gets \textsc{SortDescend}(\textsc{Flatten}(C_i))$ \Comment{Sort the heatmap}
    \State $M_i \gets \text{zero mask in } \{0,1\}^{k \times k}$ \Comment{Binary map tracking already selected regions}
    \For{$m = 0$ to $k^2 - 1$} \Comment{Loop over pixels}
        \State $(u, v) \gets \textsc{IndexToCoords}(I_i[m], k)$ \Comment{Convert flat index to 2D coordinates}
        \If{$S_i[m] > T$ \textbf{and} patch around $(u,v)$ of size $d \times d$ does not overlap in $M_i$} 
            \State $P \gets Y_i[u - d/2 : u + d/2,\, v - d/2 : v + d/2]$
            \State $\mathcal{A} \gets \mathcal{A} \cup \{P\}$
            \State Fill with ones the region of size $d \times d$ around $(u,v)$ in $M_i$
        \EndIf
    \EndFor
\EndFor
\State \Return $\mathcal{A}$
\end{algorithmic}
\end{algorithm}

\subsection{Validation under stationary correlated Gaussian noise}
\label{subsec:colored_noise_validation}

To complement the white-noise experiments in the main text, we carried out an additional supplementary simulation study under stationary correlated Gaussian noise, with the goal of empirically validating the prediction of Theorem~\ref{thm:classesCentersVersusTemplatesInformalStationary}. This experiment should be viewed as the correlated-noise analogue of the white-noise validation presented in the main text: there, the selected pure-noise samples yield class centers that align with the templates themselves, whereas here the theory predicts that the selected class centers should align with a covariance-shaped transform of the templates.

\begin{figure}[t!]
    \centering
    \includegraphics[width=0.9\linewidth]{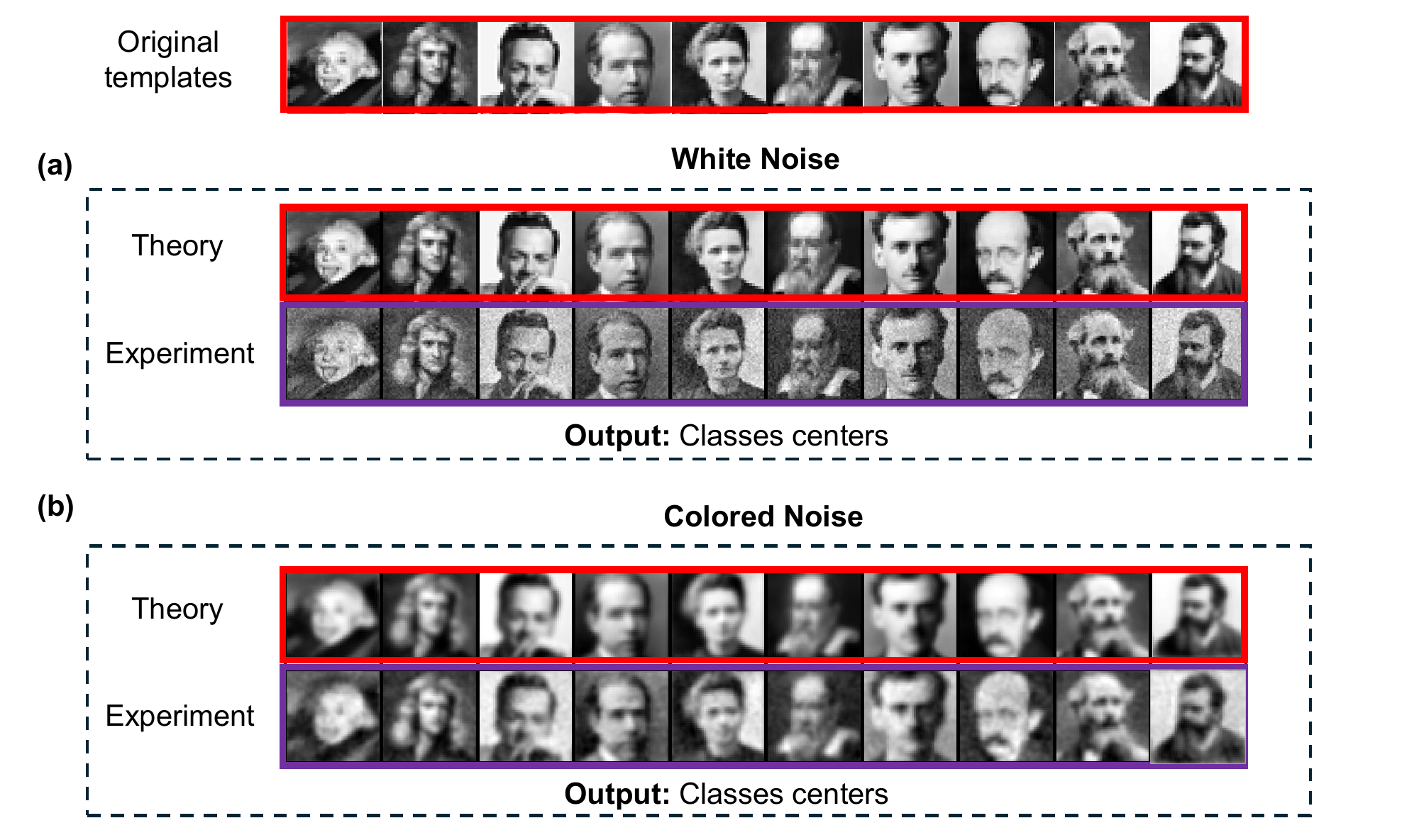}
    \caption{\textbf{Validation of template-induced bias under white and stationary correlated Gaussian noise.}
    Top row: the original templates. Panel (a) shows the white-noise setting. In this case, the theoretical prediction for the normalized class centers coincides with the templates themselves, and the empirical class centers obtained from the selected noise samples closely match this prediction. Panel (b) shows the stationary correlated-noise setting. Here, the theory predicts that the normalized class centers are no longer equal to the original templates, but rather to their covariance-shaped transforms under the noise covariance operator, given by~\eqref{eqn:colored_noise_align}. The empirical class centers again closely agree with the theoretical prediction.}
    \label{fig:colored_noise}
\end{figure}

In this experiment, the input patches are drawn from a zero-mean stationary Gaussian field with a prescribed non-flat two-dimensional power spectrum. Specifically, if $W$ denotes white Gaussian noise and $S(k_x,k_y)$ is a nonnegative power spectrum, then each realization is generated as
\begin{align}
    Y = \mathcal{F}^{-1}\!\big(\sqrt{S}\,\mathcal{F}(W)\big),
\end{align}
where $\mathcal{F}$ and $\mathcal{F}^{-1}$ denote the two-dimensional discrete Fourier transform and its inverse, respectively. In the colored-noise experiment, we choose $S$ explicitly as a Gaussian low-pass spectrum,
\begin{align}
    S(k_x,k_y)
    =
    \exp\!\left(
        -\frac{k_x^2+k_y^2}{2\ell^2}
    \right),
\end{align}
where $\ell>0$ is a bandwidth parameter and $(k_x,k_y)$ range over the discrete Fourier frequencies. We then normalize $S$ by its average value so that the resulting spatial variance is approximately one. Equivalently, the covariance operator of $Y$ is a stationary convolution operator $\Sigma$, which acts by pointwise multiplication by $S$ in the Fourier domain. In contrast to the white-noise case, $\Sigma$ is not proportional to the identity, and therefore the induced template-matching bias depends explicitly on the correlation structure of the noise.

We then apply the same template-matching selection rule as in the main text to a large collection of pure-noise samples, using a bank of normalized templates $\{x_\ell\}_{\ell=1}^L$. For each selected sample, we record the template attaining the largest correlation score, and for each template class $\ell$ we compute the empirical conditional mean of the selected samples. Theorem~\ref{thm:classesCentersVersusTemplatesInformalStationary} predicts that, after normalization by the selection threshold, this empirical class mean should align with
\begin{align}
    \label{eqn:colored_noise_align}
    \frac{\Sigma x_\ell}{x_\ell^\top \Sigma x_\ell}.
\end{align}
In the white-noise case, where $\Sigma \propto I$, this expression reduces to the template itself. In the correlated-noise setting, however, it becomes a covariance-shaped transform of the template.

The results, shown in Figure~\ref{fig:colored_noise}, confirm this prediction. Panel~(a) reproduces the white-noise benchmark: the theoretical and empirical normalized class centers both closely match the original templates. Panel~(b) shows the stationary correlated-noise setting: although the theoretical outputs are now visibly distorted relative to the original templates, the empirical class centers remain in close agreement with the corresponding theoretical predictions. Quantitatively, this agreement is also reflected in a PCC close to one between the empirical and theoretical class centers. Thus, as in the white-noise case studied in the main text, template-based selection from pure noise still produces a structured biased output; the main difference is that under correlated noise the limiting bias is shaped not only by the templates, but also by the covariance operator $\Sigma$.

\end{appendices}

\end{document}